\documentclass[journal, doublecolumn]{IEEEtran}
\usepackage{mathrsfs}
\usepackage{bbding}
\usepackage{amsfonts}
\usepackage{amsmath}
\usepackage{graphicx}
\usepackage{subfigure}
\usepackage{latexsym}
\usepackage{amssymb}
\usepackage{amsmath}
\usepackage{enumerate}
\usepackage{color}
\usepackage{qtree}
\usepackage{mdframed}
\usepackage{amsthm}

\newtheorem{Assumption}{Assumption}

\newtheorem{Algorithm}{Algorithm} 

\newtheorem{Definition}{Definition}
\newtheorem{Lemma}{Lemma}
\newtheorem{Problem}{Problem}

\newtheorem{Remark}{Remark}
\newtheorem{Theorem}{Theorem}

\newcommand{\bs}{\begin{small}}
\newcommand{\bsc}{\end{small}}
\newcommand{\txtblue}{\textcolor{black}}

\newcommand{\bgbox}{\vspace{0.3cm} \begin{mdframed}}
\newcommand{\edbox}{\end{mdframed} \vspace{0.2cm}}

\allowdisplaybreaks[4]

\begin{document}
\title{Networked Control Systems over  Correlated Wireless Fading Channels}

\author
{\normalsize{\IEEEauthorblockN{Fan Zhang, \emph{StMIEEE}, Vincent K. N. Lau, \emph{FIEEE}, Ling Shi, \emph{MIEEE}} \\ Department of ECE,  Hong Kong University of Science and Technology, Hong Kong \\ Email: \{fzhangee, eeknlau, eesling\}@ust.hk}\vspace{-0.9cm}} 
\maketitle

\begin{abstract}
In this paper, we consider a  networked control system (NCS) in which an   dynamic plant system  is connected to a controller via a  temporally correlated wireless fading channel. We focus on communication power design at the sensor to minimize a weighted average state estimation error at the remote controller  subject to an average transmit power constraint of the sensor. The  power control optimization problem is formulated as  an infinite horizon average cost Markov decision process (MDP).  We propose a novel continuous-time perturbation approach and derive an asymptotically optimal closed-form value function for the MDP. Under this approximation, we propose a low complexity dynamic power control solution which has an event-driven control structure. We also  establish technical conditions for asymptotic optimality, and sufficient conditions for NCS stability under  the proposed scheme. 
\end{abstract}

\vspace{-0.3cm}

\section{Introduction}
Networked control systems (NCSs) have drawn great attention in recent years due to the  growing applications in industrial  automation, remote robotic control, etc. \cite{coup1}.   A typical NCS consists of a \emph{plant}, a \emph{sensor} and a \emph{controller}  which are connected over a communication network  as illustrated in Fig. \ref{artNSC}. The presence of  communication channels in the NCSs complicates the design and analysis due to the interactions between communication and control. Conventional closed-loop control theories \cite{optimalc1}  (e.g., stabilization, optimal control) must be reevaluated when considering communication constraints. 

There are various  works on the analysis and optimization of  NCSs. In \cite{rate1}, \cite{rate2},  a necessary minimum rate  requirement for  NCS stability is computed for noiseless and memoryless Gaussian   channels between the sensor and the controller. \txtblue{Many other works \cite{ratenoise}--\cite{qvidoer1}  consider encoder/decoder design and give sufficient conditions for  NCS stability  under scenario-specific communication channels. The authors in \cite{ratenoise} design an encoder and decoder structure to achieve  asymptotic stability for an NCS with a packet dropout channel. In \cite{qoiuli}, the authors study  multi-input networked stabilization with a fading channel between the controller and plant. In \cite{qvidoer1}, the authors give stability conditions for an NCS with fading packet dropout channels, where  the evolution of the fading channel follows a Markov process  with a finite discrete state space. In all these works \cite{rate1}--\cite{qvidoer1}, the key focus is on achieving  NCS stability, which is only a weak form of control performance.}  There  are also many works  considering optimal control of  NCSs. In   \cite{ctrlplus2},  the authors consider a static joint  communication and  control optimization and  ignore the stochastic evolutions  of the plant dynamics  and the fading channel states.  In \cite{onoff2}, a joint scheduling and control policy is proposed to minimize the   linear quadratic Gaussian (LQG) cost and the communication cost under a simplified static rate-limited error-free  channel. In \cite{bernuli1}, \cite{poweronl}, the authors consider either plant LQG  control  or sensor scheduling over a packet-dropping network with a constant symbol error rate (SER). \txtblue{In all these works \cite{onoff2}--\cite{poweronl}, the channel  between the sensor and the controller is assumed to be either error-free or with a constant SER and they ignore the effect of how the power control scheme affects the SER, which further affects the state estimation at the remote controller.}

 To optimize the performance of the NCSs over wireless fading channels, the control policy should be adaptive to the \emph{plant state information} and the \emph{fading channel  state}. The plant state realization reveals the \emph{relative importance of the individual state feedback}, while the channel fading state reveals the \emph{transmission opportunities} over the communication channels. In fact, the associated optimization problem belongs to the  Markov decision process (MDP) problem, which is well-known to be quite challenging \cite{mdpcref2}, \cite{mdpsurvey}. \txtblue{In \cite{logicc}, the  authors study dynamic control of the transmission probability by minimizing the mean square error (MSE) of the plant state estimation  and the average transmission probability for an NCS with an on-off switch channel. In \cite{qvidoer2}--\cite{mdp1}, the authors study the dynamic power control for an NCS with a wireless fading channel. Specifically, \cite{qvidoer2} solves the minimization of an average power cost subject to the  stability requirement, \cite{qvidoer3} solves the minimization of the MSE of the plant state estimation  and average power cost, and \cite{mdp1} solves the minimization of the  LQG cost and average power cost.  The MDP problems in  \cite{logicc}--\cite{mdp1} are solved using the conventional  value iteration algorithm (VIA), which induces huge complexity and suffers from slow convergence and lack of insights   \cite{mdpcref2}, \cite{mdpsurvey}. The approaches therein cannot be used in our problem, where we target to obtain a low complexity dynamic control solution.}

In this paper, we consider an NCS where a sensor delivers the  plant state information to a  controller over a temporally correlated   wireless fading channel as illustrated in Fig. \ref{artNSC}. Furthermore, we consider  \emph{error-adaptive power control}\footnote{\txtblue{In error-adaptive control, the data rate at the sensor is fixed, and the sensor dynamically adjusts the transmit power to adjust the SER \cite{digitalcom}  so as to  achieve certain objectives of the NCS.}}, where the instantaneous SER depends on the transmit power of the sensor.  Using the separation principle between  control and communication \cite{rate2}, \cite{dualeffect}, we  focus on minimizing the state estimation  cost of the plant  subject to an average communication power constraint. The communication power optimization problem is formulated as  an infinite horizon average cost  MDP and there are several first order technical challenges:
\begin{itemize}
	\item \textbf{Challenges due to  State-Dependent Dynamic Control:} To derive a state-dependent communication power control policy, we need to solve the associated MDP. The optimality equation for solving the MDP has high dimensionality and brute-force VIA cannot give viable solutions in practice  \cite{mdpcref2}, \cite{mdpsurvey}.
	\item \textbf{Challenges due to the Temporal Correlations of the Fading Channel:}  When the wireless fading channel is temporally i.i.d., the dimension of the optimality equation can be reduced \cite{mdpsurvey}, which simplifies the numerical computation of the value function. However, such a dimension reduction technique cannot be applied when the wireless channel fading is temporally correlated. This poses great challenges even for obtaining the numerical solution of the associated optimality equation. 
\end{itemize}
 To address the above challenges,  we   propose a novel continuous-time perturbation approach and obtain an asymptotically optimal closed-form value function for solving the associated optimality condition of the MDP. Based on the structural properties of the communication power control, we show that the solution has an  \emph{event-driven control} structure.  Specifically, the sensor  either transmits with maximum power or shuts down depending on a dynamic threshold rule.  Furthermore, we analyze the asymptotic optimality of the proposed scheme and also give sufficient conditions for ensuring the NCS stability while using the proposed scheme.  Finally, we compare the  proposed scheme with various state-of-the-art baselines and show that significant performance gains can be achieved with low complexity.

\txtblue{\emph{Notations:} Bold font is used to denote matrices and vectors. $\mathbf{A}^T$ and $\mathbf{A}^\dagger$ denote the transpose and conjugate transpose of matrix $\mathbf{A}$ respectively.    $\mathbf{I}$ represents  identity matrix with appropriate dimension.  $\mu_{max}(\mathbf{A})$   represents the largest   eigenvalue of  symmetric matrix $\mathbf{A}$. $\|\mathbf{A}\|$ represents the Euclidean norm of a vector $\mathbf{A}$. $\text{Re}\{ \mathbf{x} \}$ represents the real part of $\mathbf{x}$. $|x|$ represents the absolute value of a scaler $x$. $\nabla_\mathbf{x} f(\mathbf{x})$ denotes the column gradient vector with the $k$-th element being  $\frac{\partial f(\mathbf{x})}{\partial x_k}$. $\nabla_\mathbf{x}^2 f(\mathbf{x})$ denotes the Hessian matrix of $f(\mathbf{x})$ w.r.t. vector $\mathbf{x}$. $f\left(x\right)=\mathcal{O}\left(g\left(x \right) \right)$ as $x\rightarrow a$ means $\lim_{a \rightarrow a}\frac{f(x)}{g(x)}<\infty$.  }

\begin{figure}[t]
\centering
  \hspace{1cm}
  \includegraphics[width=3.3in]{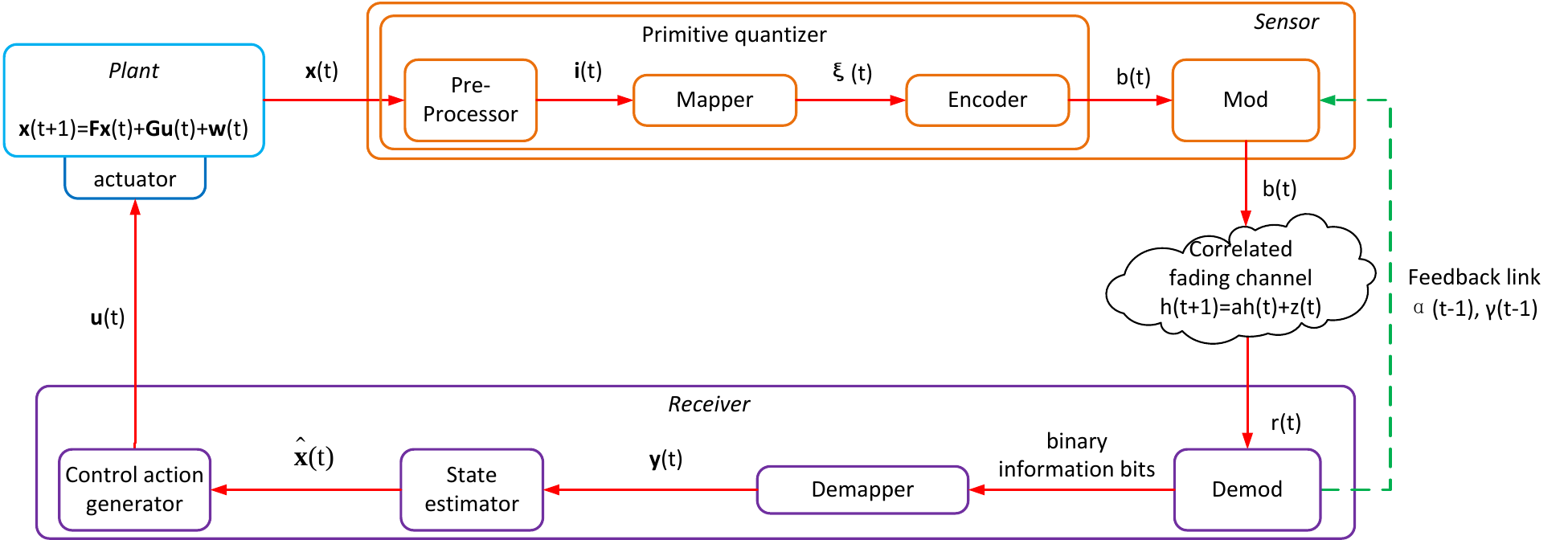}
  \caption{A typical architecture of a networked control system.}
  \label{artNSC}\vspace{-0.6cm}
\end{figure}

\vspace{-0.2cm}

\section{System Model}	\label{ncssystem}

Fig. \ref{artNSC} shows a typical networked control system (NCS), which consists of a plant, a sensor, and a controller, and they form a closed-loop control. We consider a slotted system, where the time dimension is partitioned into decision slots indexed by $t$ with  duration $\tau$.  The  sensor has perfect state observation of the plant state $\mathbf{x}(t)$.  The  controller is  geographically separated from the sensor, and there is  a temporally  correlated  wireless fading channel connecting them. \txtblue{At each time slot $t$,  the sensor  observes  the plant state and the pre-processor generates  $\mathbf{i}(t)$, which is passed to the mapper. The mapper maps $\mathbf{i}(t)$  to one of the $M=2^{R}$ quantization levels $\boldsymbol{\xi}(t)$, which is encodes  into  $R$-digit binary information  bits at the encoder.} The information bits are communicated to the remote controller over the wireless fading channel using $M$ quadrature amplitude modulation (MQAM). Specifically, the $R$-digit binary information  bits are mapped to one of the $M$ available QAM symbols for a given constellation type.  After the demodulation process, the demodulator outputs  binary information  bits to the demapper, which maps the information bits back to one of the $M$ quantization levels $\mathbf{y}(t)$. Then, $\mathbf{y}(t)$  is  passed to the state estimator to obtain a state estimate $\hat{\mathbf{x}}(t)$. After that, $\hat{\mathbf{x}}(t)$ is passed to the actuator to generate a control action $\mathbf{u}(t)$. \txtblue{The actuator, which is co-located with the plant, uses control action $\mathbf{u}(t)$ for plant actuation.}

Such an NCS with a wireless fading channel covers a lot of practical applications\footnote{Please refer to \cite{coup1} and the reference therein for a  more broad range of  application scenarios.}. For example, in a heating, ventilation and air conditioning (HVAC) system \cite{airref}, the boiler and chiller components (i.e., the controller) of the HVAC system are mounted on the rooftop. The temperature and humidity sensors are located inside a room, which  measure certain indoor environmental parameters and transmit the data to  the controller over a wireless channel. The controller then decides  whether to pump cool or hot air into the room through the  ducts based on the received data.

\subsection{Linear Stochastic Plant Model}
We consider a continuous-time stochastic  plant system with dynamics \txtblue{$\frac{d{\mathbf{x}}(t)}{dt}=\widetilde{\mathbf{F}}\mathbf{x}(t)+\widetilde{\mathbf{G}}\mathbf{u}(t)+\widetilde{\mathbf{w}}(t)$, $t\geq0$,  $\mathbf{x}(0)=\mathbf{x}_0$}, where $\mathbf{x}(t) \in \mathbb{R}^d$ is the plant state, $\mathbf{u}(t) \in \mathbb{R}^m$ is the plant control action, $\mathbf{x}_0$ is the initial plant state,  $\widetilde{\mathbf{F}}\in \mathbb{R}^{d\times d}$, $\widetilde{\mathbf{G}} \in \mathbb{R}^{d\times m}$, and $\widetilde{\mathbf{w}}(t) \in \mathbb{R}^d$ is an additive plant disturbance with  zero mean and covariance $\widetilde{\mathbf{W}} \in \mathbb{R}^{d\times d}$. \txtblue{Furthermore, we assume that the plant disturbance is bounded, i.e., $\|\widetilde{\mathbf{w}}(t)\|\leq \widetilde{w}_{max}$ for some $\widetilde{w}_{max}>0$ \cite{boundedDT1}, \cite{boundedDT2}.} Without loss of generality, we assume that $\widetilde{\mathbf{W}}$ is diagonal\footnote{For non-diagonal  $\widetilde{\mathbf{W}}$, we can pre-process the plant  using the \emph{whitening transformation} procedure \cite{white}. \txtblue{Specifically, let the  eigenvalue decomposition  of $\widetilde{\mathbf{W}}$ be $\mathbf{M}\widetilde{\mathbf{W}}\mathbf{M}^T=\mathbf{T}$, where $\mathbf{M}$ is a unitary matrix and $\mathbf{T}$ is diagonal. We have $\mathbf{x}_\mathbf{M}(t+1)=\mathbf{F}_\mathbf{M}\mathbf{x}_\mathbf{M}(t)+\mathbf{G}_\mathbf{M}\mathbf{u}(t)+\mathbf{w}_\mathbf{M}(t)$, where $\mathbf{x}_\mathbf{M}=\mathbf{M}\mathbf{x}$, $\mathbf{F}_\mathbf{M}=\mathbf{M}\mathbf{F}\mathbf{M}^T$, $\mathbf{G}_\mathbf{M}=\mathbf{M}\mathbf{G}$,  $\mathbf{w}_\mathbf{M}=\mathbf{M}\mathbf{w}$ and  $\mathbf{E}\big[\mathbf{w}_\mathbf{M}\mathbf{w}_\mathbf{M}^T\big]=\mathbf{T}$. Therefore, the  optimization  based on $\mathbf{x}$ can be  transformed to  an equivalent problem based on $\mathbf{x}_\mathbf{M}$ with diagonal plant noise covariance.}}. Since the sensor in the NCS samples the plant state once per time slot (with duration  $\tau$), the state dynamics of the sampled discrete-time stochastic  plant system  is given by  \cite{sampleplant}
\begin{align}	\label{plantain}
	\mathbf{x}(t+1)=\mathbf{F}\mathbf{x}(t)+\mathbf{G}\mathbf{u}(t)+\mathbf{w}(t), \  \mathbf{x}(0)=\mathbf{x}_0
\end{align}
for $t\geq 0$, where $\mathbf{F}=\exp( \widetilde{\mathbf{F}}\tau )$, $\mathbf{G}= \widetilde{\mathbf{F}}^{-1}\big(\exp( \widetilde{\mathbf{F}}\tau )-\mathbf{I}\big)\widetilde{\mathbf{G}}$, and $\mathbf{w}(t) = \int_0^\tau	\exp( \widetilde{\mathbf{F}}s ) \widetilde{\mathbf{w}} ((t+1)\tau-s)d s $ is an i.i.d.   random noise with zero mean and covariance $\mathbf{W} = \int_0^\tau \exp( \widetilde{\mathbf{F}}s ) \widetilde{\mathbf{W}}   \exp( \widetilde{\mathbf{F}}s )d s$.  We have the following assumptions on the plant model:
\begin{Assumption}[Stochastic Plant Model]	\label{assumpstabl}
	We assume that the  plant system $\big(\mathbf{F}, \mathbf{G}\big)$ in (\ref{plantain}) is controllable \cite{optimalc1}. \txtblue{Furthermore, we assume that the initial condition in (\ref{plantain}) is bounded, i.e.,  $\|\mathbf{x}_0\|\leq L$  and $L=\mathcal{O}(\tau)$.}~\hfill\IEEEQED
\end{Assumption}

\subsection{Wireless Fading Channel Model}
We consider a continuous-time temporally correlated wireless fading  channel  with dynamics \txtblue{$\frac{dh(t)}{dt}=-\widetilde{a} h(t)+\sqrt{2\widetilde{a}}\widetilde{z}(t)$},  $t\geq0$, with $h(0)=h_0$, where $h(t) \in \mathbb{C}$ is the channel state information (CSI) and $h_0$ is the initial channel state.  The coefficient $\widetilde{a}\in \mathbb{R}^+$  determines the temporal correlation of the fading process\footnote{The autocorrelation function is  $\mathbb{E}\left[h(t+\tau)h(t)\right] =2\exp\left(-\widetilde{a}\tau\right)$ \cite{stochasticgame}.} and \txtblue{$\widetilde{z}(t)\sim \mathcal{CN}(0,1)$ is an additive circularly-symmetric Gaussian  noise with zero mean and unit variance}\footnote{\txtblue{Specifically, $\mathbb{E}[\widetilde{z}]=0$ and $\mathbb{E}[\widetilde{z}\widetilde{z}^\ddagger]=1$, where  $\widetilde{z}^\ddagger$ is the complex conjugate of $\widetilde{z}$.}}.  Similarly, the state dynamics of the sampled discrete-time channel is given by \cite{sampleplant}
\begin{align}	\label{chnain}
	h(t+1)=a h(t)+z(t), \quad h(0)=h_0
\end{align}
for $t\geq 0$, where $a=\exp( -\widetilde{a}\tau) $ and \txtblue{$z(t) = \sqrt{2\widetilde{a}}\int_0^\tau\exp( -\widetilde{a}\tau ) \widetilde{z}((t+1)\tau-s)d s $ is an i.i.d.    noise  with zero mean and covariance $Z=1-\exp( -2\widetilde{a}\tau)$}.  The received signal at the demodulator of the controller is given by
\bs\begin{align}
	\mathbf{r}(t) = \sqrt{p(t)} h(t) \mathbf{b}(t) + \mathbf{n}(t)
\end{align}\bsc where  $p(t)\in \mathbb{R}^+$ is the  transmit SNR and $\mathbf{n}(t)\sim \mathcal{CN}(0,\mathbf{I})$ is an i.i.d. Gaussian noise.  Let $\gamma(t) \in \{0,1\}$ denote the symbol error event (where $\gamma(t) = 0$ means  symbol error). \txtblue{In this paper, we consider rectangular MQAM constellation (e.g., \cite{digitalcom}, \cite{largesnr}, \cite{largesnr1}), and the associated  symbol error rate (SER) is given by  
\bs\begin{align}
	\Pr\left[\gamma(t)=0 \right]  \approx  \exp\Big( -\frac{ p(t) \tau \big|h(t)\big|^2}{\kappa(R) B_W} \Big)	\label{serequ}
\end{align}\bsc where $B_W$ is the channel bandwidth and $\kappa(R)=\frac{2^{R+1}-2}{3}$ is a constant\footnote{\txtblue{The SER model in (\ref{serequ}) covers other  types of constellation geometry for MQAM (e.g., circular  constellation)  with appropriate adjustment of   $\kappa(R)$. In this paper, our derived results are based on the rectangular MQAM, which can be easily extended to other constellation types.}}.} The received signal is processed in the demodulator, which outputs information bits  to the reconstructor. The reconstructor maps the information bits back to one of the $M$ quantization levels and the associated output $\mathbf{y}(t)$ is given by
\begin{align}
	\mathbf{y}(t) = \gamma(t)\boldsymbol{\xi}(t)	\label{equivalentchn}
\end{align} 
\txtblue{Note that  $\boldsymbol{\xi}(t)$ is the  the output of the mapper at the sensor.} Furthermore, (\ref{equivalentchn}) means that $\mathbf{y}(t)=\boldsymbol{\xi}(t)$ for successful transmission  and $\mathbf{y}(t)=\mathbf{0}$ otherwise.

\subsection{Information Structures at the Sensor and the Controller}	\label{informs}
Let $a_0^t\triangleq \left\{a(s): 0\leq s\leq t\right\}$ be the history of  realizations of variable $a$ up to time $t$. The available knowledge at the sensor and the controller at time $t$ are represented by the \emph{information structures} $I_S(t)$ and $I_C(t)$. \txtblue{Specifically,  $I_S(t)$ is given by
\begin{align}
	I_S(t)&=\big\{\underbrace{\mathbf{x}_0^t, \mathbf{i}_0^t, \mathbf{u}_0^{t-1}}_{\text{control-related states}}, \underbrace{\boldsymbol{\xi}_0^{t}, \alpha_0^{t-1},   \gamma_0^{t-1}}_{\text{com.-related states}}\big\} \label{sensrolis}
\end{align}
for $t> 0$, and $I_S(0)=\left\{\mathbf{x}(0), \mathbf{i}(0), \boldsymbol{\xi}(0) \right\}$, where we denote $\alpha(t)=|h(t)|^2$ ($\forall t$). Moreover, $I_C(t)$ is given by
\begin{align}
	I_C(t)=\big\{\underbrace{\mathbf{u}_0^{t-1}}_{\text{control-related states}}, \underbrace{\alpha_0^{t}, \gamma_0^{\nu(t)}, \mathbf{y}_0^{\nu(t)}}_{\text{com.-related states}}\big\}	\label{controlis}
\end{align}
for $t> 0$, and $I_C(0)=\left\{\alpha(0), \gamma(\nu(0)), y(\nu(0))\right\}$, where $\nu(t)=\max\{s\leq t: \gamma(s)=1\}$ is the  slot of the latest successful transmission  by time  $t$. Note that at time slot $t$, the sensor discards\footnote{\txtblue{The reason is that the events of the form $\gamma(t)=0$ contain information about the plant state $\mathbf{x}(t)$ through the dependence of the SER in (\ref{serequ}). To avoid this  complication, we discard the  events of the form $\gamma(t)=0$ as in \cite{mdp1}.}} the  events of the form $\gamma(t)=0$.} We have the following observations on $I_S(t)$ and $I_C(t)$:
\begin{Remark}[Observations on  $I_S(t)$ and $I_C(t)$]	 \label{refmark21}  \
\begin{itemize}
	\item For  $I_S(t)$,  $\mathbf{x}_0^t$ is the plant state, $\mathbf{i}_0^t$ is the pre-processor output, $\mathbf{u}_0^{t-1}$ is the plant control action, \txtblue{$\boldsymbol{\xi}_0^{t}$ is the mapper output, and all   can be locally obtained at the sensor.  The  information $\left\{\alpha_0^{t-1},  \gamma_0^{t-1}\right\}$ can be obtained by the feedback signals from the controller as shown in Fig. \ref{artNSC}.} 
	\item For $I_C(t)$, $\mathbf{u}_0^{t-1}$ is the locally generated plant control action, $\alpha_0^{t}$ can be locally measured using the pilots from the sensor \cite{lte}, and $\{\gamma_0^{\nu(t)}, \mathbf{y}_0^{\nu(t)}\}$ are the symbol error indicators and the received signals,  which are the output of the wireless fading channel and can be locally obtained  at the controller. 
	\item  There is an intersection of $I_S(t)$ and $I_C(t)$,  which is denoted as $I_{SC}(t) \triangleq I_S(t) \cap I_C(t)$. Specifically, $I_{SC}(t)=\{\mathbf{u}_0^{t-1}, \alpha_0^{t-1},   \txtblue{\gamma_0^{\nu(t-1)}, \mathbf{y}_0^{\nu(t-1)}}\}$ for $t>0$,  and $I_{SC}(0)=\emptyset$.~\hfill~\IEEEQED
\end{itemize}
\end{Remark}

\vspace{-0.5cm}
\subsection{Communication Power and Plant Control Policies}	\label{policysdefinitions}

Based on $I_S(t)$ and $I_C(t)$, we define the communication power and plant control policies. \txtblue{Let $\mathcal{F}_S(t)=\sigma\left(\left\{I_S(s):   0\leq s \leq t\right\}\right)$ be the minimal $\sigma$-algebra containing the set $\left\{I_S(s):  0\leq   s \leq t\right\}$ and $\left\{\mathcal{F}_S(t)\right\}$ be the associated \emph{filtration}  at the sensor.  Similarly, define $\mathcal{F}_C(t)=\sigma\left(\left\{I_C(s):    0\leq s \leq t\right\}\right)$ and let $\left\{\mathcal{F}_C(t)\right\}$ be the associated \emph{filtration}  at the controller.}  At the beginning of time slot $t$, the sensor determines the power control action $p(t)$ and the controller determines the plant control action $\mathbf{u}(t)$ according to the following   control policies:
\txtblue{\begin{Definition}	[\txtblue{Plant Control Policy}]	\label{def111}
	\txtblue{A   \emph{plant control policy} $\Omega_{\mathbf{u}}$ at the controller is $\mathcal{F}_C(t)$-adapted, meaning that $\mathbf{u}(t)$ is adaptive to all the information $I_C(s)$ up to time slot $t$ (i.e., $\left\{I_C(s):  0\leq s \leq t\right\}$).~\hfill~\IEEEQED}
\end{Definition}
\begin{Definition}	[\txtblue{Communication Power Control Policy}]	\label{powerpoli}
	\txtblue{A  \emph{communication power control policy} $\Omega_p$ at the sensor is $\mathcal{F}_S(t)$-adapted, meaning that $p(t)$ is adaptive to all the information $I_S(s)$ up to time slot $t$ (i.e., $\left\{I_S(s):  0\leq s \leq t\right\}$). Furthermore, $\Omega_p$ satisfies the peak power constraint, i.e., $p(t) \in \left[0, p_{max}\right]$ for all $t$, where  $p_{max}>0$ is the maximum power the sensor can use at each time slot.~\hfill~\IEEEQED}
\end{Definition}}

\begin{Remark}	[Interpretation of the Power Control Policy]	\label{usefulremarkv}
	The  power control policy in Definition \ref{powerpoli}  is an \emph{error-adaptive} power control. In such  strategy,  the rate of the channel  $R$ is fixed, which means that the constellation size of the MQAM  (where $M=2^{R}$)  is unchanged  during the communication session. At each time slot $t$, the sensor controls the communication power $p(t)$  to dynamically adjust  the SER  in (\ref{serequ}), so that the  state estimation error at the controller is adjusted. Hence,  there is an inherent tradeoff between the plant performance (in the stability or optimal control sense) and the communication cost (in terms of the average transmit power). We shall quantify this tradeoff in the following sections. ~\hfill~\IEEEQED
\end{Remark}	

\vspace{-0.5cm}

\txtblue{\section{Communication Power  Problem Formulation}}

\txtblue{In this section, we first introduce a  primitive quantizer  at the sensor. We then establish the no dual effect property  under such a  quantizer and give the optimal plant control policy w.r.t. the joint communication power and plant control problem.  Based on the primitive quantizer  and the CE controller, we formally formulate the communication power  problem.} \txtblue{\subsection{Primitive Quantizer}	\label{sectionPQ}}

\txtblue{Following \cite{rate1} and \cite{ratenoise}, we adopt a primitive quantizer at the sensor to track the dynamic range of the plant state. Specifically, the primitive quantizer is characterized by four parameters $\left(\widetilde{\mathbf{x}}, \boldsymbol{\Psi}, \mathbf{L}, \mathbf{R}\right)$, where $\widetilde{\mathbf{x}}(t)\triangleq \mathbb{E}\left[\mathbf{x}(t)\big|I_{SC}(t)\right]\in \mathbb{R}^d$ is the shifting  vector\footnote{$\widetilde{\mathbf{x}}(t)$ also measures  the common information at both the sensor and the controller. \txtblue{Note that $\widetilde{\mathbf{x}}(t)$ is calculated based on $I_{SC}(t)$, and hence, $\widetilde{\mathbf{x}}(t)$ can be locally maintained at both the sensor and the controller according to the discussions in  Remark \ref{refmark21}.}}, $\boldsymbol{\Psi}(t)\in \mathbb{R}^{d\times d}$ is the coordinate transformation matrix, $\mathbf{L}(t)=(L_1(t), \dots, L_d(t))^T\in \mathbb{R}^d$ is the dynamic range, and $\mathbf{R}=(R_1, \dots, R_d)^T \in \mathbb{R}^d$ is the rate vector with $\sum_{n=1}^d R_n=R$ (where $R_n$ determines the quantization level of the $n$-th element of the plant state $\mathbf{x}$).}

%\begin{figure}[t]
%\centering
%  \hspace{1cm}
%  \includegraphics[width=3.5in]{quantizer.pdf}\vspace{-0.4cm}
%  \caption{\txtblue{An example  of the  primitive quantizer structure, where the plant state $\mathbf{x}$ has two dimensions and the rate vector is $\mathbf{R}=(2, 2)$ (which means that the quantization levels for the two elements in $\mathbf{x}$ are 4.).}}
%  \label{quantiillus}\vspace{-0.6cm}
%\end{figure}

\txtblue{The primitive quantizer consists of three components, i.e., a pre-processor, a mapper and an encoder.  Specifically, the pre-processor takes  $\mathbf{x}(t)$ as input and generates the innovation $\mathbf{i}(t)=\mathbf{x}(t)-\widetilde{\mathbf{x}}(t)$, which is passed to the mapper. The mapper  maps $\boldsymbol{\Psi}(t)\mathbf{i}(t)$ to one of the $M=2^{R}$  quantization levels  within the region $\big\{[-L_1(t),L_1(t)]\times \cdots \times  [-L_d(t),L_d(t)]\big\}$ and outputs $\boldsymbol{\xi}(t)$. Then $\boldsymbol{\xi}(t)$ is encoded into  $R$-digit binary information  bits. Please refer to Appendix A on how  the primitive quantizer works in detail. We summarize the property of the primitive quantizer as follows:}
\txtblue{\begin{Lemma}	[Properties of the Primitive Quantizer]	\label{lemmaqunt} \
	\begin{itemize}	
		\item The primitive quantizer  tracks the dynamic range of the plant state $\mathbf{x}(t)$, i.e.,
		\begin{align}
			&\mathbf{x}(t)\in \big\{\mathbf{x}\in \mathbb{R}^d: \boldsymbol{\Psi}(t)\big(\mathbf{x}-\widetilde{\mathbf{x}}(t)\big)\notag \\
			&\in \left\{[-L_1(t),L_1(t)]\times \cdots \times  [-L_d(t),L_d(t)]\right\}\big\}, \ \forall t	\label{quantprop1}
		\end{align}
		\item The equivalent model between input $\mathbf{x}(t)$ and  mapper output $\boldsymbol{\xi}(t)$ can be expressed as
		\begin{align}
			\boldsymbol{\xi}(t)= \boldsymbol{\Psi}(t)\big(\mathbf{x}(t)-\widetilde{\mathbf{x}}(t)\big)+\mathbf{e}(\mathbf{L}(t), t)	\label{quantprop2}
		\end{align} 
		where $\mathbf{e}(\mathbf{L}(t), t)=\big(e_1(L_1(t), t),\dots, e_d(L_d(t), t)\big)^T\in\mathbb{R}^d$ is the quantization noise, and for each $n$, $e_n(L_n(t), t)$ is uniformly distributed within the region $\big[-\frac{L_n(t)}{2^{R_n-1}},\frac{L_n(t)}{2^{R_n-1}}\big]$. ~\hfill\ \IEEEQED
	\end{itemize}
\end{Lemma}}
\txtblue{Therefore, according to (\ref{quantprop1}), we can use  the primitive quantizer to track the dynamic range of the plant state $\mathbf{x}(t)$. \vspace{-0.2cm}
\subsection{Certainty Equivalent Controller}
As in  \cite{onoff2} and \cite{mdp1}, we consider the following joint communication power and plant control optimization problem:
\bs\begin{align} 
	\min_{\Omega_\mathbf{u}, \Omega_p}	\  \limsup_{T \rightarrow \infty} \frac{1}{T} &\sum_{t=0}^{T-1} \mathbb{E}^{\Omega_\mathbf{u}, \Omega_p}\big[\mathbf{x}^T(t)\mathbf{Q} \mathbf{x}(t) + \mathbf{u}^T(t)\mathbf{D} \mathbf{u}(t)+ \lambda p(t)\big] \notag
\end{align}\bsc where $\mathbf{Q}$ and $\mathbf{D}$ are positive definite symmetric weighting matrices   for the plant state deviation cost and the plant control cost, and $\lambda \in \mathbb{R}^+$ is  the  communication power price. In general, the design of the communication power   policy  and the plant control policy  are coupled together. This is because the communication power will affect the state estimation accuracy at the controller, which will in turn affect the plant state evolution. However, by establishing   the no dual effect property (e.g., \cite{rate2}, \cite{dualeffect}), we can obtain the optimal plant control policy for the above joint optimization problem. Specifically, let $\hat{\mathbf{x}}(t)=\mathbb{E}\big[\mathbf{x}(t)\big|I_C(t)\big]$ be the  plant state estimate at the controller and  $\boldsymbol{\Delta}(t)=\mathbf{x}(t)-\hat{\mathbf{x}}(t)$  be the state estimation error. The no dual effect property is established as follows:
\begin{Lemma}	[No Dual Effect Property]	\label{lemmadual}
	Under the primitive quantizer in Section \ref{sectionPQ}, we have the following no dual effect property in our NCS:
	\begin{align}
		\mathbb{E}\big[\boldsymbol{\Delta}^T(t)\boldsymbol{\Delta}(t)\big| I_C(t)\big]=\mathbb{E}\big[\boldsymbol{\Delta}^T(t)\boldsymbol{\Delta}(t)\big| \alpha_0^{t}, \gamma_0^{\nu(t)}, \mathbf{y}_0^{\nu(t)}\big], \quad \forall t	\notag 
	\end{align}
\end{Lemma}
\begin{proof}
	please refer to Appendix B.
\end{proof}}
\txtblue{Using the no dual effect property in Lemma \ref{lemmadual} and Prop. 3.1 of \cite{rate2} (or  Theorem 1  in Section III of \cite{dualeffect}), the optimal plant control policy  is given by the certainty equivalent (CE) controller:
\begin{align}	\label{statsrr}
	\Omega_{\mathbf{u}}^\ast  \big(I_C(t)\big)=\mathbf{u}^\ast(t), \quad         \text{where }   \mathbf{u}^\ast(t)=-\mathbf{K} \hat{\mathbf{x}}(t), \quad \forall t
\end{align}
where $\mathbf{K}= (\mathbf{G}^T\mathbf{P} \mathbf{G} + \mathbf{D})^{-1}\mathbf{G}^T\mathbf{P} \mathbf{F}$ is the feedback gain matrix,  and    $\mathbf{P}$ satisfies the following  discrete time algebraic Riccati equation\footnote{We assume that  $(\mathbf{F}, \mathbf{W}^{1/2})$ is observable  as in the classical LQG control theories. This assumption together with Assumption \ref{assumpstabl} ensures that the DARE has a unique symmetric positive semidefinite solution \cite{poweronl}.} (DARE): $\mathbf{P} = \mathbf{F}^T\mathbf{P}\mathbf{F} - \mathbf{F}^T\mathbf{P} \mathbf{G}(\mathbf{G}^T\mathbf{P} \mathbf{G} + \mathbf{D})^{-1}\mathbf{G}^T\mathbf{P}\mathbf{F} + \mathbf{W}$.}\vspace{-0.2cm}\txtblue{\subsection{Communication Power  Problem Formulation}}
\txtblue{Under the primitive quantizer  and the CE controller, the per-stage state estimation error and communication power cost is given by
\begin{align}	
	c\big(\boldsymbol{\Delta}(t), p(t)\big)=\boldsymbol{\Delta}^T(t)\mathbf{S} \boldsymbol{\Delta}(t)+\lambda p(t)	\label{asdsddpers}
\end{align}
where  $\mathbf{S}$ is a positive definite symmetric weighting matrix. We consider the following  communication power optimization problem:
\begin{Problem}	\emph{(Communication Power Optimization Problem):}	\label{probformu}
\begin{align} 
	\min_{\Omega_p}	\quad  &  \limsup_{T\rightarrow \infty} \ \frac{1}{T}\sum_{t=0}^{T-1}\mathbb{E}^{\Omega_\mathbf{u}^\ast, \Omega_p} \left[c\big(\boldsymbol{\Delta}(t), p(t)\big)\right]\label{perstagecost} 
\end{align}
where $\boldsymbol{\Delta}(t)$ has the following dynamics:
\begin{align}	\label{deltadyntext}	
	 \boldsymbol{\Delta}(t)=
 		\left\{
			\begin{aligned}	 
				& \mathbf{F}\boldsymbol{\Delta}(t-1)+\mathbf{w}(t-1),  \quad &\text{if }	 \gamma(t)=0	   	\\
				&  -\boldsymbol{\Psi}^{-1}(t)\mathbf{e}(\mathbf{L}(t), t)	, \quad &\text{if }	\gamma(t)=1
			 \end{aligned}
   		\right.
 \end{align}
 where $\mathbf{e}(\mathbf{L}(t), t)$ is the quantization noise in (\ref{quantprop2}).~\hfill\IEEEQED
\end{Problem}}

We need to design a communication power control policy such that the plant system and the primitive are stable, and  state estimation error is bounded. Specifically, we have the following definition on the admissible communication power control  policy of the NCS under the CE controller $\Omega_{\mathbf{u}}^\ast$ in (\ref{statsrr}):
\begin{Definition}	\emph{(Admissible  Communication Power Control Policy):}	\label{admisscontrolpol}
	A communication power control  policy $\Omega_p$  of the NCS is admissible if,
	\begin{itemize}
		\item	 The plant state process $\left\{\mathbf{x}\left(t\right)\right\}$ is stable in the sense that $\lim_{t \rightarrow \infty}\mathbb{E}^{\Omega_{\mathbf{u}}^\ast, \Omega_p}\big[\left\|\mathbf{x}(t)\right\|^2 \big]< \infty$, where $ \mathbb{E}^{\Omega_p,\Omega_{\mathbf{u}}}$ means taking expectation w.r.t. the probability measure induced by  $\left(\Omega_{\mathbf{u}}^\ast, \Omega_p\right)$.
		\item 	 The  dynamic range of the quantizer  $\left\{\mathbf{L}(t) \right\}$   is stable in the sense that $\lim_{t \rightarrow \infty}\mathbb{E}^{\Omega_{\mathbf{u}}^\ast, \Omega_p}\big[\left\|\mathbf{L}(t)\right\|^2 \big]  < \infty$.
		\item 	 The  process $\left\{\boldsymbol{\Delta}(t) \right\}$ at the sensor  is stable in the sense that $\lim_{t \rightarrow \infty}\mathbb{E}^{\Omega_{\mathbf{u}}^\ast \Omega_p}\big[\left\|\boldsymbol{\Delta}(t)\right\|^2 \big]< \infty$.~\hfill~\IEEEQED
	\end{itemize}
\end{Definition}

\txtblue{Problem \ref{probformu} is an MDP and we show in Appendix C that it is without loss of optimality that we restrict the system state to be  $\boldsymbol{\chi}(t) \triangleq \big\{\boldsymbol{\Delta}(t-1), \alpha(t-1), \boldsymbol{\Psi}(t),\mathbf{L}(t) \big\}\in \sigma\left(\left\{I_S(s):   0\leq s \leq t\right\}\right)$ with transition kernel $\Pr\big[\boldsymbol{\chi}(t+1)\big| \boldsymbol{\chi}(t), p(t)\big]$.  Using dynamic programming  theories \cite{mdpcref2}, the  optimality conditions of Problem \ref{probformu} are given as follows:}

\begin{Theorem}	[Sufficient Conditions for Optimality]	\label{mdpvarify}
	If there exists $\left(\theta^\ast, V^\ast(\boldsymbol{\chi})\right)$ that satisfies the following optimality equation (Bellman equation):
	\bs\begin{align}
		 \theta^\ast \tau +  V^\ast \left(\boldsymbol{\chi} \right) = &\min_{p\in \Omega_p(\boldsymbol{\chi})}  \mathbb{E}\big[\big((\boldsymbol{\Delta}')^T\mathbf{S} (\boldsymbol{\Delta}') + \lambda p \big)\tau \notag \\
		 &+ \sum_{\boldsymbol{\chi} '}\Pr\left[\boldsymbol{\chi}'\big| \boldsymbol{\chi}, p  \right] V^\ast \left(\boldsymbol{\chi}'\right)\big|  \boldsymbol{\chi}\big],\quad \forall \boldsymbol{\chi} 	\label{OrgBel}
	\end{align}\bsc and for all admissible communication power control policies $\Omega_p$, $V^\ast (\boldsymbol{\chi})$ satisfies the following transversality condition:
	 \begin{align}	\label{transodts}
	\lim_{T\rightarrow \infty} \frac{1}{T}\mathbb{E}^{\Omega_p}\left[ V^\ast\left(\boldsymbol{\chi} \left(T \right) \right) |\boldsymbol{\chi} \left(0 \right)\right]=0
\end{align}
Then, we have  the following results:
\begin{itemize}
	\item $\theta^\ast =\min_{\Omega_p}  \limsup_{T\rightarrow \infty}  \frac{1}{T}\sum_{t=0}^{T-1}\mathbb{E}^{\Omega_\mathbf{u}^\ast, \Omega_p} \left[c\big(\boldsymbol{\Delta}(t), p(t)\big)\right]$ is the optimal average cost  of Problem \ref{probformu}.
	\item Suppose there exists an admissible communication power  control policy $\Omega_p^\ast$ with  $\Omega_p^*\left(\boldsymbol{\chi} \right) = p^\ast$, where  $p^\ast$ attains the minimum of the R.H.S. in (\ref{OrgBel}) for given $\boldsymbol{\chi}$. \txtblue{Then, the optimal communication power  policy of Problem \ref{probformu} is given by a stationary Markovian policy $\Omega_p^\ast$.}	
\end{itemize}
\end{Theorem}
\begin{proof}
	please refer to Appendix C.
\end{proof}

Unfortunately, the Bellman equation in (\ref{OrgBel}) is very difficult to solve because it involves a huge number of fixed point equations w.r.t. $\left(\theta^\ast, V^\ast(\boldsymbol{\chi})\right)$. Numerical solutions such as numerical VIA \cite{mdpcref2}, \cite{solvebellman}  have exponential complexity\footnote{\txtblue{Since the state space of  $\boldsymbol{\chi}$  is continuous, the numerical VIA refers to the finite difference method for  solving an equivalent discretized Bellman equation \cite{mdpcref2}, \cite{solvebellman} using the conventional VIA. Suppose the state space of each element in  $\boldsymbol{\chi} = \big\{\boldsymbol{\Delta}, \alpha, \boldsymbol{\Psi},\mathbf{L} \big\}$  is discretized into $\beta$ intervals. Then the cardinality of  $\boldsymbol{\chi}$ is $\beta^{d^2+2d+1}$.}} w.r.t. $d^2$, where $d$ is the dimension of the plant state $\mathbf{x}$).

%\vspace{-0.1cm}
%\framebox{\begin{minipage}[t]{0.88\columnwidth}
%	\begin{center}
%	 	\vspace{0.1cm} \emph{Challenge 1: }\ Exponential complexity to solve the optimality equation in (\ref{OrgBel}).\vspace{0.1cm}
%	\end{center} 
%\end{minipage}}
%\vspace{0.3cm}

 In Section \ref{pertuabationana}, we shall derive a closed-form approximation for $V(\boldsymbol{\chi})$ using continuous-time perturbation techniques. 
\vspace{-0.2cm}

\section{Low Complexity Power Control Solution}	\label{pertuabationana}

In this section, we  first adopt a continuous-time perturbation approach to analyze the difference between a closed-form approximate value function and the optimal value function. We  analyze the performance gap between the policy obtained from the continuous-time perturbation approach and the optimal control policy. Then, we focus on deriving   the closed-form approximate value function and  proposing a low complexity power control scheme. We also  give sufficient conditions for the NCS stability.

\subsection{Continuous-Time Approximation}
We first have the following theorem for solving the optimality equation in (\ref{OrgBel}) using continuous-time perturbation analysis \cite{fluid2}, \cite{dominmethod}:
\begin{Lemma}	\emph{(Perturbation Analysis for Solving the Optimality Equation):}	\label{perturbPDE}
	If there exists $\left(\theta, V(\boldsymbol{\chi})\right)$ where\footnote{$V\in \mathcal{C}^2$ means that $V(\boldsymbol{\chi})$ is second order differentiable w.r.t. to each variable in $\boldsymbol{\chi}$.} $V\in \mathcal{C}^2$  that satisfies
	
	\begin{itemize}
		\item the following multi-dimensional PDE:
	\txtblue{\begin{align}	
		&\theta= \boldsymbol{\Delta}^T \mathbf{S} \boldsymbol{\Delta} \label{bellman2} \\
		&+  \min_{0\leq p  \leq p_{max}}\big[\lambda + \left(V(\boldsymbol{\chi})\big|_{\boldsymbol{\Delta}=\mathbf{0}} - V(\boldsymbol{\chi})\right)\frac{  \alpha}{\kappa(R) B_W}\big]p      \notag \\
		&+\nabla_{\boldsymbol{\Delta}}^TV(\boldsymbol{\chi}) \widetilde{\mathbf{F}} \boldsymbol{\Delta}  +\frac{1}{2}\text{Tr}\left( \nabla_{\boldsymbol{\Delta}}^2 V(\boldsymbol{\chi} )\widetilde{\mathbf{W}}\right)    \notag \\
		&+\frac{\partial V(\boldsymbol{\chi})}{\partial \alpha} \left(2 \widetilde{a} \alpha +2 \widetilde{a}  \right)+\frac{\partial^2V^\ast(\boldsymbol{\chi})}{\partial \alpha^2}4 \widetilde{a}\alpha \notag \\
	& + \text{Tr}\Big(\frac{\partial V(\boldsymbol{\chi})}{\partial \boldsymbol{\Psi}}(\mathbf{H}\boldsymbol{\Psi}-\boldsymbol{\Psi})/\tau\Big) \notag \\
	&+ \nabla^T_{\mathbf{L}}V(\boldsymbol{\chi})\big(\boldsymbol{\Gamma}\mathbf{F}_{\mathbf{R}}\mathbf{L}+w_{max}\|\mathbf{H}\boldsymbol{\Psi}\|\mathbf{1}-\mathbf{L}\big)/\tau\notag
	\end{align}}
	\item and $V(\boldsymbol{\chi})=\mathcal{O}(\left\|\boldsymbol{\Delta}\right\|^2)$.
	\end{itemize}
	  Then,  for any $\boldsymbol{\chi}$,
	\begin{align}	\label{errortermtaur}
		V^\ast \left(\boldsymbol{\chi} \right)  =  V(\boldsymbol{\chi}) + \mathcal{O}\left(\tau\right)+ \txtblue{\mathcal{O}\Big(\frac{\tau^2}{2^{2R_{min}}}\Big)}
	\end{align}
	as  $\tau \rightarrow 0$, where \txtblue{$R_{min}=\min\{R_n: \forall n\}$}, and $\mathcal{O}\left(\tau\right)$  and \txtblue{$\mathcal{O}\big(\frac{\tau^2}{2^{2R_{min}}}\big)$} are the error terms  due to  the continuous time approximation and the quantization, respectively. Furthermore,   $V(\boldsymbol{\chi})$ satisfies the transversality condition in (\ref{transodts}).
\end{Lemma}
\begin{proof}
	please refer to Appendix D.
\end{proof}

As a result, solving the  optimality equation in (\ref{OrgBel}) is transformed into a calculus problem of solving the PDE in (\ref{bellman2}), and the difference between $V(\boldsymbol{\chi})$ and $V^\ast(\boldsymbol{\chi})$ is $\mathcal{O}\left(\tau\right)+ \mathcal{O}\big(\frac{\tau^2}{2^{2R_{min}}}\big)$ for sufficiently small  $\tau$.

Let $\widetilde{\Omega}_p^\ast$ be the control policy, under which the generated control action achieves the minimization in the PDE in (\ref{bellman2}) for any $\boldsymbol{\chi}$. Let  $\widetilde{\theta}^\ast=\limsup_{T\rightarrow \infty}  \frac{1}{T}\sum_{t=0}^{T-1}\mathbb{E}^{\Omega_\mathbf{u}^\ast, \widetilde{\Omega}_p^\ast} \big[c\big(\boldsymbol{\Delta}(t), p(t)\big)\big]$ be the associated    performance. The  gap between $\widetilde{\theta}^\ast$ and the optimal average  cost $\theta^{\ast}$ in (\ref{OrgBel}) is established as follows:
\begin{Theorem}	[Performance Gap between $\widetilde{\theta}^\ast$ and $\theta^{\ast}$]	\label{perfgap}
	If $V(\boldsymbol{\chi})=\mathcal{O}(\left\|\boldsymbol{\Delta}\right\|^2)$ and $\widetilde{\Omega}_p^\ast$ is admissible, we have 
	\begin{align}		\label{perfgapequ}
		\widetilde{\theta}^\ast - \theta^{\ast} =\mathcal{O}(\tau)+\mathcal{O}\Big(\frac{\tau^2}{2^{2R_{min}}}\Big) 
	\end{align}
	as  $\tau \rightarrow 0$.
	\end{Theorem}
\begin{proof}
Please refer to Appendix E.		
\end{proof}

Theorem \ref{perfgap} suggests that  $\widetilde{\theta}^\ast \rightarrow \theta^{\ast} $ as $\tau \rightarrow 0$. In other words, the  power control  policy  $\widetilde{\Omega}_p^\ast$ is asymptotically optimal for sufficiently small $\tau$.

%\vspace{-0.1cm}
%\framebox{\begin{minipage}[t]{0.9\columnwidth}
%	\begin{center}
%	 	\vspace{0.1cm}  \emph{Challenge 2: } \ Multi-dimensional coupled PDE  in (\ref{bellman2}).\vspace{0.1cm} 
%	\end{center}
%\end{minipage}}
%\vspace{0.2cm}

\vspace{-0.2cm}

\subsection{Closed-Form Approximate Value Function}
In this subsection, we solve the PDE in (\ref{bellman2}) to obtain the closed-form approximate value function $V(\boldsymbol{\chi})$. It can be observed that the PDE in (\ref{bellman2}) is multi-dimensional and coupled in the variables $(\boldsymbol{\Delta}, \alpha, \boldsymbol{\Psi}, \mathbf{L})$, and   it is quite challenging to obtain the closed-form solution.  In the following Lemma, we derive an asymptotic solution of the PDE using  the asymptotic expansion technique \cite{dominmethod}.
\begin{Lemma}	[Asymptotic Solution of the PDE]	\label{solpde}
The asymptotic solution of the PDE in (\ref{bellman2}) is given as follows:
	\begin{itemize}
		\item  for small $\|\boldsymbol{\Delta}\|^2\alpha$ (Bad Transmission Opportunity/Low Urgency Regime),
		\begin{align}	\label{valuefunc1}
			\hspace{-1cm}\log\left(V(\boldsymbol{\chi})\right)= \log (\boldsymbol{\Delta}^TA_{1,\mathbf{U}}(\alpha)\boldsymbol{\Delta}+b_1(\alpha))+\mathcal{O}\left({\alpha}\right)
		\end{align}
		 where $A_{1,\mathbf{U}}(\alpha)\in \mathbb{R}^{d\times d}$ and  $b_1(\alpha)\in \mathbb{R}$ are given in (\ref{appeox1}) in Appendix F.
		\item for large $\|\boldsymbol{\Delta}\|^2\alpha$ (Good Transmission Opportunity/High Urgency Regime), 
		\begin{align}	
			&\hspace{-1cm}\log\left(V(\boldsymbol{\chi})\right)= \log (\boldsymbol{\Delta}^TA_{2,\mathbf{U}}(\alpha)\boldsymbol{\Delta}+ \exp\left(2B_2\right)\alpha^{C_2} ) 	\notag\\
			&-\frac{\widetilde{a}-\widetilde{c}}{4}\alpha-\big(\frac{1}{4}-\frac{\widetilde{a}}{4\widetilde{c}}\big)\log \alpha+\mathcal{O}\big(\frac{1}{\alpha \sqrt{p_{max}}}\big)\label{valuefunc2}
		\end{align}
		where $A_{2,\mathbf{U}}(\alpha)\in \mathbb{R}^{d\times d}$ and  $B_2\in \mathbb{R}$  are given in (\ref{appeox2}) in Appendix F.
		\end{itemize}
\end{Lemma}
\begin{proof}
	please refer to Appendix F.
\end{proof}

Based on Lemma \ref{solpde},  we adopt the following approximation for the solution of the PDE  in (\ref{bellman2}):
\begin{align}	
	&\log\left(V(\boldsymbol{\chi})\right)\approx \txtblue{\log\big(\widetilde{V}_{\eta_{th}}(\boldsymbol{\chi})\big)}\notag \\
	\triangleq &\left\{
	\begin{aligned}\label{approxvaluefunc} 
		 & \log (\boldsymbol{\Delta}^TA_{1,\mathbf{U}}(\alpha)\boldsymbol{\Delta}+b_1(\alpha)), \text{for } \|\boldsymbol{\Delta}\|^2\alpha<\eta_{th}\\
		&\log (\boldsymbol{\Delta}^TA_{2,\mathbf{U}}(\alpha)\boldsymbol{\Delta}+ \exp\left(2B_2\right)\alpha^{C_2} )-\frac{\widetilde{a}-\widetilde{c}}{4}\alpha \\
		&-\left(\frac{1}{4}-\frac{\widetilde{a}}{4\widetilde{c}}\right)\log \alpha,  \hspace{0.1cm} \text{for } \|\boldsymbol{\Delta}\|^2\alpha \geq \eta_{th}
	\end{aligned}
	\right.
\end{align}
where $\eta_{th}>0$ is a solution parameter. 

\vspace{-0.2cm}

\subsection{Structural Properties of the Low Complexity Power Control}
Using Lemma  \ref{perturbPDE}  and (\ref{approxvaluefunc}), we obtain a low complexity power control policy in the following theorem:
\begin{Theorem}		[Structural Properties of Power Control Policy]	\label{thmpower}
	The optimizing power control policy  $\widetilde{\Omega}_p^\ast$ that minimize the R.H.S. of  the PDE in (\ref{bellman2}) is given by
	\bs\begin{align}	\label{powersolu}
		&\widetilde{\Omega}_p^\ast\left(\boldsymbol{\chi}\right)=p^\ast  \\
		=&\left\{\begin{aligned}
	&0,	\ \text{if }\lambda >\txtblue{\big(\widetilde{V}_{\eta_{th}}(\boldsymbol{\chi}) -\widetilde{V}_{\eta_{th}}(\boldsymbol{\chi})\big|_{\boldsymbol{\Delta}=\mathbf{0}} \big) \frac{  \alpha}{\kappa(R) B_W}}	\\
	&p_{max}, \   \text{if }\lambda \leq \txtblue{\big(\widetilde{V}_{\eta_{th}}(\boldsymbol{\chi}) -\widetilde{V}_{\eta_{th}}(\boldsymbol{\chi})\big|_{\boldsymbol{\Delta}=\mathbf{0}} \big) \frac{  \alpha}{\kappa(R) B_W}}\notag
	\end{aligned}
	\right.
	\end{align}\bsc~\hfill~\IEEEQED
\end{Theorem}

%\begin{figure}
%\centering
%  \includegraphics[width=4.3in]{fig7}
%  \caption{ The system parameters are configured as follows: $\widetilde{a}=-5$, $\tau=0.05$,  $p_{max}=180$, $R=3$, $B_W=1$,  $\widetilde{\mathbf{w}}=1$,  $\widetilde{\mathbf{F}}=-3$ with system state dimension $L=1$, $\lambda=2000$, and $\eta_{th}=0.76$ (which is determined according to the method in Section \ref{labelsdsa}).}
%  \label{fig7}
%\end{figure}

\begin{figure}
\centering
\subfigure[]{
\includegraphics[width=2.4in]{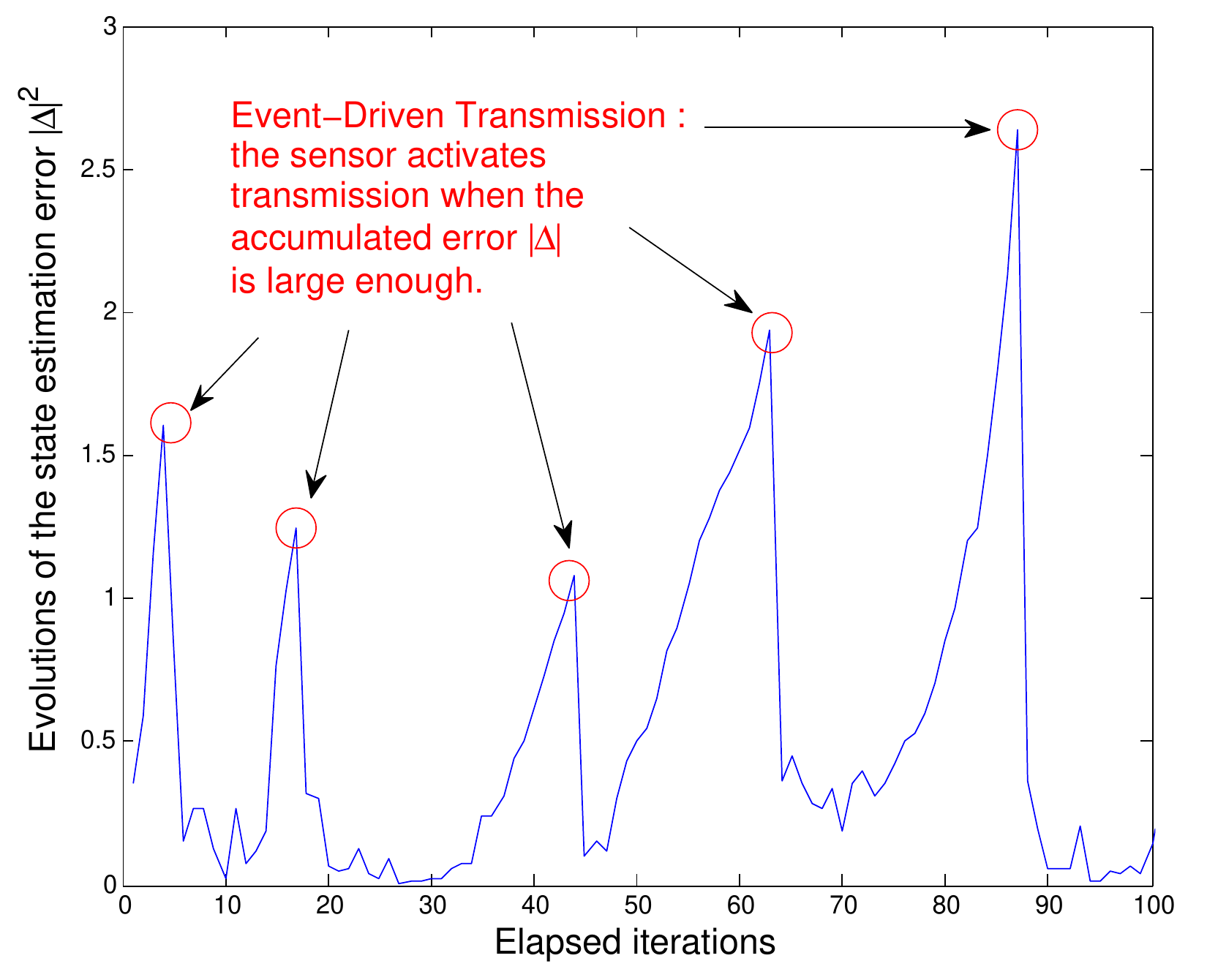}}\vspace{-0.2cm}
\subfigure[]{
\includegraphics[width=2.7in]{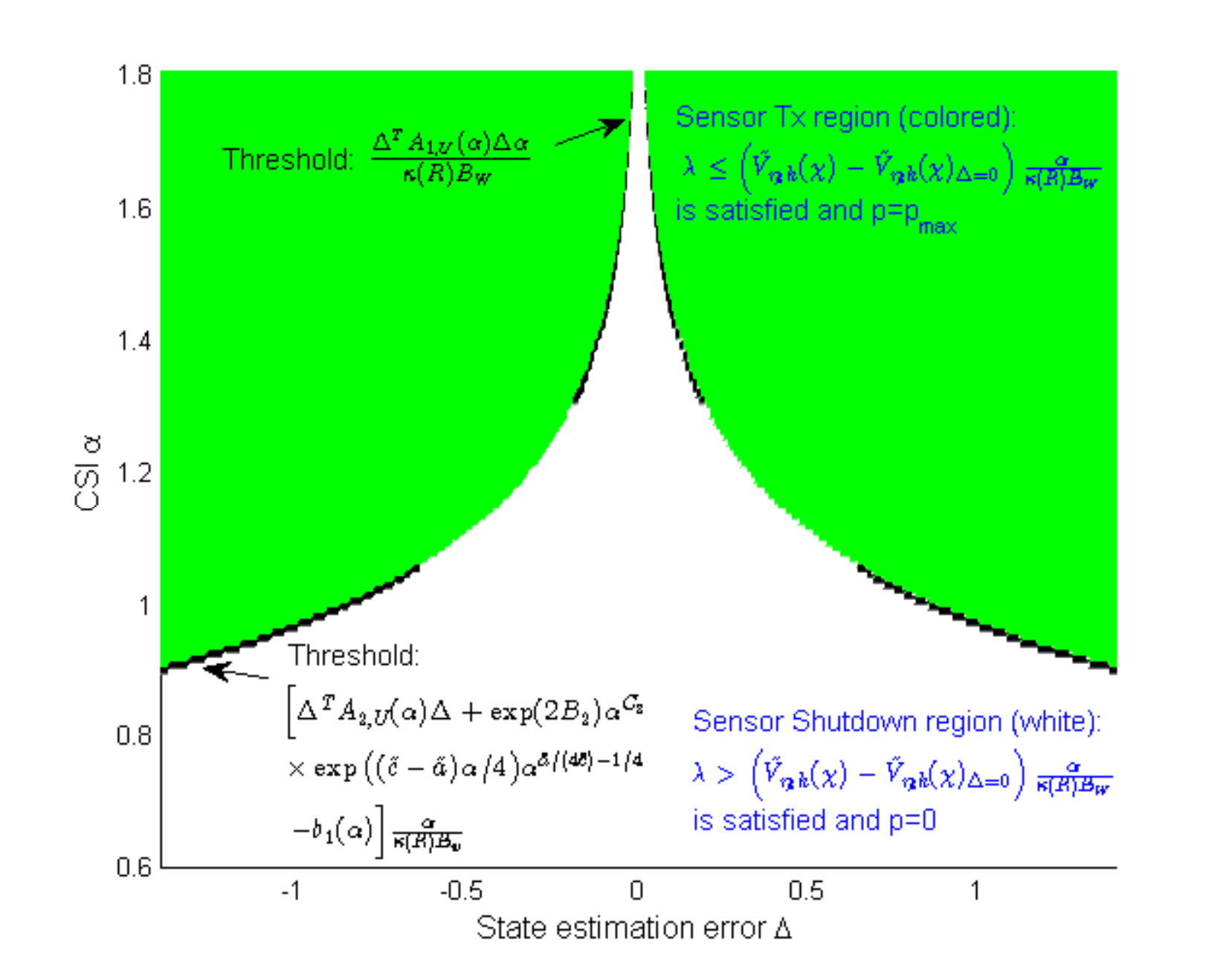}}\vspace{-0.2cm}
\caption{\txtblue{(a) Evolutions of  $\|\boldsymbol{\Delta}(t)\|^2$ under $\widetilde{\Omega}_p^\ast$.} (b) Decision region (sensor shutdown/Tx region) under $\widetilde{\Omega}_p^\ast$. The system parameters are configured as in Fig. \ref{comparsdeees2}, and  $\widetilde{\mathbf{F}}=-3$ with system state dimension $L=1$, $\lambda=2000$, and $\eta_{th}=0.43$ (which is determined according to the method in Fig.  \ref{compareees2}).}
\label{approxqua}\vspace{-0.6cm}
\end{figure}

It can be observed that the power control policy  $\widetilde{\Omega}_p^\ast$ has an \emph{event-driven control structure} with a dynamically changing threshold  $\txtblue{\big(\widetilde{V}_{\eta_{th}}(\boldsymbol{\chi}) -\widetilde{V}_{\eta_{th}}(\boldsymbol{\chi})\big|_{\boldsymbol{\Delta}=\mathbf{0}} \big) \frac{  \alpha}{\kappa(R) B_W}}$. Specifically, the sensor either transmits using the maximum power or shots down depending whether the dynamic threshold  is larger than $\lambda$ or not. \txtblue{Fig. \ref{approxqua}(a) illustrates a sample path  of the state estimation error $\boldsymbol{\Delta}(t)$  and it can be observed that the sensor only activates transmission when the accumulated  error is large enough.} Furthermore, the dynamic threshold is adaptive to the   plant state estimation error  $\boldsymbol{\Delta}$  and the CSI $\alpha$. Using the approximate value function  in (\ref{approxvaluefunc}), we have the following discussions\footnote{\txtblue{Please refer to Appendix G on  the order growth results in Remark \ref{remadksdsds}.}} on the dynamic threshold:
\begin{Remark}	[Properties of the Dynamic Threshold]	\label{remadksdsds}
	The dynamic threshold  is affected by the following factors:
	\begin{itemize}
		\item	\textbf{Dynamic Threshold  w.r.t. $\boldsymbol{\Delta}$:}  For given CSI  $\alpha$ and data rate $R$,  the dynamic threshold  increases w.r.t. the state estimation error $\boldsymbol{\Delta}$ at the order of $\mathcal{O}(\|\boldsymbol{\Delta}\|^2)$. This means that large state estimation error $\|\boldsymbol{\Delta}\|$ tends to use full  power. This is reasonable because  large state estimation error means  \txtblue{the {\em urgency} of delivering information to the controller}, which leads to use  large power.  
		\item	\textbf{Dynamic Threshold  w.r.t. $\alpha$:}   For  given state estimation error   $\boldsymbol{\Delta}$ and data rate $R$, the dynamic threshold increases w.r.t. the CSI $\alpha$  at the order of $\mathcal{O}( \exp(\alpha))$. This means that large  CSI $\alpha$ tends to use full  power.  Note that large $\alpha$ means good transmission \txtblue{opportunities}. Hence, it is reasonable to use more power to reduce the SER.
		\item	\textbf{Dynamic Threshold  w.r.t. $R$:} For given large CSI $\alpha$  or large  state estimation error   $\|\boldsymbol{\Delta}\|$, the dynamic threshold  increases w.r.t. $R$ at the order of \bs$\mathcal{O}(\exp(\sqrt{\kappa(R)})  \big/\sqrt{\kappa(R)})$\bsc. This means that large  data rate $R$ tends to use full  power. This is reasonable because large data rate leads to high SER\footnote{\txtblue{Large $R$ leads to the decrease of the minimal distance between the transmitted symbols in the constellation diagram, which results in an increase in the SER \cite{digitalcom}.}}, which also leads to use large power to  increase the chance of using $p_{max}$.~\hfill~\IEEEQED
	\end{itemize}
\end{Remark}

Fig. \ref{approxqua}(b) illustrates the decision region for $p^\ast = 0$ and $p^\ast = p_{max}$ for given system parameter configurations.

\vspace{-0.4cm}
\subsection{Stability Conditions and   Performance Gaps} 
We verify that the derived power control policy $\widetilde{\Omega}_p^\ast$ in (\ref{powersolu}) belongs to an admissible control policy according to Definition \ref{admisscontrolpol}. The conclusion  is summarized below. 
\txtblue{\begin{Theorem}	[Sufficient Conditions for NCS Stability]	\label{stab}
	If the NCS satisfies the following conditions,
	\bs\begin{align}	\label{sadsaxxxx}
		&\frac{ p_{max} \tau}{\kappa(R) B_W+p_{max} \tau} >\\
		&\max\Big\{ \frac{1}{R}\sum_{\mu_i(\mathbf{F})} \max\left\{0, \log|\mu_i(\mathbf{F})|\right\}, 1-\frac{1}{\mu_{max}(\mathbf{F}^T\mathbf{F})}\Big\} \notag 
	\end{align}\bsc 	then  using the  power control policy $\widetilde{\Omega}_p^\ast$ in (\ref{powersolu}) and the plant control policy  $\Omega_{\mathbf{u}}^\ast$ in (\ref{statsrr}),  the NCS is stable in the sense that $\lim_{t \rightarrow \infty}\mathbb{E}^{\Omega_{\mathbf{u}}^\ast, \widetilde{\Omega}_p^\ast}[\left\|\mathbf{x}(t)\right\|^2 ]< \infty$,  $\lim_{t \rightarrow \infty}\mathbb{E}^{\Omega_{\mathbf{u}}^\ast, \widetilde{\Omega}_p^\ast}[\left\|\mathbf{L}(t)\right\|^2 ]< \infty$, and $\lim_{t \rightarrow \infty}\mathbb{E}^{\Omega_{\mathbf{u}}^\ast, \widetilde{\Omega}_p^\ast}[\left\|\boldsymbol{\Delta}(t)\right\|^2 ]< \infty$.
\end{Theorem}
\begin{proof}
	please refer to Appendix H.
\end{proof}}
\begin{Remark}	[Discussions of the Stability Conditions]
	\txtblue{Theorem \ref{stab} gives the conditions to ensure  NSC stability. Under the conditions, we have $R>(1+\frac{\kappa(R) B_W}{ p_{max} \tau}) \mathcal{I}(\mathbf{F})> \mathcal{I}(\mathbf{F})$,}  where $\mathcal{I}(\mathbf{F})\triangleq \sum_{\mu_i(\mathbf{F})} \max\{0, \log|\mu_i(\mathbf{F})|\}$ is the \emph{instability measure} \cite{rate1}, \cite{rate2} of the plant system.  \txtblue{Note that  the term $(1+\frac{\kappa(R) B_W}{ p_{max} \tau})\mathcal{I}(\mathbf{F})$ is equivalent to $\frac{1}{P_{Succ, \widetilde{\Omega}_{p}^\ast}}\mathcal{I}(\mathbf{F})$,} where $P_{Succ, \widetilde{\Omega}_{p}^\ast}$ is the average successful transmission probability under  $\widetilde{\Omega}_{p}^\ast$. Hence, the   condition in (\ref{sadsaxxxx}) is equivalent to $RP_{Succ, \widetilde{\Omega}_{p}^\ast}  >\mathcal{I}(\mathbf{F})$, where $R P_{Succ, \widetilde{\Omega}_{p}^\ast} $ is the average number of bits  that are successfully delivered to the controller per channel use. Therefore,  our sufficient condition is consistent with the classical results for error-free channels \cite{rate1}, \cite{rate2}  after accounting for the SER.~\hfill~\IEEEQED
\end{Remark}

Finally, the approximate value function in (\ref{approxvaluefunc}) satisfies $\widetilde{V}_{\eta_{th}}(\boldsymbol{\chi})=\mathcal{O}(\left\|\boldsymbol{\Delta}\right\|^2)$. Based on Theorem \ref{perfgap}, the performance of the NCS  under $\widetilde{\Omega}_p^\ast$  in (\ref{powersolu}) (i.e., $\widetilde{\theta}^\ast $) is order-optimal, i.e., $\widetilde{\theta}^\ast - \theta^{\ast} =\mathcal{O}(\tau)+\mathcal{O}\big(\frac{\tau^2}{2^{2R_{min}}}\big)$ as $\tau \rightarrow 0$.  \txtblue{We give the necessary conditions for NCS stability as follows.
\begin{Theorem}	[Necessary Conditions for NCS Stability]	\label{stabasdsa}
	Using the  power control policy $\widetilde{\Omega}_p^\ast$ in (\ref{powersolu}) and the plant control policy  $\Omega_{\mathbf{u}}^\ast$ in (\ref{statsrr}), the NCS is stable (i.e., $\lim_{t \rightarrow \infty}\mathbb{E}^{\Omega_{\mathbf{u}}^\ast, \widetilde{\Omega}_p^\ast}[\left\|\mathbf{x}(t)\right\|^2 ]< \infty$, $\lim_{t \rightarrow \infty}\mathbb{E}^{\Omega_{\mathbf{u}}^\ast, \widetilde{\Omega}_p^\ast}[\left\|\mathbf{L}(t)\right\|^2 ] < \infty$, and $\lim_{t \rightarrow \infty}\mathbb{E}^{\Omega_{\mathbf{u}}^\ast, \widetilde{\Omega}_p^\ast}[\left\|\boldsymbol{\Delta}(t)\right\|^2 ]< \infty$) only if 
	\bs\begin{align}	\label{necss3}
		&\frac{ p_{max} \tau}{\kappa(R) B_W+p_{max} \tau} > \\
		&\min\Big\{ \frac{1}{R}\sum_{\mu_i(\mathbf{F})} \max\left\{0, \log|\mu_i(\mathbf{F})|\right\}, 1-\frac{1}{\mu_{max}(\mathbf{F}^T\mathbf{F})}\Big\}\notag 
	\end{align}\bsc
\end{Theorem}
\begin{proof}
	please refer to Appendix I.
\end{proof}}

\txtblue{From Theorem \ref{stab} and Theorem \ref{stabasdsa}, the tightness of the sufficient condition compared with the necessary condition depends on the difference between $\frac{1}{R}\mathcal{I}(\mathbf{F})$ and $1-\frac{1}{\mu_{max}(\mathbf{F}^T\mathbf{F})}$. For NCSs with $\frac{1}{R}\mathcal{I}(\mathbf{F})\approx1-\frac{1}{\mu_{max}(\mathbf{F}^T\mathbf{F})}$,  the sufficient condition is tight compared with the necessary condition.}
\txtblue{\begin{Remark}[Extension to Plant Output Feedback]	\label{rqads11dsdds}
	The solution framework can be easily extended to the case of plant output feedback, i.e., the input of the sensor is $\mathbf{s}(t)=\mathbf{C}\mathbf{x}(t)$, where $\mathbf{s}\in \mathbb{R}^{l}$ and $\mathbf{C}=\mathbb{R}^{l\times d}$. Using Prop. 5.1 of \cite{rate1}, the encoder has access to a Luenberger-like observer $\overline{\mathbf{x}}(t+1)=\mathbf{F}\overline{\mathbf{x}}(t)+\mathbf{G}\mathbf{u}(t)+\mathbf{M}\big(\mathbf{s}(t)-\mathbf{C}\big(\mathbf{F}\overline{\mathbf{x}}(t)+\mathbf{G}\mathbf{u}(t)\big)\big)$, where $\mathbf{M}$ is chosen such that $\big(\mathbf{I}-\mathbf{M}\mathbf{C}\big)\mathbf{A}$ is stable. Let $\overline{\mathbf{e}}(t)\triangleq \mathbf{x}(t)-\overline{\mathbf{x}}(t)$ and it can be shown\footnote{Please refer to Appendix J on the related proofs regarding the extension.} that $\|\overline{\mathbf{e}}(t)\|\leq C$ for some constant $C>0$. Therefore, we can obtain an upper bound of the per-stage state estimation error cost in (\ref{asdsddpers}) as follows:
\begin{align}
	\boldsymbol{\Delta}^T(t)\mathbf{S} \boldsymbol{\Delta}(t) \leq 4 C^2 \mu_{max}(\mathbf{S})+ 2 \overline{\boldsymbol{\Delta}}^T(t)\mathbf{S} \overline{\boldsymbol{\Delta}}(t)	\label{2388sdssads11}
\end{align}
where $\overline{\boldsymbol{\Delta}}(t)\triangleq \overline{\mathbf{x}}(t)-\mathbb{E}\big[\overline{\mathbf{x}}(t)\big|I_C(t)\big]$. Since the sensor cannot observe perfect plant state, the state estimation error  $\boldsymbol{\Delta}$ is not available at the sensor. Therefore, instead of optimizing  the average state estimation error cost  $ \limsup_{T\rightarrow \infty}  \frac{1}{T}\sum_{t=0}^{T-1}\mathbb{E} [\boldsymbol{\Delta}(t)^T\mathbf{S}\boldsymbol{\Delta}(t)]$, we optimize the following upper bound  based on (\ref{2388sdssads11}):
\begin{align}	\label{asd1csor}
	\limsup_{T\rightarrow \infty}  \frac{1}{T}\sum_{t=0}^{T-1}\mathbb{E} [\overline{\boldsymbol{\Delta}}(t)^T\mathbf{S}\overline{\boldsymbol{\Delta}}(t)]
\end{align}
Furthermore, we write the dynamics of $\overline{\mathbf{x}}(t)$ as follows:
\bs\begin{align}
	\overline{\mathbf{x}}(t+1)=\mathbf{F}\overline{\mathbf{x}}(t)+\mathbf{G}\mathbf{u}(t)+\mathbf{M}\mathbf{C} \big(\mathbf{F}\overline{\mathbf{e}}(t) + \mathbf{w}(t)\big)	\label{asdsad21sds}
\end{align}\bsc where $\mathbf{M}\mathbf{C} \big(\mathbf{F}\overline{\mathbf{e}}(t) + \mathbf{w}(t)\big)$ can be treated as the disturbance for $\overline{\mathbf{x}}(t)$, and it is bounded since both $\overline{\mathbf{e}}(t)$ and $\mathbf{w}(t)$ are bounded. Therefore, the optimization problem for plant output feedback fits into the proposed framework, where the average state estimation cost in (\ref{asd1csor}) corresponds to (\ref{perstagecost}), and the   dynamics in (\ref{asdsad21sds}) with bounded disturbance correspond to (\ref{plantain}).~\hfill~\IEEEQED
\end{Remark}}

\vspace{-0.2cm}

\section{Simulations}	\label{simsec}

In this section, we compare the  performance  of the proposed   power control scheme  in Theorem \ref{thmpower}  with   four baselines. Baseline 1 refer to a \emph{fixed power control (FPC)}, where  $p(t)=p_0$ for a fixed power $p_0>0$. Baseline 2 refer to a \emph{CSI-only power control (COPC)}, where $p(t)=\big\{\frac{\lambda}{a\alpha(t-1)}, p_{max}\big\}$ where  $a\alpha(t-1)=\mathbb{E}[\alpha(t)|\alpha(t-1)]$ is the  CSI estimation based on the previous-slot CSI and $\lambda$ is a tradeoff parameter. Baseline 3 refer to a \emph{power control for error-free channel (PCEFC)} \cite{onoff2}, where the sensor minimizes an average  weighted state estimation error and an average number of channel uses under error-free channel. Baseline 4 refers to a \emph{power control for I.I.D. channel with special information structure (PCICSIS)} \cite{mdp1}, where the sensor minimizes an average weighted state estimation error and the average power cost under i.i.d. fading  channel.  The power control  depends on $\{\boldsymbol{\Theta}(t-1), \alpha(t-1)\}$, where $\boldsymbol{\Theta}(t)\triangleq {\mathbf{F}}\boldsymbol{\Delta}(t)+\mathbf{w}(t)$. The solutions for Baseline 3 and  4 are obtained using the brute-force VIA \cite{mdpcref2}, \cite{solvebellman}.  We consider an NCS with parameters: $\widetilde{\mathbf{F}}=\bigl( \begin{smallmatrix} 
  -1 & -2\\
  3 & -4 
\end{smallmatrix} \bigr)$, $\widetilde{\mathbf{G}}=\text{diag}(2,1)$, $\widetilde{\mathbf{W}}=\text{diag}(1,1)$, $\widetilde{w}_{max}=1$, $\mathbf{Q}=\text{diag}(1,1)$, $\mathbf{D}=\text{diag}(1,2)$, $a=-5$, $B_W=1$, $R=4$, $p_{max}=160$, and $\tau=0.05$.

\begin{figure}
  \centering
  \includegraphics[width=2.5in]{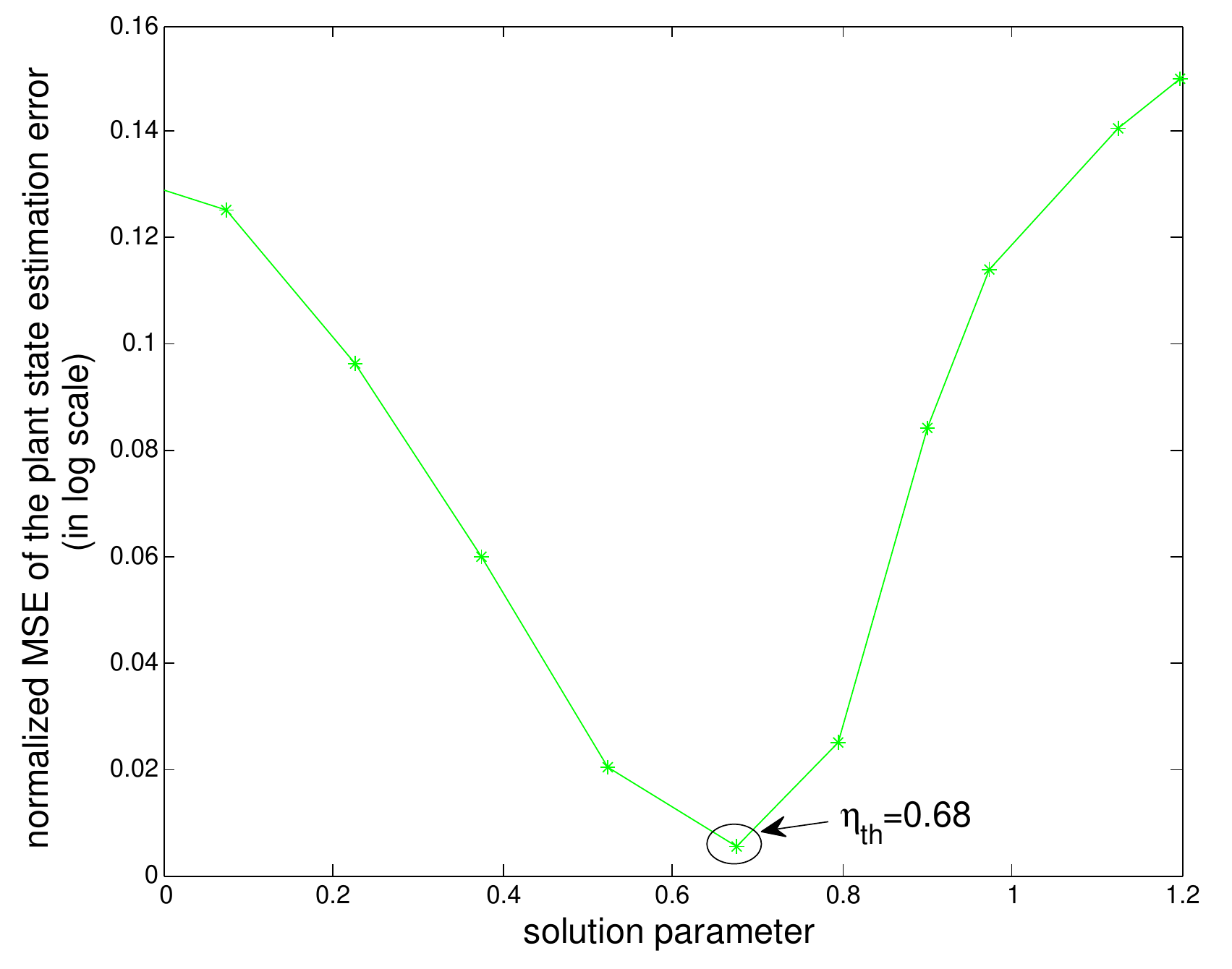}\vspace{-0.2cm}
  \caption{Normalized MSE of the plant state estimation  versus $\eta_{th}$ at  power cost 14dB.}
  \label{compareees2} \vspace{-0.2cm}
\end{figure}

\begin{figure}
\centering
  \includegraphics[width=2.5in]{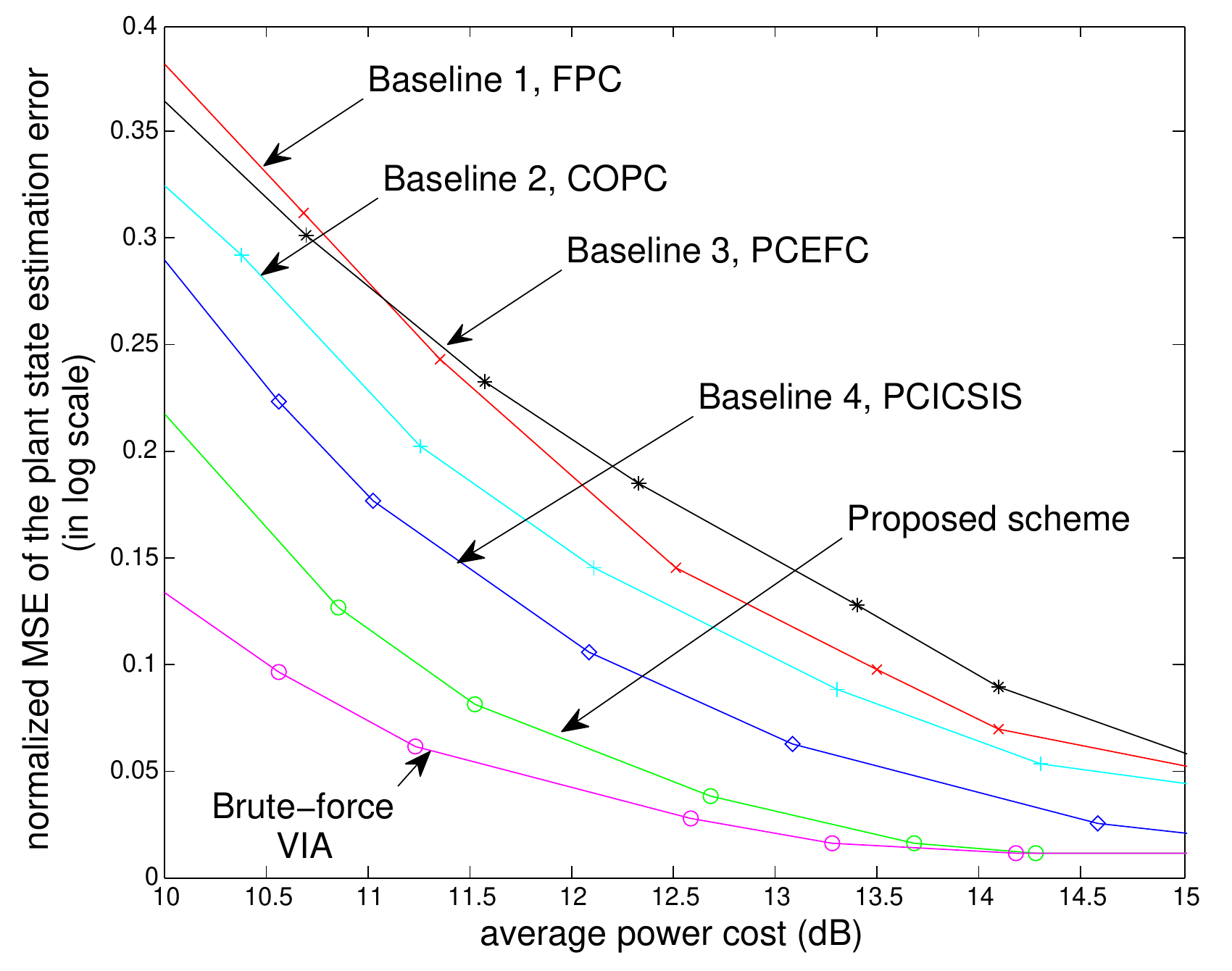}\vspace{-0.2cm}
  \caption{Normalized MSE  of the plant state estimation average power cost.}
  \label{comparsdeees2}\vspace{-0.5cm}
\end{figure}

Fig.~\ref{compareees2} illustrates the normalized MSE of the plant state estimation versus   $\eta_{th}$ at  average power cost 14dB under  our proposed scheme.  The normalized MSE achieves the minimum when $\eta_{th}$ is around 0.68. Therefore, we choose $\eta_{th}= 0.68$  when the average power cost is 14dB. The optimal choices of $\eta_{th}$ at other average power costs can be obtained using similar methods. Fig. \ref{comparsdeees2} illustrates the normalized MSE of the plant state estimation versus the average power cost. Our proposed scheme achieves  significant performance gain  compared with all the baselines and the performance of the proposed scheme is very close to that of the brute-force  VIA \cite{mdpcref2}.  Table \ref{tables} illustrates the comparison of the  computational time of the baselines, the proposed scheme, and the brute-force VIA \cite{mdpcref2}.  Our proposed solution has very small computational cost   due to the closed-form approximate value function in (\ref{approxvaluefunc}) and also outperforms the baselines.

\begin{table}[t]
	\centering
\begin{tabular}{|c | c c c c |}
	\hline
		  {Dimension of $\mathbf{x}$}& \multicolumn{1}{c|}{2}  & \multicolumn{1}{c|}{4} & \multicolumn{1}{c|}{6}  & \multicolumn{1}{c|}{8}    \\
	\hline
		 {BL 1 and  2} &  \multicolumn{4}{c|}{0.0034ms}   			\\ \hline
		 {BL 3} &   \multicolumn{1}{c|}{1.9s} &   \multicolumn{1}{c|}{3048.7s} &    \multicolumn{2}{c|}{$> 10^5$s}   	 \\ \hline
	            {BL 4} &   \multicolumn{1}{c|}{76.2s} &  \multicolumn{3}{c|} {$> 10^5$s} 	\\	\hline
	            {Proposed scheme} &     \multicolumn{1}{c|}{$0.0924$s} &   \multicolumn{1}{c|}{$0.1696$s}  &    \multicolumn{1}{c|}{$0.3316$s} &     \multicolumn{1}{c|}{$0.6284$s}	\\\hline
	             {VIA} &     \multicolumn{1}{c|}{108.6s} &   \multicolumn{3}{c|} {$> 10^5$s} 	\\
	\hline
\end{tabular}
	\caption{ {Comparison of the MATLAB computational time of the baselines, the proposed scheme, and the brute-force  VIA.}}	
		\label{tables}\vspace{-1cm}
\end{table}

\vspace{-0.2cm}

\section{Conclusion}
In this paper, we propose a low complexity   power control scheme for the NCS by solving a  weighted average state estimation error and communication power minimization  problem.  Using a continuous-time perturbation approach, we derive a closed-form approximate value  function and a low complexity  power control scheme.  The proposed  solution   is shown to have an event-driven structure with a dynamically changing  threshold. We  establish the  conditions for asymptotic optimality, and also give sufficient conditions for ensuring the NCS stability under our proposed scheme.  Numerical results show that the proposed scheme has low complexity and much better performance compared with  the baselines.

\vspace{-0.5cm}
\txtblue{\section*{Appendix A: Proof of Lemma \ref{lemmaqunt}} }

\hspace{-0.3cm} \emph{\txtblue{A. Analysis of the Primitive Quantizer}} 
\vspace{0.1cm}

\txtblue{Following \cite{rate1} and \cite{ratenoise}, we  define the following constant matrices. Let $\boldsymbol{\Phi}$ be a nonsingular matrix and $\boldsymbol{\Upsilon}\in \mathbb{R}^{d\times d}$ such that $\boldsymbol{\Phi}\mathbf{F}\boldsymbol{\Phi}^{-1}=\boldsymbol{\Upsilon}=\text{diag}(\mathbf{J}_1,\dots,\mathbf{J}_m)$, where each $\mathbf{J}_j$ is a Jordan block of dimension $d_j$ (and $d_j$ is the geometric multiplicity associated with the eigenvalue $\lambda_j$). Define $\mathbf{H}=\text{diag}(\mathbf{H}_1,\dots, \mathbf{H}_m)$ and $\boldsymbol{\Gamma}=\text{diag}(\boldsymbol{\Gamma}_1,\dots, \boldsymbol{\Gamma}_m)$, where $ \mathbf{H}_j=\mathbf{I}$ if $\mathbf{J}_j$ and $\boldsymbol{\Gamma}_j=\Bigg( \begin{smallmatrix} 
  |\lambda_j| & 1,                    &    & \\
                          & |\lambda_j| & 1 &  \\
                           &                        & \ddots & 1 \\
                           & & & |\lambda_j|         
\end{smallmatrix} \Bigg)$ if they are associated with real eigenvalue $\lambda_j$, and $\mathbf{H}_j=\text{diag}(r(\theta)^{-1},\dots, r(\theta)^{-1})$ with $r(\theta)=\Big( \begin{array}{cc}
		\cos(\theta),   \ \sin(\theta)    \\
		-\sin(\theta),   \cos(\theta)
	\end{array} \Big)$ and $\boldsymbol{\Gamma}_j=\Bigg( \begin{smallmatrix} 
  |\lambda_j| & 1,                    &    & \\
                          & |\lambda_j| & 1 &  \\
                           &                        & \ddots & 1 \\
                           & & & |\lambda_j|         
\end{smallmatrix} \Bigg)$ if $\mathbf{J}_j$ if they are associated with complex  eigenvalues $\rho(\cos(\theta)\pm\sin(\theta))$. Define $\mathbf{F}_{\mathbf{R}}=\text{diag}(2^{-R_1},\dots, 2^{-R_d})$. We summarize  the working flow of the primitive quantizer as follows:}
\txtblue{\begin{Algorithm}	[Dynamic Update of  the Primitive Quantizer]	 \label{algprimquant}
At each time slot $t$, the primitive quantizer takes the plant state $\mathbf{x}(t)$ as input and does the following:
	\begin{itemize}
		\item  \textbf{Step 1, Initialization: } If $t=0$, initialize the quantizer as follows: $\widetilde{\mathbf{x}}(0)=\mathbf{0}$, $\boldsymbol{\Psi}(0)=\boldsymbol{\Phi}$, $\mathbf{L}(0)=\|\boldsymbol{\Phi}\|L$, and  go to Step 4. Otherwise, go to Step 2. 
		\item  \textbf{Step 2, Update of the Shifting Vector and Data Pre-Processing:}
		\bs\begin{align}
			\hspace{-0.5cm}\widetilde{\mathbf{x}}(t)=\left\{
				\begin{aligned}	 \label{errordynfake}
				&  \mathbf{F}\widetilde{\mathbf{x}}(t-1)+\mathbf{G}\mathbf{u}(t-1),  \   \text{if }	\gamma(t-1)=0	   	\\
				&  \mathbf{F}\left(\boldsymbol{\Psi}^{-1}(t-1)\boldsymbol{\xi}(t-1)+\widetilde{\mathbf{x}}(t-1)\right)  \\
				&    \hspace{1cm}+\mathbf{G}\mathbf{u}(t-1),  \ \text{if }	 \gamma(t-1)=1
	   			\end{aligned}
  						  \right.
		\end{align}\bsc  and $\widetilde{\mathbf{x}}(0)=\mathbf{x}_0$.  Calculate innovation  $\mathbf{i}(t)=\mathbf{x}(t)-\widetilde{\mathbf{x}}(t)$.		
		\item  \textbf{Step 3, Update of the Coordinate Transformation Matrix and Dynamic Range:}
		\bs\begin{align}	\label{27equa}
			\boldsymbol{\Psi}(t)=\mathbf{H}\boldsymbol{\Psi}(t-1)
		\end{align}
		\begin{align} \label{28equa}
		\mathbf{L}(t)=\left\{
				\begin{aligned}	 
				&  \boldsymbol{\Gamma}\mathbf{L}(t-1)+w_{max}\|\boldsymbol{\Psi}(t)\|\mathbf{1},   \quad \text{if }	\gamma(t-1)=0	   	\\
				&  \boldsymbol{\Gamma}\mathbf{F}_{\mathbf{R}}\mathbf{L}(t-1)+w_{max}\|\boldsymbol{\Psi}(t)\|\mathbf{1},     \text{if }	 \gamma(t-1)=1
	   			\end{aligned}
  						  \right.	
		\end{align}\bsc where \bs$w_{max}=\widetilde{w}_{max}\sqrt{ \int_0^\tau\int_0^\tau\mu_{max}\big(\exp( \widetilde{\mathbf{F}}^Ts + \widetilde{\mathbf{F}}^Ts' ) \big)d s ds'}  \\ =\mathcal{O}(\tau)$\bsc and $\|\mathbf{w}(t)\|\leq w_{max}$ according to (\ref{plantain}).
		\item  \textbf{Step 4, Determination of the Output Symbol:} The quantizer partitions the following region into $2^{2R}$ boxes:
		\begin{align}
			&\mathcal{D}_{\widetilde{\mathbf{x}}}(t)=\big\{\mathbf{x}\in \mathbb{R}^d: \boldsymbol{\Psi}(t)\big(\mathbf{x}-\widetilde{\mathbf{x}}(t)\big) \notag \\
			&\in \left\{[-L_1(t),L_1(t)]\times \cdots \times  [-L_d(t),L_d(t)]\right\}\big\}\label{dregion}
		\end{align}
		with side lengths $\big\{\frac{2L_n(t)}{2^{R_n}}: \forall n\big\}$. Determine which box   $\mathbf{i}(t)$ falls into within the region $\mathcal{D}_{\widetilde{\mathbf{x}}}(t)$, and the quantizer outputs $\boldsymbol{\xi}(t)$, which is the centroid of that box. If  $\mathbf{i}(t)$ falls outside the region $\mathcal{D}_{\widetilde{\mathbf{x}}}(t)$, the quantizer outputs a special symbol representing an overflow\footnote{\txtblue{We will show in Lemma \ref{importapp} that $\mathbf{x}(t)$ never leaves $\mathcal{D}_{\widetilde{\mathbf{x}}}(t)$.}}.	~\hfill~\IEEEQED
	\end{itemize}
\end{Algorithm}}

\hspace{-0.3cm} \emph{\txtblue{B. Properties  of the Primitive Quantizer in Algorithm \ref{algprimquant}}}
\vspace{0.1cm}

\txtblue{Based on  Algorithm \ref{algprimquant}, we can write the relationship between $\mathbf{x}(t)$ and  $\boldsymbol{\xi}(t)$ as follows: 
\begin{align}
	\boldsymbol{\xi}(t)= \boldsymbol{\Psi}(t)\big(\mathbf{x}(t)-\widetilde{\mathbf{x}}(t)\big)+\mathbf{e}(t)	\label{modelinapp}
\end{align}
where $\mathbf{e}(t)$ is the quantization noise. From (\ref{plantain}), we know that $\mathbf{x}(t)$ is symmetrically distributed about the origin. Together with Algorithm \ref{algprimquant}, the primitive quantizer behaves like a multi-dimensional uniform quantizer \cite{quantmul}. Therefore, the quantization noise $\mathbf{e}(t)$ has  zero mean. We obtain the dynamics of  $\hat{\mathbf{x}}(t)$ and $\widetilde{\mathbf{x}}(t)$ as follows  to show that $\mathbf{x}(t)$  never leaves  $\mathcal{D}_{\widetilde{\mathbf{x}}}(t)$.}

\txtblue{\emph{1) Dynamics of $\hat{\mathbf{x}}(t)=\mathbb{E}\big[\mathbf{x}(t)\big|I_C(t)\big]$:} If $\gamma(t)=0$,
\begin{align}
	\hat{\mathbf{x}}(t)&\overset{(a)}{=} \mathbb{E}\big[\mathbf{x}(t)\big| \mathbf{u}(t-1), \mathbf{u}_0^{t-2}, \alpha_0^{t}, \gamma_0^{\nu(t-1)}, \mathbf{y}_0^{\nu(t-1)}\big]	\notag \\
	&\overset{(b)}= \mathbb{E}\big[\mathbf{F}\mathbf{x}(t-1)+\mathbf{w}(t-1)\big| I_C(t-1)\big]+\mathbf{G}\mathbf{u}(t-1)	\notag \\
	&\overset{(c)}= \mathbf{F} \hat{\mathbf{x}}(t-1)+\mathbf{G}\mathbf{u}(t-1)	\label{deri1}
\end{align}
where $(a)$ is  according to  the definition of $I_C(t)$ in (\ref{controlis}), (b) is because $\mathbf{u}(t-1)$ is a function of $I_C(t-1)$ and $\{\mathbf{x}(t-1),\mathbf{w}(t-1)\}$ is independent of $\alpha(t)$, and $(c)$ is because $\mathbb{E}\big[\mathbf{w}(t-1)\big|I_C(n-1)\big]=0$. If $\gamma(t)=1$,
\begin{align}
	\hat{\mathbf{x}}(t) &  \overset{(d)}{=} \mathbb{E}\big[\boldsymbol{\Psi}^{-1}(t)\big(\boldsymbol{\xi}(t)-\mathbf{e}(t)\big)+\widetilde{\mathbf{x}}(t)\big| \notag \\
			& \hspace{1cm }\boldsymbol{\xi}(t), \widetilde{\mathbf{x}}(t), \mathbf{u}_0^{t-1}, \alpha_0^{t}, \gamma_0^{\nu(t-1)}, \mathbf{y}_0^{\nu(t-1)}\big]\notag \\
			& \overset{(e)}{=} \boldsymbol{\Psi}^{-1}(t)\boldsymbol{\xi}(t)+\widetilde{\mathbf{x}}(t)    \label{deri2}	
\end{align}
where $(d)$ is due to the successful transmission of $\boldsymbol{\xi}(t)$ and $(e)$ is due to the zero mean quantization noise $\mathbf{e}(t)$ and  the fact that $\widetilde{\mathbf{x}}(t)$ can be locally obtained  at the controller. }

 \txtblue{Therefore, 
\begin{align}	
\hat{\mathbf{x}}(t)=
 \left\{
	\begin{aligned}	 \label{errordyn}
		& \mathbf{F}\hat{\mathbf{x}}(t-1)+\mathbf{G}\mathbf{u}(t-1),  \quad &\text{if }	 \gamma(t)=0   	\\
		&  \boldsymbol{\Psi}^{-1}(t)\boldsymbol{\xi}(t)+\widetilde{\mathbf{x}}(t)  , \quad &\text{if }	\gamma(t)=1	
	   \end{aligned}
   \right.
 	 \end{align}
where $\hat{\mathbf{x}}(0)=\mathbf{x}_0$ if $\gamma(0)=0$.}

\txtblue{\emph{2) Dynamics of $\widetilde{\mathbf{x}}(t)=\mathbb{E}\big[\mathbf{x}(t)\big|I_{SC}(t)\big]$:} If $\gamma(t-1)=0$, 
\begin{align}
	\widetilde{\mathbf{x}}(t)&= \mathbb{E}\big[\mathbf{F} \mathbf{x}(t-1)+\mathbf{w}(t-1)\big|I_{SC}(t-1)\big]+\mathbf{G}\mathbf{u}(t-1)	\notag \\
	&= \mathbf{F} \widetilde{\mathbf{x}}(t-1)+\mathbf{G}\mathbf{u}(t-1)   \label{37123sds}
\end{align}
where (\ref{37123sds}) follows  (b) and (c) in (\ref{deri1}). If $\gamma(t-1)=1$,
\bs\begin{align}
	\widetilde{\mathbf{x}}(t)&= \mathbb{E}\big[\mathbf{F} \mathbf{x}(t-1)+\mathbf{w}(t-1)\big|\boldsymbol{\xi}(t-1), \widetilde{\mathbf{x}}(t-1),\notag \\
	&  \mathbf{u}_0^{t-1}, \alpha_0^{t-1},   \gamma_0^{\nu(t-2)}, \mathbf{y}_0^{\nu(t-2)}\big]+\mathbf{G}\mathbf{u}(t-1)	\notag \\
	& =\mathbf{F} \big(\boldsymbol{\Psi}^{-1}(t-1)\boldsymbol{\xi}(t-1)+\widetilde{\mathbf{x}}(t)\big)+\mathbf{G}\mathbf{u}(t-1)	\label{37asda121}
\end{align}\bsc where (\ref{37asda121}) follows the calculations in (\ref{deri2}).} 

\txtblue{\emph{3) Dynamic Range of $\mathbf{x}(t)$:}  We prove the following lemma:
\begin{Lemma}	\label{importapp}
	At each time slot $t$, we have $\mathbf{x}(t)\in \mathcal{D}_{\widetilde{\mathbf{x}}}(t)$, where $\mathcal{D}_{\widetilde{\mathbf{x}}}(t)$ is given in (\ref{dregion}).
\end{Lemma}
\emph{\qquad Proof:}
We define the following range:
\begin{align}
	\mathcal{E}_{\hat{\mathbf{x}}}(t)= &\big\{\mathbf{x}\in \mathbb{R}^d: \boldsymbol{\Psi}(t)\big(\mathbf{x}-\hat{\mathbf{x}}(t)\big) \label{exasdssdef} \\
	&\in \left\{[-\Lambda_1(t),\Lambda_1(t)]\times \cdots \times  [-\Lambda_d(t),\Lambda_d(t)]\right\}\big\}	\notag
\end{align}
where  	$\boldsymbol{\Lambda}(t)=\big(\Lambda_1(t), \dots, \Lambda_d(t)\big)^T \in \mathbb{R}^d$ with the following dynamics:
\bs\begin{align}\notag 
			\boldsymbol{\Lambda}(t) =\left\{
				\begin{aligned}	
				&  \mathbf{F}_{\mathbf{R}}\mathbf{L}(t),  \hspace{2.9cm}  \text{if }	\gamma(t)=1	   	\\
				&  \mathbf{L}(t),   \hspace{3.4cm}    \text{if }	 \gamma(t)=0, \gamma(t-1)=1\\
				& \boldsymbol{\Gamma}\boldsymbol{\Lambda}(t-1)+w_{max}\|\boldsymbol{\Psi}(t)\|\mathbf{1},    \text{if }	 \gamma(t)=0, \gamma(t-1)=0
	   			\end{aligned}
  						  \right.
\end{align}\bsc and we let $\gamma(-1)=1$.We will show in sequence:  if $\mathbf{x}(t)\in \mathcal{D}_{\widetilde{\mathbf{x}}}(t)$, then $\mathbf{x}(t)\in \mathcal{E}_{\hat{\mathbf{x}}}(t)$, and if $\mathbf{x}(t)\in \mathcal{E}_{\hat{\mathbf{x}}}(t)$, then $\mathbf{x}(t+1)\in \mathcal{D}_{\widetilde{\mathbf{x}}}(t+1)$. Using induction method,  the initial condition of the primitive quantizer in Algorithm \ref{algprimquant} ensures that $\mathbf{x}(0)\in \mathcal{D}_{\widetilde{\mathbf{x}}}(0)$ (according to Prop. 5.1 of \cite{rate1}), and then, \newline
\emph{a. $\mathbf{x}(t)\in \mathcal{D}_{\widetilde{\mathbf{x}}}(t) \Rightarrow \mathbf{x}(t)\in \mathcal{E}_{\hat{\mathbf{x}}}(t)$:} 
\begin{itemize}
	\item Case 1 ($\gamma(t)=1$): Using  (\ref{modelinapp}) and (\ref{errordyn}), we have  $\boldsymbol{\Psi}(t)\ (\mathbf{x}(t)-\hat{\mathbf{x}}(t)\ )=-\mathbf{e}(t)$.  If $\mathbf{x}(t)\in \mathcal{D}_{\widetilde{\mathbf{x}}}(t) $, we have $|e_n(t)|\leq \frac{L_n(t)}{2^{-R_n}}$ for all $n$. Therefore, we have $\mathbf{x}(t)\in \mathcal{E}_{\hat{\mathbf{x}}}(t)$ with $\boldsymbol{\Lambda}(t)=\mathbf{F}_{\mathbf{R}}\mathbf{L}(t)$. 
	\item Case 2 ($\gamma(t)=0$, $\gamma(t-1)=1$): Using (\ref{errordynfake}) and (\ref{errordyn}), we have $\boldsymbol{\Psi}(t)\ (\mathbf{x}(t)-\hat{\mathbf{x}}(t)\ )=\boldsymbol{\Psi}(t)\ (\mathbf{x}(t)-\widetilde{\mathbf{x}}(t)\ )$. Therefore, if $\mathbf{x}(t)\in \mathcal{D}_{\widetilde{\mathbf{x}}}(t) $, then $\mathbf{x}(t)\in \mathcal{E}_{\hat{\mathbf{x}}}(t)$ with $\boldsymbol{\Lambda}(t)=\mathbf{L}(t)$.
	\item Case 3 ($\gamma(t)=0$, $\gamma(t-1)=0$): Using (\ref{plantain}) and (\ref{errordyn}), we have $\boldsymbol{\Psi}(t)\ (\mathbf{x}(t)-\hat{\mathbf{x}}(t)\ )=\boldsymbol{\Psi}(t)\ (\mathbf{F}\ (\mathbf{x}(t-1)-\hat{\mathbf{x}}(t-1)\ )+\mathbf{w}(t-1)\ )=\mathbf{H}\mathbf{\Upsilon}\ (\boldsymbol{\Psi}(t-1)\ (\mathbf{x}(t-1)-\hat{\mathbf{x}}(t-1)\ )\ )+\boldsymbol{\Psi}(t)\mathbf{w}(t-1)$, where the second equality is due to Lemma 4.1 of \cite{rate1}. If $\mathbf{x}(t)\in \mathcal{D}_{\widetilde{\mathbf{x}}}(t) $, using Lemma 4.2 of \cite{rate1}, then $\mathbf{x}(t)\in \mathcal{E}_{\hat{\mathbf{x}}}(t)$ with $\boldsymbol{\Lambda}(t)=\boldsymbol{\Gamma}\boldsymbol{\Lambda}(t-1)+w_{max}\|\boldsymbol{\Psi}(t)\|\mathbf{1}$.
\end{itemize}
\emph{b. $\mathbf{x}(t)\in \mathcal{E}_{\hat{\mathbf{x}}}(t) \Rightarrow \mathbf{x}(t+1)\in \mathcal{D}_{\widetilde{\mathbf{x}}}(t+1)$:} 
\begin{itemize}
	\item Case 1 ($\gamma(t)=1$): Using (\ref{plantain}), (\ref{errordynfake}),  (\ref{errordyn}) and Lemma 4.1 of \cite{rate1}, we have $\boldsymbol{\Psi}(t+1) (\mathbf{x}(t+1)-\widetilde{\mathbf{x}}(t+1) )=\mathbf{H}\mathbf{\Upsilon} (\boldsymbol{\Psi}(t) (\mathbf{x}(t)-\hat{\mathbf{x}}(t) ) )+\boldsymbol{\Psi}(t+1)\mathbf{w}(t)$. If $\mathbf{x}(t)\in \mathcal{E}_{\hat{\mathbf{x}}}(t) $, using Lemma 4.2 of \cite{rate1}, then $\mathbf{x}(t+1)\in \mathcal{D}_{\widetilde{\mathbf{x}}}(t+1)$ with $\mathbf{L}(t+1)=\boldsymbol{\Gamma}\boldsymbol{\Lambda}(t)+w_{max}\|\boldsymbol{\Psi}(t+1)\|\mathbf{1}=\boldsymbol{\Gamma}\mathbf{F}_{\mathbf{R}}\mathbf{L}(t)+w_{max}\|\boldsymbol{\Psi}(t+1)\|\mathbf{1}$. This is the evolution of $\mathbf{L}(t)$ for successful transmission  in (\ref{28equa}).
	\item Case 2 ($\gamma(t)=0$): Using (\ref{plantain}), (\ref{errordynfake}) and Lemma 4.1 of \cite{rate1}, we have $\boldsymbol{\Psi}(t+1) (\mathbf{x}(t+1)-\widetilde{\mathbf{x}}(t+1) )=\mathbf{H}\mathbf{\Upsilon} (\boldsymbol{\Psi}(t) (\mathbf{x}(t)-\widetilde{\mathbf{x}}(t) ) )+\boldsymbol{\Psi}(t+1)\mathbf{w}(t)$. If $\mathbf{x}(t)\in \mathcal{E}_{\hat{\mathbf{x}}}(t) $, using Lemma 4.2 of \cite{rate1}, and then $\mathbf{x}(t+1)\in \mathcal{D}_{\widetilde{\mathbf{x}}}(t+1)$ with $\mathbf{L}(t+1)=\boldsymbol{\Gamma}\boldsymbol{\Lambda}(t)+w_{max}\|\boldsymbol{\Psi}(t+1)\|\mathbf{1}=\boldsymbol{\Gamma}\mathbf{L}(t)+w_{max}\|\boldsymbol{\Psi}(t+1)\|\mathbf{1}$. This is the evolution of $\mathbf{L}(t)$ for unsuccessful transmission  in (\ref{28equa}).~\hfill~\IEEEQED
\end{itemize}}

\txtblue{Lemma \ref{importapp}  proves the first property in Lemma \ref{lemmaqunt}. Together with the results in  Section 5.1.4 of \cite{quantmul} on the  uniform quantizer , the quantization noise $\mathbf{e}(t)=\big(e_1(t),\dots, e_d(t)\big)^T$  has a uniform distribution, where  $e_n(t)$ is uniformly distributed within the region $\big[-\frac{L_n(t)}{2^{R_n-1}},\frac{L_n(t)}{2^{R_n-1}}\big]$. This proves the second   property in Lemma \ref{lemmaqunt} (where we use $\mathbf{e}(\mathbf{L}(t), t)$ to denote the quantization noise to show its dependence on $\mathbf{L}(t)$). }

%qtnoisedist
\vspace{-0.2cm}
\section*{Appendix B: Proof of Lemma \ref{lemmadual}}

\hspace{-0.3cm} \emph{A. Relationship between the Original NCS and an Autonomous NCS} 
\vspace{0.1cm}

We consider two NCSs. The first NCS is given as follows:
\bs\begin{align}
	&\mathbf{x}(t+1)=\mathbf{F}\mathbf{x}(t)+\mathbf{G}\mathbf{u}(t)+\mathbf{w}(t)\notag \\
	&\boldsymbol{\xi}(t)=Q_{\boldsymbol{\Psi}(t), \mathbf{L}(t)} \big(\mathbf{x}(t)-\widetilde{\mathbf{x}}(t)\big),\quad \mathbf{y}(t) = \gamma(t)\boldsymbol{\xi}(t)	\label{ncs11}
\end{align}\bsc where  $Q_{\boldsymbol{\Psi}(t), \mathbf{L}(t)} (\cdot )$ corresponds to the R.H.S. of the mapper model  in (\ref{quantprop2}). The second NCS  is given as follows with no control actions applied (i.e., an autonomous system):
\bs\begin{align}
	&\overline{\mathbf{x}}(t+1)=\mathbf{F}\overline{\mathbf{x}}(t)+\overline{\mathbf{w}}(t)\notag \\
	&\overline{\boldsymbol{\xi}}(t)=Q_{\overline{\boldsymbol{\Psi}}(t),\overline{\mathbf{L}}(t)} \big(\overline{\mathbf{x}}(t)-\widetilde{\overline{\mathbf{x}}}(t)\big), \quad \overline{\mathbf{y}}(t) = \overline{\gamma}(t)\overline{\boldsymbol{\xi}}(t)	\label{ncs22}
\end{align}\bsc where $\widetilde{\overline{\mathbf{x}}}(t)=\mathbb{E}\big[\mathbf{x}(t)\big|\overline{I}_{SC}(t)\big]$ and $\overline{I}_{SC}(t)\triangleq \{\overline{\alpha}_0^{t-1},   \overline{\gamma}_0^{\nu(t-1)}, \overline{\mathbf{y}}_0^{\nu(t-1)}\}$, and we assume the following parameters are identical in the two NCSs, i.e.,  $\mathbf{x}(0)=\overline{\mathbf{x}}(0)$, $\mathbf{w}(t) = \overline{\mathbf{w}}(t)$, $\boldsymbol{\Psi}(t)=\overline{\boldsymbol{\Psi}}(t)$, $\mathbf{L}(t)=\overline{\mathbf{L}}(t)$, $\alpha(t)=\overline{\alpha}(t)$, $\gamma(t)=\overline{\gamma}(t)$, for all $t$. Then, 
\begin{Lemma}		\label{lemma6}
	For the two NCSs in (\ref{ncs11}) and (\ref{ncs22}), we have $\mathbf{x}(t)-\mathbb{E}\big[\mathbf{x}(t)|I_C(t)\big] =\overline{\mathbf{x}}(t)-\mathbb{E}\big[\overline{\mathbf{x}}(t)|I_C(t)\big] $, $\forall t$.
\end{Lemma}
\begin{proof}
	The linearity of the  dynamics for $\mathbf{x}(t)$ and $\overline{\mathbf{x}}(t)$ implies  the existence of  $\mathbf{A}(t)$, $\mathbf{B}(t)$ and $\mathbf{C}(t)$ such that
\bs\begin{align}
	&\mathbf{x}(t)=\mathbf{A}(t)\mathbf{x}(0)+\mathbf{B}(t)\vec{\mathbf{u}}(t-1)+\mathbf{C}(t)\vec{\mathbf{w}}(t-1)\notag \\
	&\overline{\mathbf{x}}(t)=\mathbf{A}(t)\mathbf{x}(0)+\mathbf{C}(t)\vec{\mathbf{w}}(t-1)
\end{align}\bsc where $\vec{\mathbf{u}}(t)= (\mathbf{u}^T(1), \dots, \mathbf{u}(t) )^T$ and $\vec{\mathbf{w}}(t)= (\mathbf{w}^T(1), \dots, \mathbf{w}(t) )^T$.  Then, we have \bs$\mathbf{x}(t)-\mathbb{E}\big[\mathbf{x}(t)|I_C(t)\big]= \big(\mathbf{A}(t)\mathbf{x}(0)+\mathbf{C}(t)\vec{\mathbf{w}}(t-1)\big)\-\big(\mathbf{A}(t)\mathbb{E}\big[\mathbf{x}(0)|I_C(t)\big]+\mathbf{C}(t)\mathbb{E}\big[\vec{\mathbf{w}}(t-1)|I_C(t)\big]\big)=\overline{\mathbf{x}}(t)-\mathbb{E}\big[\overline{\mathbf{x}}(t)|I_C(t)\big]$\bsc.\end{proof}

\vspace{0.1cm}
\hspace{-0.3cm} \emph{B. State Estimate of an Autonomous System} 
\vspace{0.1cm}

We then prove the following Lemma. Together with Lemma \ref{lemma6}, the no dual effect property of the original NCS can be directly proven.
\begin{Lemma}	\label{lemma7}
	$\mathbb{E}\big[\overline{\mathbf{x}}(t)|I_C(t)\big]=\mathbb{E}\big[\overline{\mathbf{x}}(t)|\alpha_0^{t},   \gamma_0^{\nu(t)}, \overline{\mathbf{y}}_0^{\nu(t)}\big]$, $\forall t$.
\end{Lemma}
\begin{proof}
	We use  $X\rightarrow Y \rightarrow Z$ to denote that $(X, Y, Z)$ forms a Markov chain, i.e., $\Pr[Z|X,Y]=\Pr[Z|Y]$. It can be easily verified that $X\rightarrow Y \rightarrow Z$ implies that $Z\rightarrow Y \rightarrow X$. Using induction method, we will show that  $\overline{\mathbf{y}}(\nu(t))=\mathbf{y}(\nu(t))$ and $\overline{\mathbf{x}}(t)\rightarrow \{\alpha_0^{t}, \gamma_0^{\nu(t)}, \mathbf{y}_0^{\nu(t)}\}\rightarrow \mathbf{u}_0^{t-1}$ for all $t\geq 1$.
	
	\emph{Step 1 ($t=1$):}  If $\gamma(0)=0$, we have $\overline{\mathbf{y}}(\nu(0))=\mathbf{y}(\nu(0))=0$. If $\gamma(0)=1$, we have $\overline{\mathbf{y}}(\nu(0))=Q_{\overline{\boldsymbol{\Psi}}(0),\overline{\mathbf{L}}(0)} \big(\overline{\mathbf{x}}(0)-\widetilde{\overline{\mathbf{x}}}(0)\big)=Q_{{\boldsymbol{\Psi}}(0),{\mathbf{L}}(0)} \big({\mathbf{x}}(0)-\widetilde{{\mathbf{x}}}(0)\big)={\mathbf{y}}(\nu(0))$, where the second equality is due to the construction in part A. Therefore, $\overline{\mathbf{y}}(\nu(0))=\mathbf{y}(\nu(0))$. Since $\mathbf{u}(0) =\sigma (I_C(0))=\sigma(\{\alpha_0, \gamma_0^{\nu(0)}, \mathbf{y}_0^{\nu(0)}\})$, we have $\{\overline{\mathbf{x}}_0, \mathbf{w}(0)\}\rightarrow \{\alpha_0, \gamma_0^{\nu(0)}, \mathbf{y}_0^{\nu(0)}\}  \rightarrow \mathbf{u}(0)$. Together with $\overline{\mathbf{x}}(1)=\mathbf{F}\overline{\mathbf{x}}(0)+\overline{\mathbf{w}}(0)$, we have 
\begin{align}
	\overline{\mathbf{x}}(1)\rightarrow \{\alpha_0^1, \gamma_0^{\nu(0)}, \mathbf{y}_0^{\nu(0)} \}\rightarrow \mathbf{u}(0)	\label{sdadadszs1}
\end{align}
Therefore, we have $\overline{\mathbf{y}}(\nu(1))=\overline{\mathbf{y}}(\nu(0))=\mathbf{y}(\nu(0))=\mathbf{y}(\nu(1))$ if $\gamma(1)=0$. If $\gamma(1)=1$, we have
\bs\begin{align}
	&\overline{\mathbf{y}}(\nu(1))=\overline{\mathbf{y}}(1)=Q_{\boldsymbol{\Psi}(1),{\mathbf{L}}(1)}\big(\overline{\mathbf{x}}(1)-\mathbb{E}\big[\overline{\mathbf{x}}(1) \big| {\alpha}_0^{1},  { \gamma}_0^{\nu(0)}, \mathbf{y}_0^{\nu(0)} \big] \big)\notag \\
	\overset{(a)}{=} &Q_{\boldsymbol{\Psi}(1),{\mathbf{L}}(1)}\big(\overline{\mathbf{x}}(1)-\mathbb{E}\big[\overline{\mathbf{x}}(1) \big|I_{SC}(1) \big] \big) \label{asda2ewe47} \\
	\overset{(b)}{=}& Q_{\boldsymbol{\Psi}(1),{\mathbf{L}}(1)}\big({\mathbf{x}}(1)-\mathbb{E}\big[{\mathbf{x}}(1) \big|I_{SC}(1) \big] \big)=\mathbf{y}(1)=\mathbf{y}(\nu(1))		\notag
\end{align}\bsc where (a) is due to (\ref{sdadadszs1}), and (b) can be verified similarly as in Lemma \ref{lemma6}. Therefore, $\overline{\mathbf{y}}(\nu(1))={\mathbf{y}}(\nu(1))$. Furthermore, since $\overline{\mathbf{y}}(\nu(1))={\mathbf{y}}(\nu(1))=\sigma\big(\alpha_0^1, \gamma_0^{\nu(1)}, \mathbf{y}_0^{\nu(0)}, \overline{\mathbf{x}}(1)\big)$, we have
\begin{align}
	&\Pr[\mathbf{u}(0)| \alpha_0^{1}, \gamma_0^{\nu(1)}, \mathbf{y}_0^{\nu(1)}, \overline{\mathbf{x}}(1)]\notag \\
	=&\Pr[\mathbf{u}(0)| \alpha_0^{1}, \gamma_0^{\nu(1)}, \mathbf{y}_0^{\nu(0)}, \sigma\big(\alpha_0^1, \gamma_0^{\nu(1)}, \mathbf{y}_0^{\nu(0)}, \overline{\mathbf{x}}(1)\big), \overline{\mathbf{x}}(1)]	\notag \\
	\overset{(c)}{=}& \Pr[\mathbf{u}(0)| \alpha_0^{1}, \gamma_0^{\nu(1)}, \mathbf{y}_0^{\nu(0)}, \sigma\big(\alpha_0^1, \gamma_0^{\nu(1)}, \mathbf{y}_0^{\nu(0)}\big)]\label{stepres1}
\end{align}
where $(c)$ is due to (\ref{sdadadszs1}). Therefore, based on (\ref{stepres1}), we have $\overline{\mathbf{x}}(1)\rightarrow \{\alpha_0^{1}, \gamma_0^{\nu(1)}, \mathbf{y}_0^{\nu(1)}\}\rightarrow \mathbf{u}(0)$.

\emph{Step 2 (Induction):} By induction hypothesis, assume that $\overline{\mathbf{y}}(\nu(t))=\mathbf{y}(\nu(t))$ and for $1\leq k \leq t$.	 Note that $\{\overline{\mathbf{x}}(t),\mathbf{w}(t)\}\rightarrow \{\mathbf{u}_0^{t-1}, \alpha_0^{t}, \gamma_0^{\nu(t)}, \mathbf{y}_0^{\nu(t)}\} \rightarrow \mathbf{u}(t)$ because   $\mathbf{u}(t)$ only depends on $I_C(t)$. By the induction hypothesis and $\mathbf{w}(t)$ is independent of $\mathbf{y}_0^{\nu(t)}$ and $\mathbf{u}_0^{t-1}$, we have  $\{\overline{\mathbf{x}}(t), \mathbf{w}(t)\}\rightarrow \{\alpha_0^{t}, \gamma_0^{\nu(t)}, \mathbf{y}_0^{\nu(t)}	\} \rightarrow \mathbf{u}_0^{t-1}$. Combining the above results, we have $\{\overline{\mathbf{x}}(t),\mathbf{w}(t)\}\rightarrow \{\alpha_0^{t}, \gamma_0^{\nu(t)}, y_0^{\nu(t)}	\} \rightarrow \mathbf{u}_0^{t}$.  Together with $\overline{\mathbf{x}}(t+1)=\mathbf{F}\overline{\mathbf{x}}(t)+\overline{\mathbf{w}}(t)$, we have 
\begin{align}
	\overline{\mathbf{x}}(t+1)\rightarrow \{\alpha_0^{t+1}, \gamma_0^{\nu(t)}, \mathbf{y}_0^{\nu(t)}	\} \rightarrow \mathbf{u}_0^{t}	\label{sdad1111}
\end{align}
Therefore, we have $\overline{\mathbf{y}}(\nu(t+1))=\overline{\mathbf{y}}(\nu(t))=\mathbf{y}(\nu(t))=\mathbf{y}(\nu(t+1))$ if $\gamma(t+1)=0$. If $\gamma(t+1)=1$, following the same derivations in  (\ref{asda2ewe47}) and using (\ref{sdad1111}), we have  $\overline{\mathbf{y}}(\nu(t+1))=\overline{\mathbf{y}}(t+1)= Q_{\boldsymbol{\Psi}(t+1),{\mathbf{L}}(t+1)}\big(\overline{\mathbf{x}}(t+1)-\mathbb{E}\big[\overline{\mathbf{x}}(t+1) \big| I_{SC}(t+1) \big] \big) = Q_{\boldsymbol{\Psi}(t+1),{\mathbf{L}}(t+1)}\big({\mathbf{x}}(t+1)-\mathbb{E}\big[{\mathbf{x}}(t+1) \big| I_{SC}(t+1) \big] \big)=\mathbf{y}(\nu(t+1))$.  Therefore, $\overline{\mathbf{y}}(\nu(t+1))={\mathbf{y}}(\nu(t+1))$. Furthermore, since $\overline{\mathbf{y}}(\nu(t+1))={\mathbf{y}}(\nu(t+1))=\sigma\big(\alpha_0^{t+1}, \gamma_0^{\nu(t+1)}, \mathbf{y}_0^{\nu(t)}, \overline{\mathbf{x}}(t+1)\big)$, we have
\begin{align}
	&\Pr[\mathbf{u}_0^t| \alpha_0^{t+1}, \gamma_0^{\nu(t+1)},  \mathbf{y}_0^{\nu(t+1)}, \overline{\mathbf{x}}(t+1)]\notag \\
	=&\Pr[\mathbf{u}_0^t| \alpha_0^{t+1}, \gamma_0^{\nu(t+1)},  \mathbf{y}_0^{\nu(t)}, \notag \\
	 & \sigma\big(\alpha_0^{t+1}, \gamma_0^{\nu(t+1)}, \mathbf{y}_0^{\nu(t)}, \overline{\mathbf{x}}(t+1)\big),  \overline{\mathbf{x}}(t+1)]	\notag \\
	\overset{(d)}{=}&\Pr[\mathbf{u}_0^t| \alpha_0^{t+1}, \gamma_0^{\nu(t+1)},  \mathbf{y}_0^{\nu(t)}, \sigma\big(\alpha_0^{t+1}, \gamma_0^{\nu(t+1)}, \mathbf{y}_0^{\nu(t)}\big)]\label{stepres11}
\end{align}
where $(d)$ is due to (\ref{sdad1111}). Therefore, based on (\ref{stepres11}), we have $\overline{\mathbf{x}}(t+1)\rightarrow \{\alpha_0^{t+1}, \gamma_0^{\nu(t+1)}, \mathbf{y}_0^{\nu(t+1)}\}\rightarrow \mathbf{u}_0^{t}$.\end{proof}

Finally, combining Lemma \ref{lemma6} and Lemma \ref{lemma7}, we have
\bs\begin{align}
	\boldsymbol{\Delta}(t)=\mathbf{x}(t)-\mathbb{E}\big[\mathbf{x}(t)|I_C(t)\big] =\overline{\mathbf{x}}(t)-\mathbb{E}\big[\overline{\mathbf{x}}(t)|\alpha_0^{t},   \gamma_0^{\nu(t)}, \overline{\mathbf{y}}_0^{\nu(t)}\big]\notag
\end{align}\bsc  Therefore, $\boldsymbol{\Delta}(t)$ does not dependent on $\mathbf{u}_0^{t-1}$, which directly induces the no dual effect property in Lemma \ref{lemmadual}.

\vspace{-0.2cm}

\section*{\txtblue{Appendix C:  Proof of Theorem \ref{mdpvarify}}}
	
	\txtblue{For a given  $\Omega_p$, the induced process \bs$\mathbf{Y}(t)=\big\{\mathbf{x}_0^t, \mathbf{i}_0^t, \mathbf{u}_0^{t-1}, \boldsymbol{\xi}_0^{t}, \alpha_0^{t-1},   \gamma_0^{t-1}, \boldsymbol{\Psi}_0^t,\mathbf{L}_0^t \big\} \in \sigma\big(\big\{I_S(s):   0\leq s \leq t\big\}\big)$\bsc  is a controlled Markov process with  transition kernel:
\begin{align}	\label{totaltrans}
	& \Pr\big[\mathbf{Y}(t+1)\big|\mathbf{Y}(t),p(t)\big] = \Pr\big[\mathbf{x}(t+1)\big|\mathbf{x}(t),\mathbf{u}(t)  \big]\notag \\
	& \cdot \Pr\big[\mathbf{u}(t)\big|\gamma(t), \mathbf{Y}(t)\big]  \Pr\big[\gamma(t)\big|\alpha(t),p(t)  \big] \Pr\big[\alpha(t)\big|\alpha(t-1)  \big]   \notag \\
	&    \cdot \Pr\big[\boldsymbol{\xi}(t+1)\big|\mathbf{i}(t+1), \boldsymbol{\Psi}(t+1), \mathbf{L}(t+1)\big] \notag \\
	&\cdot  \Pr\big[\mathbf{i}(t+1)\big|\boldsymbol{\Psi}(t), \boldsymbol{\xi}(t),\gamma(t),\mathbf{i}(t)\big] \notag \\
	& \cdot \Pr\big[\boldsymbol{\Psi}(t+1)\big|\boldsymbol{\Psi}(t) \big] \Pr\big[\mathbf{L}(t+1)\big|\mathbf{L}(t), \gamma(t) \big]
\end{align}
where each term  can be  calculated using the associated state dynamics in (\ref{plantain}), (\ref{errordyn}), (\ref{serequ}), (\ref{chnain}), (\ref{modelinapp}), (\ref{errordynfake}), (\ref{27equa}), and (\ref{28equa}).}

\txtblue{Note that the dynamics of $\boldsymbol{\Delta}(t)$ in (\ref{deltadyntext}) can be obtained from the dynamics of $\mathbf{x}(t)$ and $\hat{\mathbf{x}}(t)$ in (\ref{plantain}) and (\ref{errordyn}). Based on $\boldsymbol{\Delta}(t)$  in (\ref{deltadyntext}), the per-stage cost of the MDP is reduced to 
\bs\begin{align}
	 \hspace{-0.5cm}\mathbb{E}[\big(\boldsymbol{\Delta}^T(t)\mathbf{S} \boldsymbol{\Delta}(t) + \lambda p \big)\tau|  \mathbf{Y}(t)] = \mathbb{E}[\big(\boldsymbol{\Delta}^T(t)\mathbf{S} \boldsymbol{\Delta}(t) + \lambda p \big)\tau|  \boldsymbol{\chi}(t)]\notag 
\end{align}\bsc where $\boldsymbol{\chi}(t) \triangleq \big\{\boldsymbol{\Delta}(t-1), \alpha(t-1), \boldsymbol{\Psi}(t),\mathbf{L}(t) \big\}\subset \mathbf{Y}(t)$ with  the associated transition kernel  given by
\begin{align}	\label{totaltrans}
	&\Pr\big[\boldsymbol{\chi}(t+1)\big|\boldsymbol{\chi}(t),p(t)\big]  =\Pr\big[\alpha(t)\big|\alpha(t-1)\big]	\notag \\
	& \cdot \Pr\big[\boldsymbol{\Psi}(t+1)\big|\boldsymbol{\Psi}(t) \big] \Pr\big[\mathbf{L}(t+1)\big|\mathbf{L}(t), \alpha(t), p(t) \big] \notag \\
	& \cdot \Pr\big[\boldsymbol{\Delta}(t)\big|\boldsymbol{\Delta}(t-1), \boldsymbol{\Psi}(t),\mathbf{L}(t), \alpha(t), p(t) \big] \notag 	
\end{align}
where $\Pr\big[\boldsymbol{\Delta}(t)\big|\boldsymbol{\Delta}(t-1), \boldsymbol{\Psi}(t),\mathbf{L}(t), \alpha(t), p(t) \big]$ is associated with the dynamics of $\boldsymbol{\Delta}(t)$ in (\ref{deltadyntext}). Hence, Problem \ref{probformu} is an MDP with system state being $\boldsymbol{\chi}$, and the optimality condition in Theorem \ref{mdpvarify} directly follows from Prop. 4.6.1 of \cite{mdpcref2}}.

\vspace{-0.2cm}

\section*{Appendix D: Proof of Lemma \ref{perturbPDE}}
For convenience, denote 
\bs\begin{align}
	&T_{\boldsymbol{\chi}}(\theta, V, p)= \frac{1}{\tau} \mathbb{E}\big[\big((\boldsymbol{\Delta}')^T\mathbf{S} (\boldsymbol{\Delta}') + \lambda p \big)\tau \\
	& \qquad + \sum_{\boldsymbol{\chi} '}\Pr\left[\boldsymbol{\chi}'\big| \boldsymbol{\chi}, p  \right] V \left(\boldsymbol{\chi}'\right) -V \left(\boldsymbol{\chi} \right)\big|  \boldsymbol{\chi}\big] - \theta \notag	\\
	&T_{\boldsymbol{\chi}}^\dagger(\theta, V, p)=  \boldsymbol{\Delta}^T \mathbf{S} \boldsymbol{\Delta} +  \big[\lambda + \big(V(\mathbf{0}, \alpha, \boldsymbol{\Psi}, \mathbf{L})\label{49importantequa}   \\
	&- V(\boldsymbol{\Delta}, \alpha, \boldsymbol{\Psi}, \mathbf{L})\big)\frac{  \alpha}{\kappa(R) B_W}\big]p+\nabla_{\boldsymbol{\Delta}}^TV(\boldsymbol{\Delta}, \alpha, \boldsymbol{\Psi},\mathbf{L}) \widetilde{\mathbf{F}} \boldsymbol{\Delta} \notag\\
	& +\frac{1}{2}\text{Tr}\big( \nabla_{\boldsymbol{\Delta}}^2 V(\boldsymbol{\Delta}, \alpha, \boldsymbol{\Psi},\mathbf{L} )\widetilde{\mathbf{W}}\big)+\frac{\partial V(\boldsymbol{\Delta}, \alpha, \boldsymbol{\Psi},\mathbf{L})}{\partial \alpha} \left(2 \widetilde{a} \alpha +2 \widetilde{a}  \right)  \notag \\
	&+  \frac{\partial^2V^\ast(\boldsymbol{\Delta}, \alpha, \boldsymbol{\Psi},\mathbf{L})}{\partial \alpha^2}4 \widetilde{a}\alpha  + \text{Tr}\big(\frac{\partial V(\boldsymbol{\Delta}, \alpha, \boldsymbol{\Psi}, \mathbf{L})}{\partial \boldsymbol{\Psi}}(\mathbf{H}\boldsymbol{\Psi}-\boldsymbol{\Psi})/\tau\big) \notag \\
	& + \nabla^T_{\mathbf{L}}V(\boldsymbol{\Delta}, \alpha, \boldsymbol{\Psi}, \mathbf{L})\big(\boldsymbol{\Gamma}\mathbf{F}_{\mathbf{R}}\mathbf{L}+w_{max}\|\mathbf{H}\boldsymbol{\Psi}\|\mathbf{1}-\mathbf{L}\big)/\tau-\theta \notag
\end{align}\bsc

\hspace{-0.3cm} \emph{A. Relationship between $T_{\boldsymbol{\chi}}(\theta, V, p)$ and $T_{\boldsymbol{\chi}}^\dagger(\theta, V, p)$}
\txtblue{\begin{Lemma}	\label{applema}
	For any $\boldsymbol{\chi}$, $ T_{\boldsymbol{\chi}}(\theta, V, p)=T_{\boldsymbol{\chi}}^\dagger(\theta, V, p)+\mathcal{O}(\tau)+\mathcal{O}\big(\frac{\tau^2}{2^{2R_{min}}}\big)$.
\end{Lemma}}
\begin{proof}	[Proof of Lemma \ref{applema}]

\txtblue{\emph{a. Calculation of the per-stage cost:} We first calculate the per-stage cost in (\ref{OrgBel}):
\bs\begin{align}	
	&\mathbb{E}\big[(\boldsymbol{\Delta}')^T\mathbf{S} (\boldsymbol{\Delta}')\tau \big| \boldsymbol{\chi} \big]	\label{45cal} \\
	= &\mathbb{E}\big[\mathbf{e}^T(\mathbf{L})\boldsymbol{\Psi}^{-T}\boldsymbol{\Psi}^{-1}\mathbf{e}(\mathbf{L})\tau\Pr\left[\gamma(t)=1 \right] \notag \\
	&+ \left(\mathbf{F}\boldsymbol{\Delta}+\mathbf{w} \right)^T \mathbf{S} \left(\mathbf{F}\boldsymbol{\Delta}+\mathbf{w} \right)\tau\Pr\left[\gamma(t)=0 \right]\big|\boldsymbol{\chi}\big] 		\notag \\
	\overset{(a)}{=} &\mathbb{E}\big[\left(\mathbf{F}\boldsymbol{\Delta}+\mathbf{w} \right)^T \mathbf{S} \left(\mathbf{F}\boldsymbol{\Delta}+\mathbf{w} \right)\tau\Pr\left[\gamma(t)=0 \right]\big|\boldsymbol{\chi}\big] +\mathcal{O}\Big(\frac{\tau^4}{2^{2R_{min}}}\Big)	\notag
\end{align}\bsc where the calculations in $(a)$ are as follows:
\begin{align}
	&\mathbb{E}\big[\mathbf{e}^T(\mathbf{L})\boldsymbol{\Psi}^{-T}\boldsymbol{\Psi}^{-1}\mathbf{e}(\mathbf{L}) \big|\boldsymbol{\chi}\big] \leq \mathbb{E}\big[\|\mathbf{e}(\mathbf{L})\|^2 \|\boldsymbol{\Psi}^{-1}\|^2\big|\boldsymbol{\chi}\big] \notag \\
	\overset{(b)}{\leq} & C_1 \sum_{n=1}^d \frac{L_n^2}{12 \times 2^{2 R_i}}\overset{(c)}{=}\mathcal{O}\Big(\frac{\tau^2}{2^{2R_{min}}}\Big)		\label{51equation}
\end{align}
where (b) is because   $\|\boldsymbol{\Psi}^{-1}\|^2\leq C_1$ for some $C_1>0$ (according to Prop. 5.1 of \cite{rate1}) and the distribution of the quantization noise in Lemma \ref{lemmaqunt}, and (c) is because $L_n(t)=\mathcal{O}(\tau)$ for all $i$ and $t$ under $L_i(0)=\mathcal{O}(\tau)$, $w_{max}=\mathcal{O}(\tau)$ in (\ref{28equa}) and the dynamics of $\mathbf{L}(t)$ in (\ref{28equa}). Furthermore, $\Pr\left[\gamma(t)=1 \right]=1-\exp\big( -\frac{ p(t) \tau |h(t)|^2}{\kappa(R) B_W} \big)=\mathcal{O}(\tau)$. Therefore, $\mathbb{E}\big[\mathbf{e}^T(\mathbf{L})\boldsymbol{\Psi}^{-T}\boldsymbol{\Psi}^{-1}\mathbf{e}(\mathbf{L})\tau   \Pr\left[\gamma(t)=1 \right]\big|\boldsymbol{\chi}\big]  =\mathcal{O}\big(\frac{\tau^4}{2^{2R_{min}}}\big)$.} According to  (\ref{plantain}), we have  $\mathbf{F}=\mathbf{I} +\widetilde{\mathbf{F}}\tau+\mathcal{O}(\tau^2)$, $ \mathbf{W}=\widetilde{\mathbf{W}}\tau+\mathcal{O}(\tau^2)$. Together with $\Pr\left[\gamma(t)=0 \right]=1+\mathcal{O}\left(\tau\right) $. Hence, (\ref{45cal}) can be simplified as 
\begin{align}		\label{veryimp1}
	\bs\mathbb{E}\left[(\boldsymbol{\Delta}')^T\mathbf{S} (\boldsymbol{\Delta}')\tau \big| \boldsymbol{\chi} \right]= \boldsymbol{\Delta}^T \mathbf{S} \boldsymbol{\Delta}\tau +\mathcal{O}(\tau^2)  +\mathcal{O}\Big(\frac{\tau^4}{2^{2R_{min}}}\Big)\bsc
\end{align}

\txtblue{\emph{b. Calculation of the expectation involving the transition kernel:}  Substituting  $V\in \mathcal{C}^2$ that satisfies the PDE in (\ref{bellman2}) into the R.H.S. of  the Bellman equation in (\ref{OrgBel}), we calculate the expectation involving the transition kernel as follows:
\bs\begin{align}
	&\mathbb{E}\big[\sum_{\boldsymbol{\chi} '}\Pr\left[\boldsymbol{\chi}'\big| \boldsymbol{\chi}, p  \right] V \left(\boldsymbol{\chi}'\right)\big|  \boldsymbol{\chi}\big]\notag	 \\
	=& \mathbb{E}\Big[V( -\boldsymbol{\Psi}^{-1}\mathbf{e}(\mathbf{L}), \alpha', \mathbf{H}\boldsymbol{\Psi}, \boldsymbol{\Gamma}\mathbf{F}_{\mathbf{R}}\mathbf{L}+w_{max}\|\mathbf{H}\boldsymbol{\Psi}\|\mathbf{1})\notag \\
	&\Pr\left[\gamma(t)=1 \right]  + V(\mathbf{F}\boldsymbol{\Delta} +\mathbf{w}', \alpha', \mathbf{H}\boldsymbol{\Psi},\boldsymbol{\Gamma}\mathbf{L}+w_{max}\|\mathbf{H}\boldsymbol{\Psi}\|\mathbf{1})  \notag \\
	& \Pr\left[\gamma(t)=0 \right] \Big|\boldsymbol{\chi}\Big]\label{47equal}	
\end{align}\bsc We then handle (\ref{47equal}). Since $V \in \mathcal{C}^2$, we do the Taylor expansion\footnote{\txtblue{Note that although the optimal value function $V^\ast(\boldsymbol{\chi})$ may not be $\mathcal{C}^2$, the proof just requires the approximate value function $V(\boldsymbol{\chi})$ to be $\mathcal{C}^2$. In other words, we are seeking a $\mathcal{C}^2$ approximation of $V^\ast(\boldsymbol{\chi})$ with asymptotically vanishing errors for small $\tau$.}} of  the first term in (\ref{47equal}), and we have
\bs\begin{align}
	&\mathbb{E}\big[V( -\boldsymbol{\Psi}^{-1}\mathbf{e}(\mathbf{L}), \alpha', \mathbf{H}\boldsymbol{\Psi}, \boldsymbol{\Gamma}\mathbf{F}_{\mathbf{R}}\mathbf{L}+w_{max}\|\mathbf{H}\boldsymbol{\Psi}\|\mathbf{1})\notag \\
	&\cdot \Pr\left[\gamma(t)=1 \right]  \big|\boldsymbol{\chi}\big] \label{sub1}\\
	=&\mathbb{E}\Big[\Pr\left[\gamma(t)=1 \right] \big(V(\mathbf{0}, \alpha, \boldsymbol{\Psi}, \mathbf{L})-\nabla_{\boldsymbol{\Delta}}^TV(\mathbf{0}, \alpha, \boldsymbol{\Psi}, \mathbf{L})\boldsymbol{\Psi}^{-1}\mathbf{e}(\mathbf{L})\notag \\
	&+\mathcal{O}\big(\|\boldsymbol{\Psi}^{-1}\mathbf{e}(\mathbf{L})\|^2 \big)+\nabla_{\alpha}V(\mathbf{0}, \alpha, \boldsymbol{\Psi}, \mathbf{L})(\alpha'-\alpha)\notag  \\
	& +\frac{1}{2}\nabla_{\alpha}^2 V(\mathbf{0}, \alpha, \boldsymbol{\Psi}, \mathbf{L})(\alpha'-\alpha)^2  +\mathcal{O}\big((\alpha'-\alpha)^3 \big)\notag  \\
	& + \text{Tr}\Big(\frac{\partial V(\mathbf{0}, \alpha, \boldsymbol{\Psi}, \mathbf{L})}{\partial \boldsymbol{\Psi}}(\mathbf{H}\boldsymbol{\Psi}-\boldsymbol{\Psi})\Big)+\mathcal{O}\big(\|\mathbf{H}\boldsymbol{\Psi}-\boldsymbol{\Psi}\|^2\big)\notag \\
	& + \nabla^T_{\mathbf{L}}V(\mathbf{0}, \alpha, \boldsymbol{\Psi}, \mathbf{L})\big(\boldsymbol{\Gamma}\mathbf{F}_{\mathbf{R}}\mathbf{L}+w_{max}\|\mathbf{H}\boldsymbol{\Psi}\|\mathbf{1}-\mathbf{L}\big)\notag \\
	& +\mathcal{O}\big(\|\boldsymbol{\Gamma}\mathbf{F}_{\mathbf{R}}\mathbf{L}+w_{max}\|\mathbf{H}\boldsymbol{\Psi}\|\mathbf{1}-\mathbf{L}\|^2\big)\big)\big|\boldsymbol{\chi}\Big]\notag \\
	\overset{(d)}{=}&V(\mathbf{0}, \alpha, \boldsymbol{\Psi}, \mathbf{L}) \frac{ p \tau \alpha}{\kappa(R) B_W} +\mathcal{O}(\tau^2)  +\mathcal{O}\Big(\frac{\tau^3}{2^{2R_{min}}}\Big)\notag 
\end{align}\bsc where (d) is because of the following: 1) $\Pr\big[\gamma(t)=1 \big]  =\mathcal{O}(\tau)$, 2)  $\mathbb{E}\big[\mathbf{e}(\mathbf{L})|\big|\boldsymbol{\chi}\big]=0$ and $\mathcal{O}\big(\|\boldsymbol{\Psi}^{-1}\mathbf{e}(\mathbf{L})\|^2 \big)=\mathcal{O}\big(\frac{\tau^2}{2^{2R_{min}}}\big)$, 3) $\mathbb{E}\big[\alpha'-\alpha \big|\boldsymbol{\chi}\big]=(2 \widetilde{a}\alpha+2 \widetilde{a})\tau=\mathcal{O}(\tau)$, $\mathbb{E}\big[(\alpha'-\alpha)^2 \big|\boldsymbol{\chi}\big]=8\widetilde{a}\alpha\tau $ and $\mathcal{O}\big((\alpha'-\alpha)^3 \big)=\mathcal{O}(\tau^2)$ under the dynamics of $h(t)$ in (\ref{chnain}), 4) $\mathbf{H}\boldsymbol{\Psi}-\boldsymbol{\Psi}=\mathcal{O}(\tau)$ and $\|\mathbf{H}\boldsymbol{\Psi}-\boldsymbol{\Psi}\|=\mathcal{O}(\tau^2)$ according to Theorem 4.1 of \cite{rate1} and $\mathbf{F}=\mathbf{I} +\widetilde{\mathbf{F}}\tau+\mathcal{O}\big(\tau^2\big)$ in (\ref{plantain}), and 5) $\boldsymbol{\Gamma}\mathbf{F}_{\mathbf{R}}\mathbf{L}+w_{max}\|\mathbf{H}\boldsymbol{\Psi}\|\mathbf{1}-\mathbf{L}=\mathcal{O}(\tau)$ and $\|\boldsymbol{\Gamma}\mathbf{F}_{\mathbf{R}}\mathbf{L}+w_{max}\|\mathbf{H}\boldsymbol{\Psi}\|\mathbf{1}-\mathbf{L}\|^2=\mathcal{O}(\tau^2)$ according to (c) of (\ref{51equation}).}

\txtblue{Using the same calculations  in  (d) of (\ref{sub1}), we do the second order Taylor expansion of  the second term in (\ref{47equal}) as follows:
\bs\begin{align}
	&\mathbb{E}\big[ V(\mathbf{F}\boldsymbol{\Delta} +\mathbf{w}', \alpha', \mathbf{H}\boldsymbol{\Psi},\boldsymbol{\Gamma}\mathbf{L}+w_{max}\|\mathbf{H}\boldsymbol{\Psi}\|\mathbf{1})\notag \\
	&\cdot \Pr\left[\gamma(t)=0 \right] \big|\boldsymbol{\chi}\big]\label{sub2}  \\
	=&V(\boldsymbol{\Delta}, \alpha, \boldsymbol{\Psi},\mathbf{L})-V(\boldsymbol{\Delta}, \alpha, \boldsymbol{\Psi},\mathbf{L})\frac{ p \tau \alpha}{\kappa(R) B_W}\notag  \\
	& +\nabla_{\boldsymbol{\Delta}}^TV(\boldsymbol{\Delta}, \alpha, \boldsymbol{\Psi},\mathbf{L}) \widetilde{\mathbf{F}} \boldsymbol{\Delta} \tau +\frac{1}{2}\text{Tr}\left( \nabla_{\boldsymbol{\Delta}}^2 V(\boldsymbol{\Delta}, \alpha, \boldsymbol{\Psi},\mathbf{L} )\widetilde{\mathbf{W}}\right) \tau \notag \\
	& +\frac{\partial V(\boldsymbol{\Delta}, \alpha, \boldsymbol{\Psi},\mathbf{L})}{\partial \alpha} \left(2 \widetilde{a} \alpha +2 \widetilde{a}  \right)\tau   +\frac{\partial^2V(\boldsymbol{\Delta}, \alpha, \boldsymbol{\Psi},\mathbf{L})}{\partial \alpha^2}4\widetilde{a}\alpha\tau \notag \\
	& + \text{Tr}\Big(\frac{\partial V(\boldsymbol{\Delta}, \alpha, \boldsymbol{\Psi}, \mathbf{L})}{\partial \boldsymbol{\Psi}}(\mathbf{H}\boldsymbol{\Psi}-\boldsymbol{\Psi})\Big)  + \nabla^T_{\mathbf{L}}V(\boldsymbol{\Delta}, \alpha, \boldsymbol{\Psi}, \mathbf{L})\notag \\
	&\cdot \big(\boldsymbol{\Gamma}\mathbf{F}_{\mathbf{R}}\mathbf{L}+w_{max}\|\mathbf{H}\boldsymbol{\Psi}\|\mathbf{1}-\mathbf{L}\big)+\mathcal{O}(\tau^2) \notag
\end{align}\bsc}Substituting (\ref{sub1}) and (\ref{sub2}) into (\ref{47equal}), and together with (\ref{veryimp1}), we obtain $ T_{\boldsymbol{\chi}}(\theta, V, p)=T_{\boldsymbol{\chi}}^\dagger(\theta, V, p)+\mathcal{O}(\tau)  +\mathcal{O}\Big(\frac{\tau^2}{2^{2R_{min}}}\Big)$. \end{proof}

\vspace{0.1cm}
\hspace{-0.3cm} \emph{B. Growth Rate of $T_{\boldsymbol{\chi}}(\theta, V, p)$}

Denote 
\begin{align}
	\hspace{-0.5cm}T_{\boldsymbol{\chi}}(\theta, V)=\min_{  p} T_{\boldsymbol{\chi}}(\theta, V, p), \ T_{\boldsymbol{\chi}}^\dagger(\theta, V)=\min_{  p} T_{\boldsymbol{\chi}}^\dagger(\theta, V, p)	\label{zerofuncsa}		
\end{align}
Suppose $(\theta^\ast,V^\ast)$ satisfies the Bellman equation in (\ref{OrgBel}) and $(\theta,V)$ satisfies the PDE in (\ref{bellman2}). We have for any $\boldsymbol{\chi}$,
\begin{align}
	T_{\boldsymbol{\chi}}(\theta^\ast, V^\ast)=0, \quad T_{\boldsymbol{\chi}}^\dagger(\theta, V)=0	\label{zerofunc}
\end{align}
Then, we establish the following lemma.
\txtblue{\begin{Lemma}	\label{applemma}
$\big|T_{\boldsymbol{\chi}}  (\theta, V)\big|=\mathcal{O}(\tau)  +\mathcal{O}\Big(\frac{\tau^2}{2^{2R_{min}}}\Big)$,  $\forall {\boldsymbol{\chi}} $.
\end{Lemma}
\begin{proof}	[Proof of Lemma \ref{applemma}]
For any $\boldsymbol{\chi}$, we have \bs$T_{\boldsymbol{\chi}} (\theta, V)=\min_{ p}\big[ T_{\boldsymbol{\chi}}^\dagger(\theta, V, p)+\mathcal{O}(\tau)  +\mathcal{O}\big(\frac{\tau^2}{2^{2R_{min}}}\big)\big]  \geq \min_{ p} T_{\boldsymbol{\chi}}^\dagger(\theta, V, p)  + \mathcal{O}(\tau)  +\mathcal{O}\big(\frac{\tau^2}{2^{2R_{min}}}\big)$\bsc. On the other hand, \bs$T_{\boldsymbol{\chi}} (\theta, V) \leq \min_{ p} T_{\boldsymbol{\chi}}^\dagger(\theta, V, p) + \mathcal{O}(\tau)  +\mathcal{O}\big(\frac{\tau^2}{2^{2R_{min}}}\big)$\bsc, where \ $p^\dagger= \arg \min_{p} T_{\boldsymbol{\chi}}^\dagger(\theta, V, p) $. From (\ref{zerofuncsa}) and (\ref{zerofunc}), \bs$\min_{ p} T_{\boldsymbol{\chi}}^\dagger(\theta, V, p)=T_{\boldsymbol{\chi}}^\dagger(\theta, V)=0$\bsc. Therefore, \bs$\big|T_{\boldsymbol{\chi}}  (\theta, V)\big|=\mathcal{O}(\tau)  +\mathcal{O}\big(\frac{\tau^2}{2^{2R_{min}}}\big)$\bsc.\end{proof}}

\vspace{0.1cm}
\hspace{-0.3cm} \emph{C. Difference between $V^\ast\left(\boldsymbol{\chi}\right)$ and $V\left(\boldsymbol{\chi}\right)$}

We then prove the following Lemma:
\begin{Lemma}		\label{tenlemma}
	Suppose $T_{\boldsymbol{\chi}}(\theta^\ast, V^\ast) = 0$ for all $\boldsymbol{\chi}$ together with the transversality condition in (\ref{transodts})  has a unique solution $(\theta^*, V^\ast)$. If $T_{\boldsymbol{\chi}}^\dagger(\theta, V)=0$ and $V(\boldsymbol{\chi})=\mathcal{O}(\left\|\boldsymbol{\Delta}\right\|^2)$,  then $V^\ast\left(\boldsymbol{\chi}\right)-V\left(\boldsymbol{\chi}\right)=\mathcal{O}(\tau)  +\mathcal{O}\big(\frac{\tau^2}{2^{2R_{min}}}\big)$ for all  $\boldsymbol{\chi} $.
\end{Lemma}

%applemuseful

\begin{proof}	[Proof of Lemma \ref{tenlemma}]
Using $V(\boldsymbol{\chi})=\mathcal{O}(\left\|\boldsymbol{\Delta}\right\|^2)$ and definition \ref{admisscontrolpol}, we have $\lim_{t \rightarrow \infty}\mathbb{E}^{\Omega_p} \left[V(\boldsymbol{\chi})\right]<\infty$ for any admissible policy $\Omega_p$. Then, we have $\lim_{T\rightarrow \infty} \frac{1}{T}\mathbb{E}^{\Omega_p}\left[ V\left(\boldsymbol{\chi}(T)\right)|\boldsymbol{\chi} \left(0 \right)\right]= 0$ and the transversality condition in (\ref{transodts}) is satisfied for $V(\boldsymbol{\chi})$.

	Suppose for some $\boldsymbol{\chi}'$, we have $V\left(\boldsymbol{\chi}' \right)=V^\ast\left(\boldsymbol{\chi}' \right)+c$ for some $c \neq 0$ as $\tau \rightarrow 0$.  Let $\tau \rightarrow 0$. From Lemma \ref{applemma}, we have $(\theta, V)$ satisfies $T_{\boldsymbol{\chi}}(\theta, V)= 0$ for all $\boldsymbol{\chi}$ and   the transversality condition in (\ref{transodts}). However, $V\left(\boldsymbol{\chi}' \right) \neq V^\ast\left(\boldsymbol{\chi}' \right)$ because of the assumption that $V\left(\boldsymbol{\chi}' \right)=V^\ast\left(\boldsymbol{\chi}' \right)+c$. This contradicts  the condition that $(\theta^*, V^\ast)$ is a unique solution of $T_{\boldsymbol{\chi}}(\theta^\ast, V^\ast)=0$ for all $\boldsymbol{\chi}$  and the transversality condition in (\ref{transodts}). Hence, we must have $V\left(\boldsymbol{\chi} \right)-V^\ast\left(\boldsymbol{\chi} \right)=\mathcal{O}(\tau)  +\mathcal{O}\big(\frac{\tau^2}{2^{2R_{min}}}\big)$ for all  $\boldsymbol{\chi} $.\end{proof}

\vspace{-0.2cm}

\section*{Appendix E: Proof of Theorem \ref{perfgap}}

We  calculate the performance under policy $\widetilde{\Omega}_p^\ast$ as follows:
\bs\begin{align}
	&  \widetilde{\theta}^\ast\tau  =\mathbb{E}^{\widetilde{\Omega}_p^\ast}\big[\mathbb{E}\big[\big((\boldsymbol{\Delta}')^T\mathbf{S} (\boldsymbol{\Delta}') + \lambda p \big)\tau\big]\big| \boldsymbol{\chi}\big]\label{finalexpr}  \\
	\overset{(a)}=& \mathbb{E}^{\widetilde{\Omega}_p^\ast}\big[\mathbb{E}\big[\big((\boldsymbol{\Delta}')^T\mathbf{S} (\boldsymbol{\Delta}') + \lambda p \big)\tau\notag \\
	& +\sum_{\boldsymbol{\chi} '} {\Pr}\big[ \boldsymbol{\chi} '| (\boldsymbol{\chi},  \widetilde{\Omega}_p^\ast\left((\boldsymbol{\chi}  \right)\big]V \left(\boldsymbol{\chi}  '\right)  - V \left(\boldsymbol{\chi}  \right)   \big| \boldsymbol{\chi}\big]\big] 	\notag \\
	\overset{(b)}=& \mathbb{E}^{\widetilde{\Omega}_p^\ast}\Big[T_{\boldsymbol{\chi}}^\dagger(\theta, V, p)\tau+\theta \tau+\txtblue{\mathcal{O}(\tau^2)  +\mathcal{O}\Big(\frac{\tau^3}{2^{2R_{min}}}\Big)} \Big]\notag 
\end{align}\bsc where ${\Pr}\big[ \boldsymbol{\chi} '|\boldsymbol{\chi},  \widetilde{\Omega}_p^\ast\left(\boldsymbol{\chi}\right)\big]$ is the  transition kernel under policy $\widetilde{\Omega}_p^\ast$. (a) is due to 1) $\mathbb{E}^{\widetilde{\Omega}_p^\ast}\big[V(\boldsymbol{\chi} ) \big]<\infty$  (according to $V(\boldsymbol{\chi})=\mathcal{O}(\left\|\boldsymbol{\Delta}\right\|^2)$ and Definition \ref{admisscontrolpol}, and  2) \bs$ \mathbb{E}^{\widetilde{\Omega}_p^\ast}\big[\sum_{\boldsymbol{\chi} '} \mathbb{E} [ {\Pr}\big[ \boldsymbol{\chi} '\big| \boldsymbol{\chi},    \widetilde{\Omega}_p^\ast\left(\boldsymbol{\chi}  \right)\big]\big| \boldsymbol{\chi}]V \left(\boldsymbol{\chi}  '\right)   \big]=\mathbb{E}^{\widetilde{\Omega}_p^\ast}\big[\mathbb{E}^{\widetilde{\Omega}_p^\ast}\big[V(\boldsymbol{\chi} ') \big|\boldsymbol{\chi}  \big]\big]=\mathbb{E}^{\widetilde{\Omega}_p^\ast}\big[V(\boldsymbol{\chi} ) \big]$\bsc, and $(b)$ is due to the calculations in (\ref{47equal}).

Following the notation in  Appendix D, we define two mappings: $T_{\boldsymbol{\chi}}^\dagger( V, p) = T_{\boldsymbol{\chi}}^\dagger(\theta, V, p) + \theta $, $T_{\boldsymbol{\chi}}(V, p)=T_{\boldsymbol{\chi}}^\dagger(V, p)+\nu  G_{\boldsymbol{\chi}}(V,p)$. Let $\Omega_p^\ast$ be the optimal policy solving the discrete time Bellman equation in (\ref{OrgBel}). Then we have
\begin{align}
	T_{\boldsymbol{\chi}}({V^\ast}, \Omega_p^\ast(\boldsymbol{\chi}))= \theta^\ast, \quad \forall \boldsymbol{\chi} 	\label{asdadadasd}
\end{align}
Furthermore, we have
\begin{align}	\label{minach}
	T_{\boldsymbol{\chi}}^\dagger( V, \widetilde{\Omega}_p^\ast(\boldsymbol{\chi} ))= \min_{\Omega_p(\boldsymbol{\chi} )} T_{\boldsymbol{\chi}}^\dagger( V, \Omega_p(\boldsymbol{\chi})), \quad \forall \boldsymbol{\chi}	
\end{align}
Dividing $\tau$ on both sizes of (\ref{finalexpr}), we obtain
\begin{align}
	&\widetilde{\theta}^\ast =\mathbb{E}^{\widetilde{\Omega}_p^\ast}\big[T_{\boldsymbol{\chi}}^\dagger(V, \widetilde{\Omega}_p^\ast(\boldsymbol{\chi} ))+ \txtblue{\mathcal{O}(\tau)  +\mathcal{O}\big(\frac{\tau^2}{2^{2R_{min}}}\big)}\big] 	\notag \\
	\overset{(c)} \leq & \mathbb{E}^{\widetilde{\Omega}_p^\ast}\big[T_{\boldsymbol{\chi}}^\dagger(V, {\Omega}_p^\ast(\boldsymbol{\chi} ))+ \txtblue{\mathcal{O}(\tau)  +\mathcal{O}\big(\frac{\tau^2}{2^{2R_{min}}}\big)}\big] \notag \\
	 \overset{(d)}{=} & \mathbb{E}^{\widetilde{\Omega}_p^\ast}\big[T_{\boldsymbol{\chi}}(V, {\Omega}_p^\ast(\boldsymbol{\chi} ))+ \txtblue{\mathcal{O}(\tau)  +\mathcal{O}\big(\frac{\tau^2}{2^{2R_{min}}}\big)}\big] 	\notag \\
	\overset{(e)}=&  \mathbb{E}^{\widetilde{\Omega}_p^\ast}\big[T_{\boldsymbol{\chi}}(V, {\Omega}_p^\ast(\boldsymbol{\chi} ))-T_{\boldsymbol{\chi}}(V^\ast, {\Omega}_p^\ast(\boldsymbol{\chi} ))+\theta^\ast \notag \\
	& + \txtblue{\mathcal{O}(\tau)  +\mathcal{O}\big(\frac{\tau^2}{2^{2R_{min}}}\big)}\big] 	\label{finaleeeee}
\end{align} where (c) is due to (\ref{minach}), (d) is due to Lemma \ref{applema},  and  (e) is due to (\ref{asdadadasd}).  \txtblue{Then, from (\ref{finaleeeee}),   we have
\bs\begin{align}
	&\widetilde{\theta}^\ast -  \theta^\ast  \notag \\
	\leq & \mathbb{E}^{\widetilde{\Omega}_p^\ast}\Big[T_{\boldsymbol{\chi}}(V, {\Omega}_p^\ast(\boldsymbol{\chi} ))-T_{\boldsymbol{\chi}}(V^\ast, {\Omega}_p^\ast(\boldsymbol{\chi} ))\Big]  + \mathcal{O}(\tau)  +\mathcal{O}\Big(\frac{\tau^2}{2^{2R_{min}}}\Big) \notag \\
	\overset{(f)}{\leq} & \gamma \mathbb{E}^{\widetilde{\Omega}_p^\ast}\Big[\omega(\boldsymbol{\chi})\|\mathbf{V}^\ast-\mathbf{V}\|_{\infty}^{\overline{\omega}}\Big]+ \mathcal{O}(\tau)  +\mathcal{O}\Big(\frac{\tau^2}{2^{2R_{min}}}\Big)\notag \\
	\overset{(g)}{=}& \gamma \mathbb{E}^{\widetilde{\Omega}_p^\ast}\Big[\omega(\boldsymbol{\chi})\Big(\mathcal{O}(\tau)  +\mathcal{O}\Big(\frac{\tau^2}{2^{2R_{min}}}\Big)\Big)\Big] + \mathcal{O}(\tau)  +\mathcal{O}\Big(\frac{\tau^2}{2^{2R_{min}}}\Big)\notag \\
	\overset{(h)}{=} & \mathcal{O}(\tau)  +\mathcal{O}\Big(\frac{\tau^2}{2^{2R_{min}}}\Big)
\end{align}\bsc where $(f)$ holds because
\begin{align}	\label{comtracmapp}
	\|T_{\boldsymbol{\chi}}(V, {\Omega}_p^\ast(\boldsymbol{\chi} ))-T_{\boldsymbol{\chi}}(V^\ast, {\Omega}_p^\ast(\boldsymbol{\chi} ))\|_{\infty}^{\overline{\omega}}\leq \gamma\|\mathbf{V}^\ast-\mathbf{V}\|_{\infty}^{\overline{\omega}}
\end{align}
with $\mathbf{V}^\ast=\{V^\ast(\boldsymbol{\chi}):\forall \boldsymbol{\chi}\}$ and $\mathbf{V}=\{V(\boldsymbol{\chi}):\forall \boldsymbol{\chi}\}$, for $0<\gamma<1$ according to Lemma 3 of \cite{hinidhsd} and $\|\cdot\|_{\infty}^{\overline{\omega}}$ is a weighted sup-norm with weights $\overline{\omega}=\{0<\omega(\boldsymbol{\chi})<1:\forall \boldsymbol{\chi}\}$ chosen according to the following rule (Lemma 3 of \cite{hinidhsd}): The state space w.r.t. $\boldsymbol{\chi}$ is partitioned into non-empty subsets $\mathcal{S}_1,\dots, \mathcal{S}_r$, in which for any $\boldsymbol{\chi}\in \mathcal{S}_n$ with $n=1,\dots, r$,  there exists some $\boldsymbol{\chi}'\in \mathcal{S}_1\cup \mathcal{S}_{n-1}$ such that $\Pr[\boldsymbol{\chi}'|\boldsymbol{\chi},{\Omega}_p^\ast(\boldsymbol{\chi}) ]>0$. Then, we let $\rho=\min \{\Pr[\boldsymbol{\chi}'|\boldsymbol{\chi},{\Omega}_p^\ast(\boldsymbol{\chi}) ]:\forall \boldsymbol{\chi},\boldsymbol{\chi}' \}$, and choose $\omega(\boldsymbol{\chi})=1-\rho^{2n}$ if $\boldsymbol{\chi}\in \mathcal{S}_n$ for $n=1,\dots, r$. Therefore, based on the contraction mapping property in (\ref{comtracmapp}) and the definition of the weighted sup-norm, we have
\bs\begin{align}
	& \frac{T_{\boldsymbol{\chi}}(V, {\Omega}_p^\ast(\boldsymbol{\chi} ))-T_{\boldsymbol{\chi}}(V^\ast, {\Omega}_p^\ast(\boldsymbol{\chi} ))}{\omega(\boldsymbol{\chi})}  \\
	\leq & \sup_{\boldsymbol{\chi} }\Big\{\frac{T_{\boldsymbol{\chi}}(V, {\Omega}_p^\ast(\boldsymbol{\chi} ))-T_{\boldsymbol{\chi}}(V^\ast, {\Omega}_p^\ast(\boldsymbol{\chi} ))}{\omega(\boldsymbol{\chi})}\Big\}\notag \\
	=& \|T_{\boldsymbol{\chi}}(V, {\Omega}_p^\ast(\boldsymbol{\chi} ))-T_{\boldsymbol{\chi}}(V^\ast, {\Omega}_p^\ast(\boldsymbol{\chi} ))\|_{\infty}^{\overline{\omega}}\leq \gamma\|\mathbf{V}^\ast-\mathbf{V}\|_{\infty}^{\overline{\omega}}\notag \\
	\Rightarrow &T_{\boldsymbol{\chi}}(V, {\Omega}_p^\ast(\boldsymbol{\chi} ))-T_{\boldsymbol{\chi}}(V^\ast, {\Omega}_p^\ast(\boldsymbol{\chi} ))\leq \gamma \omega(\boldsymbol{\chi})\|\mathbf{V}^\ast-\mathbf{V}\|_{\infty}^{\overline{\omega}}\notag
\end{align}\bsc This proves (f), and (g) is because $\|\mathbf{V}^\ast-\mathbf{V}\|_{\infty}^{\overline{\omega}}=\sup_{\boldsymbol{\chi}}\big\{\frac{|V^\ast(\boldsymbol{\chi})-V(\boldsymbol{\chi})|}{\omega(\boldsymbol{\chi})}\big\}=\mathcal{O}(\tau)  +\mathcal{O}\Big(\frac{\tau^2}{2^{2R_{min}}}\Big)$ according to Lemma \ref{perturbPDE}, and (h) is because $0<\omega(\boldsymbol{\chi})<1$ for all $\boldsymbol{\chi}$.}

\vspace{-0.2cm}

\section*{Appendix F: Proof of Lemma \ref{solpde}}
Since the optimization problem on the R.H.S. of (\ref{bellman2}) is a linear programming, the optimal power control that achieves this minimum can be directly obtained which is summarized in Theorem \ref{thmpower}.Substituting $p^\ast$, the PDE in (\ref{perturbPDE}) becomes 
\bs\begin{align}	
		&\theta= \boldsymbol{\Delta}^T \mathbf{S}\boldsymbol{\Delta} - \big[  \left(V(\boldsymbol{\chi})-V(\boldsymbol{\chi})\big|_{\boldsymbol{\Delta}=\mathbf{0}}   \right)\frac{  \alpha}{\kappa(R) B_W}-\lambda\big]^+ p_{max}    \notag \\
		&+\nabla_{\boldsymbol{\Delta}}^TV(\boldsymbol{\chi}) \widetilde{\mathbf{F}} \boldsymbol{\Delta}  +\frac{1}{2}\text{Tr}\big( \nabla_{\boldsymbol{\Delta}}^2 V(\boldsymbol{\chi} )\widetilde{\mathbf{W}}\big) +\frac{\partial V(\boldsymbol{\chi})}{\partial \alpha} \left(2 \widetilde{a} \alpha +2 \widetilde{a}  \right)  \notag \\
		& +\frac{\partial^2V^\ast(\boldsymbol{\chi})}{\partial \alpha^2}4\widetilde{a}\alpha + \text{Tr}\big(\frac{\partial V(\boldsymbol{\chi})}{\partial \boldsymbol{\Psi}}(\mathbf{H}\boldsymbol{\Psi}-\boldsymbol{\Psi})/\tau\big)\notag \\
	&+ \nabla^T_{\mathbf{L}}V(\boldsymbol{\chi})\big(\boldsymbol{\Gamma}\mathbf{F}_{\mathbf{R}}\mathbf{L}+w_{max}\|\mathbf{H}\boldsymbol{\Psi}\|\mathbf{1}-\mathbf{L}\big)/\tau
\label{pdeintheapp}
\end{align}\bsc

\hspace{-0.3cm}\emph{A. Solution of (\ref{pdeintheapp}) for small $\|\boldsymbol{\Delta}\|^2\alpha$}   
\vspace{0.1cm}

We solve the above PDE in (\ref{pdeintheapp})   when $\lambda >(V(\boldsymbol{\chi}) -V(\boldsymbol{\chi})\big|_{\boldsymbol{\Delta}=\mathbf{0}} ) \frac{  \alpha}{\kappa(R) B_W}$ (i.e., $p=0$). We will show   later that small $\|\boldsymbol{\Delta}\|^2\alpha$ leads to this case. For this case, (\ref{pdeintheapp})  becomes  
\bs\begin{align}	
		&\theta= \boldsymbol{\Delta}^T \mathbf{S}\boldsymbol{\Delta}+\nabla_{\boldsymbol{\Delta}}^TV(\boldsymbol{\chi}) \widetilde{\mathbf{F}} \boldsymbol{\Delta} +\frac{1}{2}\text{Tr}\big( \nabla_{\boldsymbol{\Delta}}^2 V(\boldsymbol{\chi} )\widetilde{\mathbf{W}}\big)     \notag \\
		&+\frac{\partial V(\boldsymbol{\chi})}{\partial \alpha} \left(2 \widetilde{a} \alpha +2 \widetilde{a}  \right)+  \frac{\partial^2V^\ast(\boldsymbol{\chi})}{\partial \alpha^2}4\widetilde{a}\alpha + \text{Tr}\big(\frac{\partial V(\boldsymbol{\chi})}{\partial \boldsymbol{\Psi}}(\mathbf{H}\boldsymbol{\Psi}-\boldsymbol{\Psi})/\tau\big) \notag \\
		&+ \nabla^T_{\mathbf{L}}V(\boldsymbol{\chi})\big(\boldsymbol{\Gamma}\mathbf{F}_{\mathbf{R}}\mathbf{L}+w_{max}\|\mathbf{H}\boldsymbol{\Psi}\|\mathbf{1}-\mathbf{L}\big)/\tau	\label{pde1wrding}
\end{align}\bsc The PDE is separable with solution of the form $V(\boldsymbol{\chi})=\boldsymbol{\Delta}^TA_{1,\mathbf{U}}(\alpha)\boldsymbol{\Delta}+b_1(\alpha)$ for some $A_{1,\mathbf{U}}(\alpha)$ and $b_1(\alpha)$. Substituting this form into (\ref{pde1wrding}), we obtain that 
\bs\begin{align}
	&\boldsymbol{\Delta}^T \Big[\mathbf{S}+ \left(2 \widetilde{a} \alpha + 2 \widetilde{a}\right)A_{1,\mathbf{U}}'(\alpha)+ 4\widetilde{a} \alpha A_{1,\mathbf{U}}''(\alpha) +\big(A_{1,\mathbf{U}}(\alpha) \notag \\
		&+A_{1,\mathbf{U}}^T(\alpha)\big) \widetilde{\mathbf{F}}  \Big] \boldsymbol{\Delta} + \big[\left(2 \widetilde{a} \alpha + 2 \widetilde{a}\right)b_1'(\alpha) +4\widetilde{a} \alpha b_1''(\alpha)\notag \\ 
		&+\frac{1}{2}\text{Tr}\big( \left(A_{1,\mathbf{U}}(\alpha)+A_{1,\mathbf{U}}^T(\alpha) \right)  \widetilde{\mathbf{W}} \big)-\theta \big]	=0\label{59veryimp}
\end{align}\bsc Let the eigenvalue decomposition of $\widetilde{\mathbf{F}}$  be $\widetilde{\mathbf{F}}=\mathbf{U}^{-1}\boldsymbol{\Gamma}_{\widetilde{\mathbf{F}}} \mathbf{U}$, where $\mathbf{U}$ is an $d \times d$  matrix and $\boldsymbol{\Gamma}_{\widetilde{\mathbf{F}}} =\text{diag}\left(\mu_1, \mu_2, \dots, \mu_d \right)$. Using the change of variable $\mathbf{Z}=\mathbf{U}\boldsymbol{\Delta}$ and denoting $A_{1,\mu}(\alpha)=(\mathbf{U}^{-1})^\dagger A_{1, \mathbf{U}}(\alpha) \mathbf{U}^{-1}$,  we require the following equation for (\ref{59veryimp}) to hold for any $\boldsymbol{\Delta}$:
\begin{align}	
	\mathbf{S}_U+ \left(2 \widetilde{a} \alpha + 2\widetilde{a}\right)A_{1,\mu}'(\alpha)& + 4 \widetilde{a}\alpha A_{1,\mu}''(\alpha)  \label{basedonequ} \\
		&+\left(A_{1,\mu}(\alpha)+A_{1,\mu}^T(\alpha)\right)\Gamma_{\widetilde{\mathbf{F}}} =0\notag
\end{align}
Let $ A_{1,\mu}(\alpha)=[a_{1,\mu, ij}]$. For the diagonal elements in (\ref{basedonequ}),  
\begin{align}	
	\mathbf{S}_{U,ii}+ \left(2 \widetilde{a} \alpha + 2 \widetilde{a}\right)a_{1,\mu,ii}'(\alpha)&+ 4 \widetilde{a}\alpha a_{1,\mu, ii}''(\alpha)\notag \\
	& +2 \mu_i a_{1,\mu,ii}(\alpha)=0 \label{8964equ}
\end{align}
Using the method of dominant balance (MDB) \cite{dominmethod}, the asymptotic solution of (\ref{8964equ}) is given by
\begin{align}	
	\hspace{-0.5cm} a_{1,\mu,ii}(\alpha) = \exp\Big[-\frac{ \mu_i}{\widetilde{a} \alpha}  -\frac{\mu_i}{\widetilde{a}}\Big(1-\frac{2\mu_i}{\widetilde{a}}\Big)\log \alpha+\mathcal{O}\left({\alpha} \right) \Big] \label{62equ}
\end{align}
as $\alpha\rightarrow 0$. For the non-diagonal elements $a_{1,\mu, ij}$ and $a_{1,\mu, ji}$, they satisfy the following coupled ODEs based on (\ref{basedonequ}):
\begin{align}
	\mathbf{S}_{U,ij}&+ \left(2 \widetilde{a} \alpha + 2\widetilde{a}\right)a_{1,\mu,ij}'(\alpha)+ 4 \widetilde{a}\alpha a_{1,\mu, ij}''(\alpha) \notag \\
	&+ (a_{1,\mu,ij}(\alpha)+a_{1,\mu,ji}(\alpha))\mu_j=0 \label{xxy}	\\
	\mathbf{S}_{U,ji}&+ \left(2 \widetilde{a} \alpha + 2\widetilde{a}\right)a_{1,\mu,ji}'(\alpha)+ 4 \widetilde{a}\alpha a_{1,\mu, ji}''(\alpha) \notag \\
	&+ (a_{1,\mu,ij}(\alpha)+a_{1,\mu,ji}(\alpha))\mu_i=0 \label{yyx}
\end{align}
Even though (\ref{xxy}) and (\ref{yyx}) are coupled, we can first obtain $a_{1,\mu,ij}+a_{1,\mu,ji}$ by solving the  ODE by adding  them  together. Then, we obtain either $a_{1,\mu,ij}$ or $a_{1,\mu,ji}$ by solving one of  them. The results for $i>j$ are given by
\bs\begin{align}
	a_{1,\mu,ij}= &\frac{\mu_j}{\mu_i + \mu_j}\exp\Big[-\frac{\mu_i+\mu_j}{2 \widetilde{a}\alpha}- \frac{\mu_i+\mu_j}{2\widetilde{a}}\Big(1-\frac{\mu_i+\mu_j}{\widetilde{a}}\Big)\notag \\
	&\cdot\log \alpha+\mathcal{O}\left({\alpha} \right)\Big], \label{56eq} \\
	a_{1,\mu,ji}= &\frac{\mu_i}{\mu_i + \mu_j}\exp\Big[-\frac{\mu_i+\mu_j}{2 \widetilde{a}\alpha}- \frac{\mu_i+\mu_j}{2\widetilde{a}}\Big(1-\frac{\mu_i+\mu_j}{\widetilde{a}}\Big)\notag \\
	&\cdot\log \alpha+\mathcal{O}\left({\alpha} \right)\Big],\label{57eq}
\end{align}\bsc  as $\alpha\rightarrow 0$.  Using (\ref{62equ}), (\ref{56eq}), (\ref{57eq}) and the relationship $a_{1,\mu}(\alpha)=(\mathbf{U}^{-1})^\dagger A_{1, \mathbf{U}}(\alpha) \mathbf{U}^{-1}$, we can obtain $A_{1, \mathbf{U}}(\alpha)$. For obtaining $b_1(\alpha)$, from (\ref{59veryimp}), we require
\begin{align}
	\left(2 \widetilde{a} \alpha + 2 \widetilde{a}\right)b_1'(\alpha)&+ 4 \widetilde{a}\alpha b_1''(\alpha)  \\
	&+\frac{1}{2}\text{Tr}\big( \big(a_{1,\mu}+a_{1,\mu}^T \big)\mathbf{U}  \widetilde{\mathbf{W}}\mathbf{U}^\dagger \big)-\theta =0\notag
\end{align}
Similarly, using the MDB approach, we  obtain the asymptotic solution of  $b_1(\alpha)$. Based on the above calculations, we summarize the  overall solution of this case  below: 
\begin{align}		\label{appeox1}
			\log\left(V(\boldsymbol{\chi})\right)= \log (\boldsymbol{\Delta}^TA_{1,\mathbf{U}}(\alpha)\boldsymbol{\Delta}+b_1(\alpha))+\mathcal{O}\left({\alpha}\right)
		\end{align}
		 where $A_{1,\mathbf{U}}(\alpha)=\mathbf{U}^\dagger A_{1,\mu} (\alpha)\mathbf{U} \in \mathbb{R}^{d\times d}$,  $A_{1,\mu}(\alpha)=[a_{1,\mu, ij}(\alpha)]$ is a $d\times d$ constant matrix with $a_{1,\mu, ii}(\alpha)=\exp\big[-\frac{ \mu_i}{\widetilde{a} \alpha}  -\frac{\mu_i}{\widetilde{a}}\big(1-\frac{2\mu_i}{\widetilde{a}}\big)\log \alpha \big]$, $a_{2,\mu, ij}(\alpha)= \frac{\mu_j}{\mu_i + \mu_j}\exp\big[-\frac{\mu_i+\mu_j}{2 \widetilde{a}\alpha}- \frac{\mu_i+\mu_j}{2\widetilde{a}}\big(1-\frac{\mu_i+\mu_j}{\widetilde{a}}\big)\log \alpha\big]$ ($i\neq j$). Let $w_{ij}$ be the $i,j$-th element of $\mathbf{U}\widetilde{\mathbf{W}}\mathbf{U}^\dagger$, then  $b_1(\alpha)=\sum_i \text{Re}\big\{\frac{w_{ii}}{2 \mu_i} \exp\big[-\frac{\mu_i}{\widetilde{a} \alpha}  -\frac{\mu_i}{\widetilde{a}}\big(1-\frac{2\mu_i}{\widetilde{a}}\big)\log \alpha \big]\big\}+\sum_i\sum_{j>i}\text{Re}\big\{ \frac{w_{ij}}{\mu_i+\mu_j}\exp\big[-\frac{\mu_i+\mu_j}{2\widetilde{a}\alpha}-  \frac{\mu_i+\mu_j}{2\widetilde{a}}(1-\frac{\mu_i+\mu_j}{\widetilde{a}})  \log \alpha\big]\big\}$.
		 
Using  (\ref{appeox1}), the condition $\lambda >(V(\boldsymbol{\chi}) -V(\boldsymbol{\chi})\big|_{\boldsymbol{\Delta}=\mathbf{0}} ) \frac{  \alpha}{\kappa(R) B_W}$ can be written as $\frac{\boldsymbol{\Delta}^TA_{1,\mathbf{U}}(\alpha)\boldsymbol{\Delta}{\alpha}}{\kappa(R) B_W}\leq \lambda$. Based on (\ref{basedonequ}), $A_{1,\mathbf{U}}(\alpha)$ approaches to a constant matrix as $\alpha \rightarrow 0$. Hence,  the L.H.S. of the above equation decreases at least at the order of $\mathcal{O}(\|\boldsymbol{\Delta}\|^2\alpha)$ for sufficiently small $\|\boldsymbol{\Delta}\|^2\alpha$. Hence, the above condition  is satisfied for small  $\|\boldsymbol{\Delta}\|^2\alpha$.

\vspace{0.1cm}
\hspace{-0.3cm}\emph{B. Solution of (\ref{pdeintheapp}) for large $\|\boldsymbol{\Delta}\|^2\alpha$}   

We then solve the PDE in (\ref{pdeintheapp})   when $\lambda <(V(\boldsymbol{\chi}) -V(\boldsymbol{\chi})\big|_{\boldsymbol{\Delta}=\mathbf{0}} ) \frac{  \alpha}{\kappa(R) B_W}$ (i.e., $p=p_{max}$). We show later that large $\|\boldsymbol{\Delta}\|^2\alpha$ leads to this case. For this case, (\ref{pdeintheapp})  becomes  
\bs\begin{align}	
		& \theta= \boldsymbol{\Delta}^T \mathbf{S}\boldsymbol{\Delta} - \big[  \left(V(\boldsymbol{\chi})-V(\boldsymbol{\chi})\big|_{\boldsymbol{\Delta}=\mathbf{0}}   \right)\frac{  \alpha}{\kappa(R) B_W}-\lambda\big] p_{max}    \notag \\
		&+\nabla_{\boldsymbol{\Delta}}^TV(\boldsymbol{\chi}) \widetilde{\mathbf{F}} \boldsymbol{\Delta}  +\frac{1}{2}\text{Tr}\big( \nabla_{\boldsymbol{\Delta}}^2 V(\boldsymbol{\chi} )\widetilde{\mathbf{W}}\big)  +\frac{\partial V(\boldsymbol{\chi})}{\partial \alpha} \left(2 \widetilde{a} \alpha +2 \widetilde{a}  \right)  \notag \\
		&+  \frac{\partial^2V^\ast(\boldsymbol{\chi})}{\partial \alpha^2}4\widetilde{a}\alpha  + \text{Tr}\big(\frac{\partial V(\boldsymbol{\chi})}{\partial \boldsymbol{\Psi}}(\mathbf{H}\boldsymbol{\Psi}-\boldsymbol{\Psi})/\tau\big)\notag \\
	&+ \nabla^T_{\mathbf{L}}V(\boldsymbol{\chi})\big(\boldsymbol{\Gamma}\mathbf{F}_{\mathbf{R}}\mathbf{L}+w_{max}\|\mathbf{H}\boldsymbol{\Psi}\|\mathbf{1}-\mathbf{L}\big)/\tau
\label{pde1wrding1}
\end{align}\bsc This PDE is separable with solution $V(\boldsymbol{\chi})=\boldsymbol{\Delta}^TA_{2,\mathbf{U}}(\alpha)\boldsymbol{\Delta}+b_2(\alpha)$ for some $A_{2,\mathbf{U}}(\alpha)$ and $b_2(\alpha)$. Substituting this form into (\ref{pde1wrding1}), we obtain that 
\bs\begin{align}
	&\boldsymbol{\Delta}^T \big[\mathbf{S} - \frac{\alpha p_{max}}{\kappa(R) B_W} A_{2,\mathbf{U}}(\alpha)+ \left(2 \widetilde{a} \alpha + 2\widetilde{a}\right)A_{2,\mathbf{U}}'(\alpha) \label{59veryimp1}  \\
	&+ 4\widetilde{a}\alpha A_{2,\mathbf{U}}''(\alpha) +\left(A_{2,\mathbf{U}}(\alpha)+A_{2,\mathbf{U}}^T(\alpha)\right) \widetilde{\mathbf{F}}  \big] \boldsymbol{\Delta}  \notag \\ 
	&+ \big[\lambda p_{max}- \frac{\alpha p_{max}}{\kappa(R) B_W} \left(b_2(H)-b_1(H)\right) +\left(2 \widetilde{a} \alpha + 2\widetilde{a}\right)b_2'(\alpha)\notag \\
	& +4\widetilde{a}\alpha b_2''(\alpha) +\frac{1}{2}\text{Tr}\big( \left(A_{2,\mathbf{U}}(\alpha)+A_{2,\mathbf{U}}^T(\alpha) \right)  \widetilde{\mathbf{W}} \big)-\theta \big]=0	\notag
\end{align}\bsc Using the same approach for calculating $A_{1, \mathbf{U}}(\alpha)$ and denoting $A_{2,\mu}(\alpha)=(\mathbf{U}^{-1})^\dagger A_{2,\mathbf{U}}(\alpha) \mathbf{U}^{-1}$ and $\widetilde{c}\triangleq \sqrt{\widetilde{a}^2+4\widetilde{a}p_{max}/\left(\kappa\left(R\right)B_W\right)}$,  we  calculate $A_{2,\mu}(\alpha)=[a_{2,\mu, ij}]$  as follows: for $a_{2,\mu,ii}(\alpha) $,
\begin{align}
	a_{2,\mu,ii}(\alpha) = & \exp\Big[\frac{-\widetilde{a}+\widetilde{c}}{4\widetilde{a}}\alpha \label{62equ1} \\
	& + \Big(-\frac{1}{4}+\frac{\widetilde{a}-4 \mu_i}{4\widetilde{c}}\Big) \log(\alpha)+\frac{c_{ii}}{\alpha}+o\Big(\frac{1}{\alpha} \Big) \Big] \notag
\end{align}
as $\alpha\rightarrow \infty$, and $c_{ii}=\frac{\widetilde{a}^3+3\widetilde{a} \widetilde{c}^2+16 \widetilde{a}\widetilde{c}\mu_i-4\widetilde{a}(\widetilde{c}+2\mu)}{8\widetilde{c}^3}=\mathcal{O}\Big(\frac{1}{\sqrt{p_{max}}}\Big)$ as $p_{max}\rightarrow \infty$. For $a_{2,\mu,ij}(\alpha)$,
\begin{align}
	&a_{2,\mu,ij} =  \exp\Big[\frac{-\widetilde{a}+\widetilde{c}}{4\widetilde{a}}\alpha+\Big(-\frac{1}{4}+\frac{\widetilde{a}}{4 \widetilde{c}}  \label{56eq1} \\
	&  +L_f\Big(-\frac{\mu_j}{2\widetilde{c}}\exp\Big(-\frac{\mu_i+\mu_j}{2\widetilde{c}}\Big) \Big)\Big)\log(\alpha) +\frac{c_{ij}}{\alpha}+o\left(\frac{1}{\alpha} \right) \Big]       \notag 
\end{align}
as $\alpha\rightarrow \infty$, where  $L_f(\cdot)$ is the Lambert function \cite{lambert}, $c_{ij}=\mathcal{O}\Big(\frac{1}{\sqrt{p_{max}}}\Big)$. Using (\ref{62equ1}), (\ref{56eq1})  and the relationship $A_{2,\mu}(\alpha)=(\mathbf{U}^{-1})^\dagger A_{2,\mathbf{U}}(\alpha) \mathbf{U}^{-1}$, we can obtain $A_{2,\mathbf{U}}(\alpha)$.  Using a similar approach to calculating $b_1(\alpha)$, we can calculate $b_2(\alpha)$.  Based on the above calculations, we summarize the  overall solution of this case  below: 
	\begin{align}	
			&\log\left(V(\boldsymbol{\chi})\right)=\log (\boldsymbol{\Delta}^TA_{2,\mathbf{U}}(\alpha)\boldsymbol{\Delta}+ \exp\left(2B_2\right)\alpha^{C_2} )\notag \\
			& \hspace{1cm}-\frac{\widetilde{a}-\widetilde{c}}{4}\alpha -\Big(\frac{1}{4}-\frac{\widetilde{a}}{4\widetilde{c}}\Big)\log \alpha+\mathcal{O}\Big(\frac{1}{\alpha \sqrt{p_{max}}}\Big)	\label{appeox2}
	\end{align}
		where   $A_{2,\mathbf{U}}(\alpha)=\mathbf{U}^\dagger A_{2,\mu}(\alpha) \mathbf{U}\in \mathbb{R}^{d\times d}$, and $A_{2,\mu}(\alpha)=[a_{2,\mu, ij}(\alpha)]$ is a $d\times d$ constant matrix with $a_{2,\mu, ii}(\alpha)=\alpha^{\frac{- \mu_i}{\widetilde{c}}}$, $a_{2,\mu, ij}(\alpha)=\alpha^{L_f\big(-\frac{\mu_j}{2\widetilde{c}}\exp\big(-\frac{\mu_i+\mu_j}{2\widetilde{c}}\big) \big)}$ ($i\neq j$), and  $B_2=\sum_{i}\text{Re}\big\{\log \frac{w_{ii}}{4 \mu_i}\big\}+\sum_i\sum_{j>i}\text{Re}\big\{\log\big( {w_{ij}}/{\big(-4 \widetilde{c} L_f\big(\frac{\mu_j}{-2\widetilde{c}}\exp\big(\frac{\mu_i+\mu_j}{-2\widetilde{c}}\big)\big)}\big)\big)  +\log\big( {w_{ij}}/\big(-4 \widetilde{c}   L_f\big(\frac{\mu_i}{-2\widetilde{c}}   \exp\big(\frac{\mu_i+\mu_j}{-2\widetilde{c}}\big)\big)\big)\big)\big\}$,  $C_2=\sum_i \frac{2\text{Re}\{\mu_i\}}{-\widetilde{c}}+\sum_i\sum_{j>i}2\text{Re}\big\{L_f\big(\frac{\mu_i}{-2\widetilde{c}}\exp\big(\frac{\mu_i+\mu_j}{-2\widetilde{c}}\big)\big)+L_f\big(\frac{\mu_j}{-2\widetilde{c}}\exp\big(\frac{\mu_i+\mu_j}{-2\widetilde{c}}\big)\big)\big\}-\frac{d^2}{2}(1-\frac{\widetilde{a}}{\widetilde{c}})$.

Using  (\ref{appeox2}), the condition $\lambda <(V(\boldsymbol{\chi}) -V(\boldsymbol{\chi})\big|_{\boldsymbol{\Delta}=\mathbf{0}} ) \frac{  \alpha}{\kappa(R) B_W}$ can be written as $\big[\big(\boldsymbol{\Delta}^TA_{2,\mathbf{U}}(\alpha)\boldsymbol{\Delta}+\exp\left(2B_2\right)\alpha^{C_2} \big)\exp\big(\frac{\widetilde{c}-\widetilde{a}}{4}\alpha  -\big(\frac{1}{4}-\frac{\widetilde{a}}{4\widetilde{c}}\big) \log \alpha\big) -b_1(\alpha)\big]\frac{  \alpha}{\kappa(R) B_W}\geq \lambda$. Since   $V(\boldsymbol{\chi})$ increases  w.r.t. $\|\boldsymbol{\Delta}\|$ (due to Lemma 3  of \cite{shamai}) and $\frac{\widetilde{c}-\widetilde{a}}{4}>0$, the L.H.S. of the above equation increases at least at the order of  $\mathcal{O}(\|\boldsymbol{\Delta}\|^2\alpha)+ f(\alpha)$ for some function $f(\alpha)$ and sufficiently large  $\|\boldsymbol{\Delta}\|^2\alpha$. Hence, the above condition  is satisfied for large  $\|\boldsymbol{\Delta}\|^2\alpha$.

\vspace{-0.2cm}

\section*{\txtblue{Appendix G: Proof of the  Results in Remark \ref{remadksdsds}}}

\txtblue{\emph{1): Dynamic Threshold  w.r.t. State Estimation Error:}
Since $\widetilde{V}_{\eta_{th}}(\boldsymbol{\chi})$ increases  w.r.t. $\|\boldsymbol{\Delta}\|$  according to Lemma 3  of \cite{shamai}, for large $\|\boldsymbol{\Delta}\|$, we have $\big(\widetilde{V}_{\eta_{th}}(\boldsymbol{\chi}) -\widetilde{V}_{\eta_{th}}(\boldsymbol{\chi})\big|_{\boldsymbol{\Delta}=\mathbf{0}} \big) \frac{  \alpha}{\kappa(R) B_W}  =\boldsymbol{\Delta}^TA_{2,\mathbf{U}}(\alpha)\boldsymbol{\Delta}\exp\big(\frac{\widetilde{c}-\widetilde{a}}{4}\alpha -\big(\frac{1}{4}-\frac{\widetilde{a}}{4\widetilde{c}}\big)\log \alpha\big) + C(\alpha)$ for some function $C(\alpha)$. This expression grows at the order of  $\mathcal{O}\big(\|\boldsymbol{\Delta}\|^2\big)$ as $\|\boldsymbol{\Delta}\|^2$ increases for given  $\alpha$ and $R$.}

\txtblue{\emph{2): Dynamic Threshold  w.r.t. CSI:} According to the analysis in part B of Appendix F, for large $\alpha$, we have  $\big(\widetilde{V}_{\eta_{th}}(\boldsymbol{\chi}) -\widetilde{V}_{\eta_{th}}(\boldsymbol{\chi})\big|_{\boldsymbol{\Delta}=\mathbf{0}} \big) \frac{  \alpha}{\kappa(R) B_W}=\big[\big(\boldsymbol{\Delta}^TA_{2,\mathbf{U}}(\alpha)\boldsymbol{\Delta}+\exp\big(2B_2\big)\alpha^{C_2} \big)\exp\big(\frac{\widetilde{c}-\widetilde{a}}{4}\alpha -\big(\frac{1}{4}-\frac{\widetilde{a}}{4\widetilde{c}}\big)\log \alpha\big) -b_1(\alpha)\big]\frac{  \alpha}{\kappa(R) B_W}=\mathcal{O}(\alpha^{C_2+\frac{3}{4}+\frac{\widetilde{a}}{4\widetilde{c}}} \exp\big(\alpha\big) ) $.  This expression grows  at the order of $\mathcal{O}\big( \exp(\alpha)\big)$  as $\alpha$ increases for given  $\boldsymbol{\Delta}$ and $R$.}

\txtblue{\emph{3): Dynamic Threshold  w.r.t. Data Rate:}  Based on the above analysis, for given large  $\alpha$  or large  $\|\boldsymbol{\Delta}\|$,  we have  $\big(\widetilde{V}_{\eta_{th}}(\boldsymbol{\chi}) -\widetilde{V}_{\eta_{th}}(\boldsymbol{\chi})\big|_{\boldsymbol{\Delta}=\mathbf{0}} \big) \frac{  \alpha}{\kappa(R) B_W}=\big[\big(\boldsymbol{\Delta}^TA_{2,\mathbf{U}}(\alpha)\boldsymbol{\Delta}+\exp\big(2B_2\big)\alpha^{C_2} \big)\exp\big(\frac{\widetilde{c}-\widetilde{a}}{4}\alpha -\big(\frac{1}{4}-\frac{\widetilde{a}}{4\widetilde{c}}\big)\log \alpha\big) -b_1(\alpha)\big]\frac{  \alpha}{\kappa(R) B_W}$. For large $R$, we have $B_2=\mathcal{O}\big(\sqrt{\kappa\big(R\big)}\exp\big(\sqrt{\kappa\big(R\big)}\big) \big)$, $C_2=\mathcal{O}(1)$, $\frac{\widetilde{c}-\widetilde{a}}{4}\alpha -\big(\frac{1}{4}-\frac{\widetilde{a}}{4\widetilde{c}}\big)\log \alpha=\mathcal{O}(1)$. Therefore, the dynamic threshold grows  at the order of ${\sqrt{\kappa\big(R\big)}\exp\big(\sqrt{\kappa\big(R\big)}\big)}\big/{\kappa\big(R\big)}={\exp\big(\sqrt{\kappa\big(R\big)}\big)}\big/{\sqrt{\kappa\big(R\big)}} $ as $R$ increases for given large  $\alpha$  or large    $\boldsymbol{\Delta}$.}

\vspace{-0.2cm}

\section*{Appendix H: Proof of Theorem \ref{stab}}
\txtblue{First, we establish the following Lemma:
\begin{Lemma}		\label{sda2323s}
	If $\boldsymbol{\Delta}(t)$ is stable under $\widetilde{\Omega}_p^\ast$, then $(\Omega_{\mathbf{u}}^\ast, \widetilde{\Omega}_p^\ast)$ is admissible.
\end{Lemma}}
\begin{proof}
\txtblue{We verify the admissibility of $(\Omega_{\mathbf{u}}^\ast, \widetilde{\Omega}_p^\ast)$ according to  Definition \ref{admisscontrolpol} as follows:}

	\txtblue{\emph{1) Stability of $\mathbf{x}(t)$:} Taking expectation on condition of $I_C(t+1)$ on both sides of  (\ref{plantain}) and substituting $\mathbf{u}^\ast(t)$ in (\ref{statsrr}), 
\begin{align}
	&\hat{\mathbf{x}}(t+1) = \mathbb{E}\big[\mathbf{F} (\boldsymbol{\Delta}(t)+\hat{\mathbf{x}}(t))+\mathbf{G}\mathbf{u}^\ast(t)+\mathbf{w}(t)\big|I_C(t+1)\big]\notag \\
	&=\mathbf{F}\hat{\mathbf{x}}(t)+\mathbf{G}\mathbf{u}(t)+\hat{\mathbf{w}}(t)  = \big(\mathbf{F}-\mathbf{G}\mathbf{K}\big)\hat{\mathbf{x}}(t)+\hat{\mathbf{w}}(t)	\label{xhatdyn}
\end{align}
where $\hat{\mathbf{w}}(t)=\mathbb{E}\big[\mathbf{F} \boldsymbol{\Delta}(t)+\mathbf{w}(t)\big|I_C(t+1)\big]$ and according to Section III.B of \cite{rate2}, we have $\hat{\mathbf{W}}\triangleq \lim_{t\rightarrow \infty}\mathbb{E}[\hat{\mathbf{w}}(t)\hat{\mathbf{w}}^T(t)]=\mathbf{F}\big( \lim_{t\rightarrow \infty}\mathbb{E}[\boldsymbol{\Delta}(t)\boldsymbol{\Delta}^T(t)]\big)\mathbf{F}^T+\mathbf{W}- \lim_{t\rightarrow \infty}\mathbb{E}[\boldsymbol{\Delta}(t)\boldsymbol{\Delta}^T(t)]$. Therefore, if $\lim_{t\rightarrow \infty}\mathbb{E}[\boldsymbol{\Delta}(t)\boldsymbol{\Delta}^T(t)]<\infty$, we have $\|\hat{\mathbf{W}}\|<\infty$. Furthermore, from (\ref{xhatdyn}), for large $t$, we have 
\begin{align}
	\mathbb{E}\left[\| \hat{\mathbf{x}}(t+1)\|^2 \right]<\|\mathbf{F}-\mathbf{G}\mathbf{K}\|^2\mathbb{E}\left[\| \hat{\mathbf{x}}^T(t)\|^2 \right]+\|\hat{\mathbf{W}}\|
\end{align}
Since $\|\mathbf{F}-\mathbf{G}\mathbf{K}\|<1$ \cite{poweronl}, we have $\lim_{t \rightarrow \infty}\mathbb{E}\left[\| \hat{\mathbf{x}}(t)\|^2 \right]=\frac{\|\hat{\mathbf{W}}\|^2}{1-\|\mathbf{F}-\mathbf{G}\mathbf{K}\|^2}<\infty$.}

\txtblue{\emph{2) Stability of $\mathbf{L}(t)$:} Based on the proof of Lemma \ref{importapp},  if $\mathbf{x}(t)\in \mathcal{D}_{\widetilde{\mathbf{x}}}(t)$, then $\mathbf{x}(t)\in \mathcal{E}_{\hat{\mathbf{x}}}(t)$. Hence, $\mathcal{D}_{\widetilde{\mathbf{x}}}(t)\subset \mathcal{E}_{\hat{\mathbf{x}}}(t)$. Under the definition of $\mathcal{E}_{\hat{\mathbf{x}}}(t)$ in (\ref{exasdssdef}), if $\boldsymbol{\Delta}(t)$ is stable, we have $\lim_{t\rightarrow \infty}\mathbb{E}\left[\|\boldsymbol{\Psi}(t)\big(\mathbf{x}-\hat{\mathbf{x}}(t)\big)\|\right]<\infty$.  Under the definition of $\mathcal{D}_{\widetilde{\mathbf{x}}}(t)$ in (\ref{dregion}) and $\mathcal{D}_{\widetilde{\mathbf{x}}}(t)\subset \mathcal{E}_{\hat{\mathbf{x}}}(t)$, we have $\lim_{t\rightarrow \infty} \mathbb{E}\left[\|\boldsymbol{\Psi}(t)\big(\mathbf{x}-\widetilde{\mathbf{x}}(t)\big)\|\right]<\infty$. Therefore, we have $\lim_{t\rightarrow \infty}\mathbb{E}\left[\|\mathbf{L}(t)\|^2\right]<\infty$.}

\txtblue{Based on the above analysis, if  $\boldsymbol{\Delta}(t)$ is stable under $\widetilde{\Omega}_p^\ast$, then $(\Omega_{\mathbf{u}}^\ast, \widetilde{\Omega}_p^\ast)$ is admissible.}
\end{proof}

 In the following, we show that $\boldsymbol{\Delta}(t)$ is stable under  $\widetilde{\Omega}_p^\ast$.

\vspace{0.1cm}

\hspace{-0.3cm}\emph{A. Relationship between the State Drift and the Lyapunov Drift}   

Define a Lyapunov function $L(\boldsymbol{\Delta})=\|\boldsymbol{\Delta}\|^2$. Define the conditional state drift $d^{\widetilde{\Omega}_p^\ast}(\boldsymbol{\Delta})$  and the conditional Lyapunov drift $d^{\widetilde{\Omega}_p^\ast}L(\boldsymbol{\Delta})$ w.r.t $\boldsymbol{\Delta}$ as follows:
\begin{align}
	&d^{\widetilde{\Omega}_p^\ast}(\boldsymbol{\Delta})=\mathbb{E}^{\widetilde{\Omega}_p^\ast}\big[ \|\boldsymbol{\Delta}(t)\|-\|\boldsymbol{\Delta}(t-1)\|\big|\boldsymbol{\Delta}\left(t-1\right)=\boldsymbol{\Delta}\big]	\notag	 \\
	&d^{\widetilde{\Omega}_p^\ast}L(\boldsymbol{\Delta})=\mathbb{E}^{\widetilde{\Omega}_p^\ast}\big[L\left(\boldsymbol{\Delta}\left(t\right)\right)-L\left(\boldsymbol{\Delta}\left(t-1\right)\right)\big|\boldsymbol{\Delta}\left(t-1\right)=\boldsymbol{\Delta}\big]		 	\notag
\end{align}
For  sufficiently large $\|\boldsymbol{\Delta}\|$,
\bs\begin{align}
		&d^{\widetilde{\Omega}_p^\ast}L(\boldsymbol{\Delta})\notag \\
		\overset{(a)}{=}&\mathbb{E}^{\widetilde{\Omega}_p^\ast}\big[\left(\boldsymbol{\Delta}(t)-\boldsymbol{\Delta}(t-1)\right)^T \nabla_{\boldsymbol{\Delta}}L\left(\boldsymbol{\Delta}\right)\big|  \boldsymbol{\Delta}\left(t-1\right)=\boldsymbol{\Delta}\big] +\mathcal{O}(\tau^2)	\notag \\
		\overset{(b)}{\geq}&\mathbb{E}^{\widetilde{\Omega}_p^\ast}\big[\left(\boldsymbol{\Delta}(t)-\boldsymbol{\Delta}(t-1)\right)^T \nabla_{\boldsymbol{\Delta}}L\left(\boldsymbol{\Delta}\right)\big|\boldsymbol{\Delta}\left(t-1\right)=\boldsymbol{\Delta}\big]+C\tau^2	\notag \\
		{\geq} &\mathbb{E}^{\widetilde{\Omega}_p^\ast}\left[\|\boldsymbol{\Delta}(t)-\boldsymbol{\Delta}(t-1)\| \big|\boldsymbol{\Delta}\left(t-1\right)=\boldsymbol{\Delta}\right]\geq d^{\widetilde{\Omega}_p^\ast}(\boldsymbol{\Delta})\label{sda2323ssadad}
	\end{align}\bsc where (a) is due to the second order taylor expansion of $L\left(\boldsymbol{\Delta}(t)\right)$  and $\mathcal{O}(\|\boldsymbol{\Delta}(t)-\boldsymbol{\Delta}(t-1)\|^2)=\mathcal{O}(\tau^2)$ according to  (\ref{sub2}), where $(b)$ is because $L(\boldsymbol{\Delta})$  increase w.r.t. $\|\boldsymbol{\Delta}\|$.
% w.r.t. $\|\boldsymbol{\Delta}\|$ according to Lemma 3 in Section VI of \cite{shamai} and hence, $L(\boldsymbol{\chi})$ equals (\ref{appeox2}) for   sufficiently large $\|\boldsymbol{\Delta}\|$.

\vspace{0.1cm}
\hspace{-0.3cm}\emph{\txtblue{B. Negativity of  the Lyapunov  Drift for $\|\boldsymbol{\Delta}\|$ under $\widetilde{\Omega}_p^\ast$}}

\txtblue{We then establish the following lemma on $d^{\widetilde{\Omega}_p^\ast}L(\boldsymbol{\Delta})$.}
\begin{Lemma}	\label{realdrifttsds}
	\txtblue{If $\frac{ p_{max} \tau}{\kappa(R) B_W+p_{max} \tau} >\max\big\{ \frac{1}{R}\sum_{\mu_i(\mathbf{F})} \max\big\{0, \log|\mu_i(\mathbf{F})|\big\}, 1-\frac{1}{\mu_{max}(\mathbf{F}^T\mathbf{F})}\big\}$, we have $d^{\widetilde{\Omega}_p^\ast}L(\boldsymbol{\Delta})<0$ for sufficiently large $\|\boldsymbol{\Delta}\|$.}
\end{Lemma}
\begin{proof}
	\txtblue{Under the dynamics of $\boldsymbol{\Delta}(t)$ in (\ref{deltadyntext}),  we calculate $d^{\widetilde{\Omega}_p^\ast}L(\boldsymbol{\Delta})$ as follows:
	\bs\begin{align}
	&d^{\widetilde{\Omega}_p^\ast}L(\boldsymbol{\Delta})\notag \\
	=&\mathbb{E}^{\widetilde{\Omega}_p^\ast}\big[(1-\widetilde{p}(\boldsymbol{\chi}))\big(\mathbf{F}\boldsymbol{\Delta}+\mathbf{w}(t-1)\big)^T\big(\mathbf{F}\boldsymbol{\Delta}+\mathbf{w}(t-1)\big) \notag \\
	&+\widetilde{p}(\boldsymbol{\chi})  \mathbf{e}^T(\mathbf{L}(t), t)\boldsymbol{\Psi}^{-T}(t)\boldsymbol{\Psi}^{-1}(t)\mathbf{e}(\mathbf{L}(t), t) \notag \\
	&-\boldsymbol{\Delta}^T\boldsymbol{\Delta}\big|\boldsymbol{\Delta}\left(t-1\right)=\boldsymbol{\Delta}\big]\notag \\
	\overset{(c)}{<}&\mathbb{E}^{\widetilde{\Omega}_p^\ast}\big[(1-\widetilde{p}(\boldsymbol{\chi}))\boldsymbol{\Delta}^T\mathbf{F}^T\mathbf{F}\boldsymbol{\Delta}-\boldsymbol{\Delta}^T\boldsymbol{\Delta}+\text{Tr}(\mathbf{W}) \notag \\
	&+ C_1' \|\mathbf{L}(t)\|^2\big|\boldsymbol{\Delta}\left(t-1\right)=\boldsymbol{\Delta}\big]\notag \\
	\overset{(d)}{<}&\big[\big(1-\mathbb{E}^{\widetilde{\Omega}_p^\ast}\big[\widetilde{p}(\boldsymbol{\chi})\big|\boldsymbol{\Delta}\left(t-1\right)=\boldsymbol{\Delta}\big]\big)\mu_{max}(\mathbf{F}^T\mathbf{F})-1\big]\boldsymbol{\Delta}^T\boldsymbol{\Delta}\notag \\
	&+\text{Tr}(\mathbf{W})+C_1'\mathbb{E}^{\widetilde{\Omega}_p^\ast}\big[ \|\mathbf{L}(t)\|^2\big|\boldsymbol{\Delta}\left(t-1\right)=\boldsymbol{\Delta}\big]	\label{94equasdsd}
\end{align}\bsc where $\widetilde{p}(\boldsymbol{\chi})\triangleq \exp\big( -\frac{ p_{max} \tau \alpha}{\kappa(R) B_W} \big)\mathbf{1}\big(\lambda \leq \big(\widetilde{V}_{\eta_{th}}(\boldsymbol{\chi}) -\widetilde{V}_{\eta_{th}}(\boldsymbol{\chi})\big|_{\boldsymbol{\Delta}=\mathbf{0}} \big) \frac{  \alpha}{\kappa(R) B_W}\big)$ is the SER for given $\boldsymbol{\chi}$ and $C_1'>0$ is a constant,  (c) is because the cross-product terms w.r.t. of the L.H.S. equals to zero and  (\ref{51equation}), and (d) is because $\boldsymbol{\Delta}^T\mathbf{F}^T\mathbf{F}\boldsymbol{\Delta}<\mu_{max}(\mathbf{F}^T\mathbf{F})\boldsymbol{\Delta}^T\boldsymbol{\Delta}$.} 

\txtblue{In the following, we show that for sufficiently large $\|\boldsymbol{\Delta}\|$, if $\frac{ p_{max} \tau}{\kappa(R) B_W+p_{max} \tau}>1-\frac{1}{\mu_{max}(\mathbf{F}^T\mathbf{F})}$,  $(1-\mathbb{E}^{\widetilde{\Omega}_p^\ast}[\widetilde{p}(\boldsymbol{\chi})|\boldsymbol{\Delta}(t-1)=\boldsymbol{\Delta}])\mu_{max}(\mathbf{F}^T\mathbf{F})-1<0$ and $\mathbb{E}^{\widetilde{\Omega}_p^\ast}[ \|\mathbf{L}(t)\|^2|\boldsymbol{\Delta}\left(t-1\right)=\boldsymbol{\Delta}]$ is bounded as $t\rightarrow \infty$, so that the R.H.S. of (\ref{94equasdsd}) is negative for sufficiently large $\|\boldsymbol{\Delta}\|$.}

\txtblue{\emph{1) $(1-\mathbb{E}^{\widetilde{\Omega}_p^\ast}[\widetilde{p}(\boldsymbol{\chi})|\boldsymbol{\Delta}\left(t-1\right)=\boldsymbol{\Delta}])\mu_{max}(\mathbf{F}^T\mathbf{F})-1<0$ for sufficiently large $\|\boldsymbol{\Delta}\|$:} We calculate $\mathbb{E}^{\widetilde{\Omega}_p^\ast}\big[\widetilde{p}(\boldsymbol{\chi})\big|\boldsymbol{\Delta}\left(t-1\right)=\boldsymbol{\Delta}\big]$  for sufficiently large $\|\boldsymbol{\Delta}\|$ as follows:
\bs\begin{align}
	&\mathbb{E}^{\widetilde{\Omega}_p^\ast}\big[\widetilde{p}(\boldsymbol{\chi})\big|\boldsymbol{\Delta}\left(t-1\right)=\boldsymbol{\Delta}\big]\label{95sdsas} \\
	\overset{(e)}{=}&\mathbb{E}^{\widetilde{\Omega}_p^\ast}\big[\mathbb{E}^{\widetilde{\Omega}_p^\ast}\big[\exp\big( -\frac{ p_{max} \tau \alpha}{\kappa(R) B_W} \big)\big|\boldsymbol{\Delta}\left(t-1\right)=\boldsymbol{\Delta}, \alpha(t-1)=\alpha\big]\big] \notag \\
	\overset{(f)}{=}&\int_0^\infty \exp\big( -\frac{ p_{max} \tau x}{\kappa(R) B_W} \big)\exp(-x) dx = \frac{ p_{max} \tau}{\kappa(R) B_W+p_{max} \tau}	\notag
\end{align}\bsc where (e) is because $\widetilde{V}_{\eta_{th}}(\boldsymbol{\chi})$ increases  w.r.t. $\|\boldsymbol{\Delta}\|$ according to Lemma 3  of \cite{shamai} and hence $\mathbf{1}\big(\lambda \leq \big(\widetilde{V}_{\eta_{th}}(\boldsymbol{\chi}) -\widetilde{V}_{\eta_{th}}(\boldsymbol{\chi})\big|_{\boldsymbol{\Delta}=\mathbf{0}} \big) \frac{  \alpha}{\kappa(R) B_W}\big)=1$ for given $\alpha(t-1)=\alpha$ and sufficiently large $\|\boldsymbol{\Delta}\|$, and (f) is because   $h(t)$ in (\ref{chnain}) has a steady state  distribution $\mathcal{CN}\left(0, 1\right)$ according to Prop.1 of \cite{stochasticgame} and hence $|h(t)|^2\sim \exp(1)$. Therefore, based on (\ref{95sdsas}),  if $\frac{ p_{max} \tau}{\kappa(R) B_W+p_{max} \tau}>1-\frac{1}{\mu_{max}(\mathbf{F}^T\mathbf{F})}$,  $(1-\mathbb{E}^{\widetilde{\Omega}_p^\ast}[\widetilde{p}(\boldsymbol{\chi})|\boldsymbol{\Delta}(t-1)=\boldsymbol{\Delta}])\mu_{max}(\mathbf{F}^T\mathbf{F})-1<0$ for sufficiently large $\|\boldsymbol{\Delta}\|$.}

\txtblue{\emph{2) $\mathbb{E}^{\widetilde{\Omega}_p^\ast}\big[ \|\mathbf{L}(t)\|^2\big|\boldsymbol{\Delta}\left(t-1\right)=\boldsymbol{\Delta}\big]$ is bounded as $t\rightarrow \infty$ for sufficiently large $\|\boldsymbol{\Delta}\|$:} We write the dynamics of $\mathbf{L}(t)$  in (\ref{28equa}) in  the following form:
\begin{align} \label{93equationsss}
	\mathbf{L}(t) = \boldsymbol{\Gamma}\widetilde{\mathbf{F}}_{\mathbf{R}}(t)\mathbf{L}(t-1)+w_{max}\|\boldsymbol{\Psi}(t)\|\mathbf{1}
\end{align}
where $\widetilde{\mathbf{F}}_{\mathbf{R}}(t)=\mathbf{I}$ if $\gamma(t-1)=0$ and $\widetilde{\mathbf{F}}_{\mathbf{R}}(t)=\mathbf{F}_{\mathbf{R}}\mathbf{I}$ if $\gamma(t-1)=1$. Following (\ref{93equationsss}) and  starting from $\mathbf{L}(0)$, we have
\bs\begin{align}
	&\mathbb{E}^{\widetilde{\Omega}_p^\ast}\left[\|\mathbf{L}(t)\|^2\right]=\mathbb{E}^{\widetilde{\Omega}_p^\ast}\big[\big\| \prod_{i=0}^{t-1} (\boldsymbol{\Gamma}\widetilde{\mathbf{F}}_{\mathbf{R}}(i)) \mathbf{L}(0) \label{94equationsd} \\
	&+\sum_{j=0}^t \Big(\prod_{i=j}^{t-1} (\boldsymbol{\Gamma}\widetilde{\mathbf{F}}_{\mathbf{R}}(i))\Big) w_{max}\|\boldsymbol{\Psi}(j)\|\mathbf{1}\big\|^2\Big| \boldsymbol{\Delta}\left(t-1\right)=\boldsymbol{\Delta}\big] \notag  \\
	&\overset{(g)}{\leq}  \mathbb{E}^{\widetilde{\Omega}_p^\ast}\big[\big\|\prod_{i=0}^{t-1} (\boldsymbol{\Gamma}\widetilde{\mathbf{F}}_{\mathbf{R}}(i))\big\|^2 \|\mathbf{L}(0)\|^2 \notag \\
	& +C_2'\sum_{j=0}^t \big\|\prod_{i=j}^{t-1} (\boldsymbol{\Gamma}\widetilde{\mathbf{F}}_{\mathbf{R}}(i))\big\|^2\Big| \boldsymbol{\Delta}\left(t-1\right)=\boldsymbol{\Delta}\big] 	\notag \\
	& \overset{(h)}{\leq}  \mathbb{E}^{\widetilde{\Omega}_p^\ast}\Big\{\mathbb{E}^{\widetilde{\Omega}_p^\ast}\Big[\big\|\prod_{i=0}^{t_0} (\boldsymbol{\Gamma}\widetilde{\mathbf{F}}_{\mathbf{R}}(i))\big\|^2 \|\mathbf{L}(0)\|^2+C_2'\sum_{j=0}^t \big\|\prod_{i=j}^{t_0} (\boldsymbol{\Gamma}\widetilde{\mathbf{F}}_{\mathbf{R}}(i))\big\|^2 \notag \\
	&+\big\|\prod_{i=t_0}^{t-1} (\boldsymbol{\Gamma}\widetilde{\mathbf{F}}_{\mathbf{R}}(i))\big\|^2 \|\mathbf{L}(0)\|^2 +C_2'\sum_{j=0}^t \big\|\prod_{i=t_0}^{t-1} (\boldsymbol{\Gamma}\widetilde{\mathbf{F}}_{\mathbf{R}}(i))\big\|^2\Big|	\notag \\
	&  \{\boldsymbol{\Delta}\left(s\right):t_0\leq s \leq t-1\}\Big]\Big|\boldsymbol{\Delta}\left(t-1\right)=\boldsymbol{\Delta}\Big\}\notag
\end{align}\bsc
where $C_2'>0$ is a constant, (g)  is because $\|\boldsymbol{\Psi}(j)\|$ is uniformly bounded according to Prop. 5.2 of \cite{rate1}, and in (h) $t_0$ is chosen such that $\{\|\boldsymbol{\Delta}(s)\|:t_0\leq s \leq t-1\}$ are sufficiently large. The existence of such $t_0$ is ensured because of the dynamics of  $\boldsymbol{\Delta}(t)$ in (\ref{deltadyntext}) and $\|\boldsymbol{\Delta}(t-1)\|$ is sufficiently large. Furthermore, in (\ref{94equationsd}), we have $\mathbb{E}^{\widetilde{\Omega}_p^\ast}[\|\prod_{i=0}^{t_0} (\boldsymbol{\Gamma}\widetilde{\mathbf{F}}_{\mathbf{R}}(i))\|^2 \|\mathbf{L}(0)\|^2+C_2'\sum_{j=0}^t \|\prod_{i=j}^{t_0} (\boldsymbol{\Gamma}\widetilde{\mathbf{F}}_{\mathbf{R}}(i))\|^2 | \{\boldsymbol{\Delta}(s):t_0\leq s \leq t-1\}]<\infty$. For sufficiently large $\{\|\boldsymbol{\Delta}(s)\|:t_0\leq s \leq t-1\}$, using (e) of (\ref{95sdsas}), we have that $\mathbb{E}^{\widetilde{\Omega}_p^\ast}[\widetilde{p}(\boldsymbol{\chi}(s))|\{\boldsymbol{\Delta}(s):t_0\leq s \leq t-1\}]=\frac{ p_{max} \tau}{\kappa(R) B_W+p_{max} \tau}$ for all $t_0\leq s \leq t-1$. Using Prop. 4.4 of \cite{ratenoise} and Theorem 1 of \cite{stabprodthem}, under the condition that 
\begin{align}
	\frac{ p_{max} \tau}{\kappa(R) B_W+p_{max} \tau}R_i>\max\left\{0, \log|\mu_i(\mathbf{F})|\right\}, \quad \forall i	\label{pexpreetiel}
\end{align}  
we have $\frac{1}{t-t_0-1}\log \big\|\prod_{i=t_0}^{t-1} (\boldsymbol{\Gamma}\widetilde{\mathbf{F}}_{\mathbf{R}}(i))\big\|<0$, $\frac{1}{t-t_0-1}\log\big\|\prod_{i=t_0}^{t-1} (\boldsymbol{\Gamma}\widetilde{\mathbf{F}}_{\mathbf{R}}(i))\big\| <0$, $\forall j$, as $t \rightarrow \infty$. Therefore, from (\ref{94equationsd}), there exists $\delta>0$ and $\delta_j>0$ for all $j$ such that as $t\rightarrow\infty$,
\bs\begin{align}
	& \mathbb{E}^{\widetilde{\Omega}_p^\ast}\Big[\big\|\prod_{i=t_0}^{t-1} (\boldsymbol{\Gamma}\widetilde{\mathbf{F}}_{\mathbf{R}}(i))\big\|^2 \|\mathbf{L}(0)\|^2  +C_2'\sum_{j=0}^t \big\|\prod_{i=t_0}^{t-1} (\boldsymbol{\Gamma}\widetilde{\mathbf{F}}_{\mathbf{R}}(i))\big\|^2 \notag \\
	&\big| \{\boldsymbol{\Delta}\left(s\right):t_0\leq s \leq t-1\}\Big]{\leq} \mathbb{E}^{\widetilde{\Omega}_p^\ast}\big[\exp\left(-2\delta(t-t_0-1)\right)\|\mathbf{L}(0)\|\notag \\
	&+C_2'\sum_{j=0}^t \exp\left(-2\delta_j(t-t_0-1) \right)\big]<\infty	\label{97sdasd1}
\end{align}\bsc Substituting (\ref{97sdasd1}) into (\ref{94equationsd}), we have that $\mathbb{E}^{\widetilde{\Omega}_p^\ast}\big[ \|\mathbf{L}(t)\|^2\big|\boldsymbol{\Delta}\left(t-1\right)=\boldsymbol{\Delta}\big]$ is bounded as $t\rightarrow \infty$ for sufficiently large $\|\boldsymbol{\Delta}\|$.}

\txtblue{Combining the results in 1) and 2), if the conditions in (\ref{sadsaxxxx}) are satisfied and we allocate the $\{R_i\}$ according to (\ref{pexpreetiel}), we have $d^{\widetilde{\Omega}_p^\ast}L(\boldsymbol{\Delta})<0$ for sufficiently large $\|\boldsymbol{\Delta}\|$.}\end{proof}

\vspace{0.1cm}
\hspace{-0.3cm}\emph{C. Stability of $\boldsymbol{\Delta}(t)$ under $\Omega_p^\ast$}   

\txtblue{Define the \emph{semi-invariant moment generating function} of $d^{\widetilde{\Omega}_p^\ast}(\boldsymbol{\Delta})$ as 
\begin{align}
	\phi(r,\boldsymbol{\Delta} )=  \ln \big(\mathbb{E}^{\widetilde{\Omega}_p^\ast}\big[e^{\left(\|\boldsymbol{\Delta}(t)\|-\|\boldsymbol{\Delta}(t-1)\| \right)r}\big| \boldsymbol{\Delta}\left(t-1\right)=\boldsymbol{\Delta}  \big] \big)\notag
\end{align}
According to  (\ref{sda2323ssadad}) and Lemma \ref{realdrifttsds}, we have $\mathbb{E}^{\widetilde{\Omega}_p^\ast}\big[\|\boldsymbol{\Delta}(t)\|-\|\boldsymbol{\Delta}(t-1)\| \big| \boldsymbol{\Delta}\left(t-1\right)\big]=d^{\widetilde{\Omega}_p^\ast}(\boldsymbol{\Delta}) < 0$ when $\|\boldsymbol{\Delta}\| > \|\widetilde{\boldsymbol{\Delta}}\|$ for some large $\|\widetilde{\boldsymbol{\Delta}}\|>0$. Hence,  $\phi(r,\boldsymbol{\Delta} )$ will have a unique positive root  $r^\ast(\boldsymbol{\Delta} )$ (i.e., $\phi_k(r_k^\ast(\boldsymbol{\Delta}),\boldsymbol{\Delta} )=0$) \cite{kingman1}. Using Kingman bound \cite{kingman1}, we have  $\Pr\big[ \|\boldsymbol{\Delta}\| \geq x\big] \leq e^{-r^\ast(\widetilde{\boldsymbol{\Delta}}) x} $, if $x \geq \|\widetilde{\boldsymbol{\Delta}}\|$ for sufficiently large $\|\widetilde{\boldsymbol{\Delta}}\|$. Therefore,
\bs\begin{align}
	 & \mathbb{E}^{\widetilde{\Omega}_p^\ast} \left[\|\boldsymbol{\Delta}\|^2 \right] = \int_0^{\infty} \Pr \left[\|\boldsymbol{\Delta}\|^2 >s \right] \mathrm{d}s  	\notag \\
	 &= \int_0^{\|\widetilde{\boldsymbol{\Delta}}\|} \Pr \left[\|\boldsymbol{\Delta}\|^2 >s \right] \mathrm{d}s  + \int_{\|\widetilde{\boldsymbol{\Delta}}\|}^\infty \Pr \left[\|\boldsymbol{\Delta}\|^2 >s \right] \mathrm{d}s  \notag \\
	 &\leq  \|\widetilde{\boldsymbol{\Delta}}\| + \int_{\|\widetilde{\boldsymbol{\Delta}}\|}^{\infty}  e^{-r_k^\ast(\|\widetilde{\boldsymbol{\Delta}}\|) s}  \mathrm{d}s 	 < \infty	\notag 
\end{align}\bsc Using Lemma \ref{sda2323s}, $\widetilde{\Omega}_p^\ast$ is an admissible control policy.}
\vspace{-0.2cm}

\section*{\txtblue{Appendix I: Proof of Theorem \ref{stabasdsa}}}

\txtblue{From the dynamics  of $\boldsymbol{\Delta}(t)$ in (\ref{deltadyntext}) and  (\ref{94equasdsd}), we have
\begin{align}
	&\mathbb{E}^{\widetilde{\Omega}_p^\ast} \big[\|\boldsymbol{\Delta}(t)\|^2\big]   \\
	\overset{(a)}{\geq} &\mathbb{E}^{\widetilde{\Omega}_p^\ast}\big[\lambda_{min}(\boldsymbol{\Phi}^{-T}\boldsymbol{\Phi}^{-1})(1-\widetilde{p} (\boldsymbol{\chi}))\notag \\
	& \cdot \boldsymbol{\Delta}^T(t-1)\boldsymbol{\Phi}^T\boldsymbol{\Upsilon}^T \boldsymbol{\Upsilon}\boldsymbol{\Phi}\boldsymbol{\Delta}(t-1)+ C_1' \|\mathbf{L}(t)\|^2\big]\notag \\
	\overset{(b)}{\geq} &\mathbb{E}^{\widetilde{\Omega}_p^\ast}\big[\lambda_{min}(\boldsymbol{\Phi}^{-T}\boldsymbol{\Phi}^{-1})(1-\frac{ p_{max} \tau}{\kappa(R) B_W+p_{max} \tau})\notag \\
	&\cdot \boldsymbol{\Delta}^T(t-1)\boldsymbol{\Phi}^T\boldsymbol{\Upsilon}^T \boldsymbol{\Upsilon}\boldsymbol{\Phi}\boldsymbol{\Delta}(t-1)+ C_1' \|\mathbf{L}(t)\|^2\big]\notag
\end{align}
where  (a) is because $\boldsymbol{\Phi}\mathbf{F}\boldsymbol{\Phi}^{-1}=\boldsymbol{\Upsilon}$  in Part A of Appendix A and $\mathbf{F}^T\mathbf{F}\geq \lambda_{min}(\boldsymbol{\Phi}^{-T}\boldsymbol{\Phi}^{-1})\boldsymbol{\Phi}^T\boldsymbol{\Upsilon}^T \boldsymbol{\Upsilon}\boldsymbol{\Phi}$, and (b) is because  $\widetilde{p}(\boldsymbol{\chi}) \leq \exp\big( -\frac{ p_{max} \tau \alpha}{\kappa(R) B_W} \big)$ and the calculations in (\ref{95sdsas}). Following the above iterations, we have
\bs\begin{align}
	&\mathbb{E}^{\widetilde{\Omega}_p^\ast} \big[\|\boldsymbol{\Delta}(t)\|^2\big]  \geq \mathbb{E}^{\widetilde{\Omega}_p^\ast}\big[\lambda_{min}(\boldsymbol{\Phi}^{-T}\boldsymbol{\Phi}^{-1})\label{1032ddsds} \\
	& \cdot \big(1-\frac{ p_{max} \tau}{\kappa(R) B_W+p_{max} \tau}\big)^t \boldsymbol{\Delta}^T(0)\boldsymbol{\Phi}^T(\boldsymbol{\Upsilon}^t)^T \boldsymbol{\Upsilon}^t\boldsymbol{\Phi}\boldsymbol{\Delta}(0)\big]\notag
\end{align}\bsc Note that the diagonal elements of $(\boldsymbol{\Upsilon}^t)^T \boldsymbol{\Upsilon}^t$ are $\{\mu_i(\mathbf{F}^T\mathbf{F})^{t}\}$. Since $\lim_{t\rightarrow \infty} \mathbb{E}^{\widetilde{\Omega}_p^\ast} \big[\|\boldsymbol{\Delta}(t)\|^2\big]<\infty$,  the R.H.S. of (\ref{1032ddsds}) is bounded as $t\rightarrow \infty$. Therefore, we have  $\big(1-\frac{ p_{max} \tau}{\kappa(R) B_W+p_{max} \tau}\big)\mu_i(\mathbf{F}^T\mathbf{F})<1$ for all $i$. This implies 
\bs\begin{align}	\label{necss1}
		\frac{ p_{max} \tau}{\kappa(R) B_W+p_{max} \tau} >1-\frac{1}{\mu_{max}(\mathbf{F}^T\mathbf{F})}
\end{align}\bsc}
\txtblue{\ From the dynamics  of $\mathbf{L}(t)$ in (\ref{93equationsss}) and (\ref{94equationsd}), we have
\begin{align}
	&\mathbb{E}^{\widetilde{\Omega}_p^\ast}\left[\|\mathbf{L}(t)\|^2\right]\geq \mathbb{E}^{\widetilde{\Omega}_p^\ast}\Big[ \prod_{i=0}^{t-1}\big\| (\boldsymbol{\Gamma}\widetilde{\mathbf{F}}_{\mathbf{R}}(i))\big\|^2\Big]\| \mathbf{L}(0)\|^2\notag \\
	\geq &\mathbb{E}^{\widetilde{\Omega}_p^\ast}\Big[ \exp\big(2t \big[\frac{1}{t}\sum_{i=0}^{t-1} \log  \big\| (\boldsymbol{\Gamma}\widetilde{\mathbf{F}}_{\mathbf{R}}(i))\big\|\big]\big)\Big]\| \mathbf{L}(0)\|^2\label{1032ddsds1}
\end{align}
Since $\lim_{t\rightarrow \infty} \mathbb{E}^{\widetilde{\Omega}_p^\ast} \big[\|\mathbf{L}(t)\|^2\big]<\infty$,  the R.H.S. of (\ref{1032ddsds1}) is bounded as $t\rightarrow \infty$.  This implies 
\bs\begin{align}
	\lim_{t\rightarrow \infty}\frac{1}{t}\sum_{i=0}^{t-1} \log  \big\| (\boldsymbol{\Gamma}\widetilde{\mathbf{F}}_{\mathbf{R}}(i))\big\|=\mathbb{E}^{\widetilde{\Omega}_p^\ast}\big[\log  \big\| (\boldsymbol{\Gamma}\widetilde{\mathbf{F}}_{\mathbf{R}}(i))\big\|\big]<0	 \notag 
\end{align}\bsc Furthermore,  using Prop. 4.4 of \cite{ratenoise} and Theorem 1 of \cite{stabprodthem},  the above equation implies $\mathbb{E}^{\widetilde{\Omega}_p^\ast}\left[\widetilde{p}(\boldsymbol{\chi})\right] R_i>\max\left\{0, \log|\mu_i(\mathbf{F})|\right\}$,  $\forall i$.  Since $\mathbb{E}^{\widetilde{\Omega}_p^\ast}\left[\widetilde{p}(\boldsymbol{\chi})\right] \leq \frac{ p_{max} \tau}{\kappa(R) B_W+p_{max} \tau} $, we  further have 
\bs\begin{align}	\label{necss2}
	\frac{ p_{max} \tau}{\kappa(R) B_W+p_{max} \tau} > \frac{1}{R}\sum_{\mu_i(\mathbf{F})} \max\left\{0, \log|\mu_i(\mathbf{F})|\right\}
\end{align}\bsc}
\txtblue{\ Combining (\ref{necss1}) and (\ref{necss2}), we obtain the necessary conditions for NCS stability  in Theorem \ref{stabasdsa}.}

\vspace{-0.2cm}

\section*{\txtblue{Appendix J: Proof of the Results in Remark  \ref{rqads11dsdds}}}

\txtblue{\emph{1) Boundness of $\|\overline{\mathbf{e}}(t)\|$: } Based on the dynamics of $\mathbf{x}(t)$ and $\overline{\mathbf{x}}(t)$, we have $\overline{\mathbf{e}}(t)=\big(\mathbf{I}-\mathbf{M}\mathbf{C}\big)\big(\mathbf{F}\overline{\mathbf{e}}(t-1)+\mathbf{w}(t-1)\big)$. Since $\mathbf{I}-\mathbf{M}\mathbf{C}$ is stable, $\|\overline{\mathbf{e}}(0)\|$ is bounded (because $\mathbf{x}(0)$ is bounded), and $\|\mathbf{w}(t)\|$ is bounded,  it can be directly shown that $\|\overline{\mathbf{e}}(t)\|$ is bounded, i.e., $\|\overline{\mathbf{e}}(t)\|\leq C$ for some constant $C$.}

\txtblue{\emph{2) Upper bound of $\boldsymbol{\Delta}^T(t)\mathbf{S} \boldsymbol{\Delta}(t) $:} We write $\boldsymbol{\Delta}(t)=\mathbf{x}(t)-\overline{\mathbf{x}}(t)+\overline{\mathbf{x}}(t)-\mathbb{E}[\overline{\mathbf{x}}(t)|I_C(t)] +\mathbb{E}[(\overline{\mathbf{x}}(t)-\mathbf{x}(t))|I_C(t)]$ and derive an upper of $\boldsymbol{\Delta}^T(t)\mathbf{S} \boldsymbol{\Delta}(t) $ as follows:
\bs\begin{align}
	&\boldsymbol{\Delta}^T(t)\mathbf{S} \boldsymbol{\Delta}(t) \leq  2\mu_{max}(\mathbf{S}) \big(\mathbf{x}(t)-\overline{\mathbf{x}}(t)\big)^2 \notag \\
	 & + 2 \mu_{max}(\mathbf{S})\big(\mathbb{E}[(\overline{\mathbf{x}}(t)-\mathbf{x}(t))|I_C(t)]\big)^2 \notag \\
	 & + 2 \big(\overline{\mathbf{x}}(t)-\mathbb{E}[\overline{\mathbf{x}}(t)|I_C(t)]\big)^T \mathbf{S} \big(\overline{\mathbf{x}}(t)-\mathbb{E}[\overline{\mathbf{x}}(t)|I_C(t)]\big)\notag \\
	  \leq & 4 C^2 \mu_{max}(\mathbf{S})  + 2 \overline{\boldsymbol{\Delta}}^T(t)\mathbf{S} \overline{\boldsymbol{\Delta}}(t)
\end{align}\bsc where the last inequality is because $\|\mathbf{x}(t)-\overline{\mathbf{x}}(t)\|\leq C$.}

\end{document}